\documentclass[11pt,letterpaper,twoside,final,thmsb,notitlepage]{article}
\usepackage[latin9]{inputenc}
\usepackage{fancyhdr}
\usepackage{float}
\usepackage{amsmath}
\usepackage{amsthm}
\usepackage{amssymb}
\usepackage{graphicx}
\usepackage{setspace}
\usepackage[unicode=true,
 bookmarks=true,bookmarksnumbered=false,bookmarksopen=true,bookmarksopenlevel=1,
 breaklinks=false,pdfborder={0 0 1},backref=section,colorlinks=false]
 {hyperref}
\hypersetup{
 pdftex,pdfstartview=Fit}
\usepackage{breakurl}

\makeatletter

\theoremstyle{plain}
\newtheorem{lem}{\protect\lemmaname}
\theoremstyle{plain}
\newtheorem{prop}{\protect\propositionname}
\theoremstyle{plain}
\newtheorem{cor}{\protect\corollaryname}
\theoremstyle{definition}
 \newtheorem{example}{\protect\examplename}

\usepackage{amsfonts}
\usepackage{latexsym}
\usepackage{color}
\usepackage{fancyhdr}
\usepackage{amsthm}
\providecommand{\U}[1]{\protect \rule{.1in}{.1in}}

\makeatletter
\let\old@maketitle\maketitle
\renewcommand{\maketitle}{%
  \old@maketitle
  \begin{center}
    \vspace{-0.5em}   
    {\normalsize\itshape First Draft}
    \vspace{1.2em}    
  \end{center}
}
\makeatother

\newcommand{\notinanymore}[1]{}

\makeatother

\providecommand{\corollaryname}{Corollary}
\providecommand{\examplename}{Example}
\providecommand{\lemmaname}{Lemma}
\providecommand{\propositionname}{Proposition}

\begin{document}

\title{\textbf{Competitive Credit and Present Bias: A Stochastic Discounting
Approach}\thanks{For useful discussions, the authors are grateful to Daniele Condorelli,
Albin Erlanson, Marco Mariotti, and Arunava Sen. Chatterjee and Garrett
are grateful for support from the UKRI Frontier Research grant with
grant number EP/Z001528/1. }}

\author{\textbf{Siddharth Chatterjee}\thanks{Chatterjee: University of Essex, 123sidch@gmail.com.}\textbf{
\ and\ Daniel F. Garrett}\thanks{Garrett: University of Essex, d.garrett@essex.ac.uk.}}

\date{This draft: February, 2026}
\maketitle
\begin{abstract}
A prominent theme in behavioural contract theory is the study of present-biased
agents represented through quasi-hyperbolic discounting. In a model
of competitive credit provision, we study an alternative to this framework
in which the agent has a private stochastic discount factor and may
overestimate the likelihood of more patient values. Agent preferences,
however, are time-consistent. While a limiting case of our model corresponds
to a ``fully naive'' agent in work on quasi-hyperbolic discounting,
another case is where the agent has correct beliefs about future discounting.
In equilibrium, the agent selects options with earlier consumption
in case of less patient discount factor realisations, but is penalised
by receiving worse terms. Our model thus accounts for an important
feature of equilibrium contracts identified in Heidhues and K\H{o}szegi
(2010). Unlike Heidhues and K\H{o}szegi, our framework often predicts
excessively backloaded consumption, including when the agent holds
correct beliefs about future discounting.


\emph{JEL classification}: D82, D86

\emph{Keywords}: \ credit, lending, quasi-hyperbolic discounting,
stochastic discount factor, dynamic mechanism design
\end{abstract}
\bigskip{}

\thispagestyle{empty}
\pagebreak{}

\setcounter{page}{1}

\section{Introduction\label{Sec-Intro}}

A growing literature in behavioural contract theory studies present-biased
agents, with agents often either naive or partially naive about their
future behaviour.\footnote{Papers that include the study of such naive or partially naive agents
include O'Donoghue and Rabin (1999), DellaVigna and Malmendier (2004),
Gottlieb (2008), Heidhues and K\H{o}szegi (2010), Murooka and Schwarz
(2018), Bubb and Warren (2020), Gottlieb and Zhang (2021), Citanna
and Siconolfi (2022), Heidhues et al. (2024), Sulka (2024), and Englmaier
et al. (2026). } Present bias is modeled using quasi-hyperbolic discounting following
Laibson (1997) (who draws in turn on Phelps and Pollak, 1968). Agents
then have time-inconsistent preferences: an agent is a collection
of distinct ``selves'' which favour present over future utility. Naivety
means that agents mispredict the preferences of future selves. In
effect, they overestimate future patience.

In this paper, we suggest an alternative though complementary approach.
We propose agents with time-consistent preferences, but where discounting
is stochastic. While agents anticipate stochastic discount factors,
they may hold misspecified beliefs and overestimate future patience.
Behaviour is dynamically consistent in the sense that agents plan
contingent on future discounting realisations and follow through on
these plans. Nonetheless, a fully naive present-biased agent as studied
in existing work on quasi-hyperbolic discounting arises as a particular
limiting case in which beliefs over discounting are degenerate. Another
important case is an agent with correct, non-degenerate beliefs. We
apply this framework to competitive credit markets, building on Heidhues
and K\H{o}szegi (2010, henceforth HK) and Gottlieb and Zhang (2021,
henceforth GZ).

We follow closely the original motivation of HK, who sought to understand
how present bias of the agent interacts with competitive credit markets,
although our formal model is closest to GZ. HK pursues two related
sets of results: positive predictions about features of credit contracts
we might expect to see, and normative findings about welfare and contractual
distortions. A central positive prediction is that equilibrium credit
contracts offer the option of early repayment on good terms, while
present-biased agents ultimately repay later and are penalised for
doing so. A central normative finding is that agents borrow too much
and repay too late, as evaluated under the preferences of a long-run
self that does not exhibit present bias.

Our analysis addresses both positive and normative dimensions. On
the positive side, our model can capture similar contractual features,
and in particular penalties for early consumption. On the normative
side, time-consistent preferences seem to call for a different welfare
benchmark. As explained below, we then find that consumption often
occurs too late rather than too early. These implications arise not
only when agent beliefs about future discounting are misspecified,
but also when they are correct.

\textbf{Model. }Our model features an agent and at least two firms
that provide credit. The agent first learns a discount factor that
determines how consumption in the initial period is valued relative
to the future. At this point, firms compete by offering dynamic mechanisms
that determine agent consumption over time. The agent chooses a mechanism
to accept and then both parties are fully committed. The agent learns
a new discount factor in each period, taken to be i.i.d., again determining
how future consumption is valued relative to the present. After learning
the discount factor in a period, the agent sends a message to the
mechanism and consumption is determined. In equilibrium, the agent's
strategy results in a process for consumption that satisfies an optimality
condition. Namely, it maximises the agent's discounted expected payoff
at the initial date given (possibly misspecified) agent beliefs. This
is subject to incentive compatibility constraints and a non-negative
profit condition for the firm, where expected profits are evaluated
according to firm beliefs which we view as correct and which therefore
potentially differ from the agent's.

\textbf{Cases closer to the existing literature. }We begin in Section
3 by studying the setting that is closest both in terms of the environment
and the analysis to the earlier work. The agent learns a discount
factor in each period that he believes to be drawn from two possible
values (``patient'' and ``impatient''). Firms, however, correctly
believe that the agent always draws the impatient value. Our view
that firms hold correct beliefs is consistent with much of the behavioural
industrial organisation literature (e.g., Eliaz and Spiegler, 2006,
Grubb, 2009). The agent therefore persistently overestimates the likelihood
of being patient. In the limiting case where the agent believes he
will be patient for sure, our agent has the same preferences and beliefs
as in the quasi-hyperbolic discounting model of GZ where the agent
is ``fully naive''.

Our two-type model predicts contracts similar to those in GZ and HK:
in equilibrium, firms offer the agent the option of highly backloaded
consumption, or more front-loaded consumption but on worse terms.
The agent evaluates firm offers on the basis that backloaded consumption
will be chosen with positive probability, but ultimately always chooses
the more front-loaded consumption. The proximity of our findings to
the earlier literature (HK and GZ) suggests some initial observations
about what is required in order for models to obtain the earlier predictions
on contracts and behaviour. Although our model generates the same
basic behaviour, preferences are time consistent because, at each
date, the agent agrees on the preference ordering conditional on the
realisation of the state (i.e., discount factors). Behaviour is also
dynamically consistent: the agent makes a plan of action for each
state and follows through on the plan, although beliefs about the
likelihood of each state are incorrect. Given that the agent assigns
positive probability to the low discount factor in each period, the
realisation of this discount factor of itself does not suggest the
possibility that the discounting process is misspecified.\footnote{Here, we have in mind the perspective of Galperti (2019), in which
observations that a decision maker deems impossible leads to abandonment
of the maintained model. In our setting, since the agent assigns positive
probability to low discount factors, observing them would not lead
the agent to abandon his understanding of the stochastic process governing
discounting. This contrasts with models of quasi-hyperbolic discounting
by a naive agent, where the agent repeatedly encounters discounting
preferences that differ from his prior understanding. } This compares with the existing literature which views contractual
form as driven by quasi-hyperbolic discounting and (possibly partial)
naivety. There, agents are a sequence of ``multiple selves'' who would
like to maintain more patient behaviour in the future, but are tempted
to make decisions favouring the present. Naivety (including partial
naivety) generally means that the agent has a point belief that the
agent will discount the future less strongly than he actually does.\footnote{See, however, Citanna and Siconolfi (2022), where agent preferences
are time-inconsistent, but where beliefs about future self-control
problems are non-degenerate. We expand on the comparison to this paper
in the literature review below.}

As noted, our alternative conceptualisation suggests a different approach
to welfare, which we use throughout the paper. Given time-consistent
preferences, there is a single decision-making self in our model.
While there is then no decision about which self to prioritise in
welfare analysis, there remains the question (when the agent holds
misspecified beliefs) of the appropriate choice of beliefs. We follow
other industrial organisation studies with misspecified consumer beliefs
by measuring welfare against firm beliefs which are assumed correct.
As Grubb (2009) argues (footnote 10), this approach seems easy to
rationalise on the basis of a firm that has the benefit of encountering
many similar consumers and thus understanding their preferences, while
individual consumers lack this advantage. Then equilibrum consumption
is excessively backloaded (Proposition \ref{Prop:More_Backloaded}),
and in particular it is inefficient as the number of periods under
examination becomes large (Proposition \ref{prop:LONG_RUN_INEFFICIENT}).
The latter claim stands in contrast to GZ (Theorem 1) and Citanna
et al. (2023, Theorem 1), which find approximately efficient timing
of consumption for long horizons. The reason relates to the choice
of efficient benchmark: efficiency is judged against discount factor
realisations equal to the low value, which are viewed as the true
realisations. Our finding is perhaps intuitive: since firms' contracts
cater to an agent that believes he will be more patient in future
than he actually is, equilibrium consumption turns out to be excessively
backloaded (we find that this intuition also carries over to richer
settings -- see below). 

\textbf{Richer settings. }Section 4 then extends the model to settings
with multiple values of the discount factor, and where firm and agent
beliefs have full support over these values. We assume that agents
weakly overestimate their degree of patience, now nesting the case
of correct beliefs. We explain that it is useful to view equilibrium
mechanisms as shifting utility towards more patient agent types --
indeed incentive constraints pointing from lower to higher discount
factors bind (Lemma \ref{lem:UPWARD_BINDS_GEN}). As a consequence,
we find quite generally that more impatient type realisations receive
more front-loaded consumption but on worse terms, where the terms
of the contract are judged against the firms' rate of time preference
(Corollary \ref{cor:WORSE_TERMS}). This includes the case where the
agent has correct beliefs. The result contrasts with HK's (p 2300)
discussion of their alternative ``neoclassical'' model, which they
argue would yield ``essentially the opposite qualitative contract
features'' to the ones predicted by the model with quasi-hyperbolic
discounting.

The similarity in the form of contracts for the setting with correct
agent beliefs is a key motivation for studying our framework. While
a limiting case of our model corresponds to a fully-naive agent in
existing work on quasi-hyperbolic discounting (see GZ), some of the
key forces in the model are preserved as we move towards a model where
agent beliefs are correct. We are therefore better able to evaluate
which behaviours are truly driven by departures from neoclassical
assumptions in the existing literature. For instance, we discuss a
possible difference of our model with correct beliefs and many discounting
values (i.e., many types): penalties for early consumption can be
smaller as the agent smoothly adjusts contractual terms in response
to lower discount factor realisations.

We also show (in Proposition \ref{prop:BACKLOADING}) that excessive
backloading of consumption can arise, including when agent beliefs
are correct. The reason for this finding is that backloading of consumption
relaxes upward incentive constraints and therefore allows higher payoffs
to be delivered in case of higher discount factor realisations. Delivering
higher payoffs for higher discount factors is in fact something called
for by efficiency (see Proposition \ref{prop:EFFICIENT_POLICY}).
In this sense, equilibrium mechanisms can feature distortions in two
directions: too little consumption after high discount factors, but
also consumption that is too backloaded. The use of backloading of
consumption to relax upward incentive constraints is similar to the
backloading of the option that the agent ultimately does not choose
in the baseline models of GZ and HK. The difference when the agent
has correct beliefs is that the agent is correct in anticipating the
probability of selecting options with backloaded consumption.

Although the version of the model with correct agent beliefs also
features distortions, the implications for the role of regulation
are quite different from existing work.\footnote{See HK and Citanna et al. (2025) who consider regulatory interventions
in competitive credit markets with present-biased agents.} Here, the appropriate welfare criterion is unambiguous. Moreover,
the agent's expected payoffs are maximal, subject to firms breaking
even and to incentive compatibility. Since regulatory intervention
would be subject to the same incentive and firm participation constraints,
a regulator cannot improve agent welfare relative to the laissez-faire
outcome.

Finally, we note that our model with stochastic discount factors can
help to understand the role of agent misprediction of future discounting
quite generally. We derive (in Proposition \ref{prop:CONSTANT_PERTURBATION})
inverse Euler equations that characterise equlibrium and efficient
consumption, using a perturbation analysis related to Rogerson (1985).
In the case of logarithmic utility (see Corollary \ref{Cor:LOG-CASE}),
this has an interpretation in terms of expected consumption. When
agent beliefs are correct, the rate of growth in expected consumption
is the same in equilibrium as for the efficient benchmark. Otherwise,
when the agent overestimates patience, expected consumption grows
more quickly in equilibrium. The intuition is that equilibrium mechanisms
pander to more patient discount factor realisations, because the agent
mistakenly places excessive weight on the probability of these realisations. 

\subsection{Further literature}

This paper is related to the literature on dynamic mechanism design.
While many papers (see Pavan, Segal, and Toikka, 2014) consider persistent
types, we specialise here to the case of i.i.d. types. This is common,
for instance, in early literature on incentive-compatible intertemporal
lending and insurance, including Green (1987), Thomas and Worrall
(1990), and Atkeson and Lucas (1992). Our work differs from these
earlier papers especially due to private information on discount factors
rather than income/endowments, and because we allow agent beliefs
to be misspecified.

Work involving stochastic discount factors is quite common in macroeconomics
and finance -- see Stachurski and Zhang (2021) for a review. In our
model of the agent, discount factors can be interpreted as subjective
states, in the sense of Kreps (1979) and Dekel et al. (2001). Gul
and Pesendorfer (2004) and Higashi et al. (2009) provide applications
of this framework to dynamic decision problems. Notably, Higashi et
al. provide a decision theoretic foundation for modeling an agent
with an i.i.d. stochastic discount factor, as in our paper.

Our application to dynamic credit markets warrants further discussion
of Citanna and Siconolfi (2022). This paper, like ours, allows the
agent to have non-degenerate but possibly misspecified beliefs about
future discounting. However, key differences are that the agent has
time-inconsistent preferences and (as the authors remark, p 11) there
is disagreement between the preferences of each self and those anticipated
for future selves (ruling out the case often descibed in the literature
as ``fully naive'', see for instance GZ, p 797). Citanna and Siconolfi
show that, unlike in the model of GZ, inefficiencies need not vanish
asymptotically. They show that there may be too little or too much
consumption at the initial date (``underborrowing'' or ``overborrowing'').
Our findings on asymptotic inefficiency in Section \ref{sec:Degenerate_Impatience}
are driven by different considerations, especially related to the
choice of welfare benchmark. Our results on excessive backloading
of consumption go beyond statements about consumption at the initial
date.

\section{General Model\label{sec:MODEL}}

\textbf{Environment and payoffs. }A number $J\geq2$ of firms, $j=1,2,\dots,J$,
offer long-term lending contracts to a single representative agent
over a finite number $T\ge3$ periods, $t=1,2,\dots,T$. The firm
whose contract is accepted can be viewed as making payments to and
receiving payments from the agent, with the agent able to make payments
out of available income that he receives as an endowment in each period.
This arrangement can equally be represented as the contract determining
a non-negative agent consumption $c_{t}$ in period $t$, and the
lender being claimant on all remaining income. Let $R\geq1$ be the
gross interest rate of firms (all firms face the same interest rate).
Then we let $I_{T}>0$ denote a firm's date-1 value of the agent's
deterministic income stream when the horizon length is $T$. It is
natural to view $I_{T}$ as increasing in $T$ and converging to some
$I_{\infty}<\infty$ with $T$. A firm's profit in a setting with
horizon length $T$ is 
\[
I_{T}-\sum_{t=1}^{T}\frac{c_{t}}{R^{t-1}}.
\]
We sometimes refer to the date-$t$ NPV of a consumption stream $\left(c_{\tau}\right)_{\tau=t}^{T}$
as the firm's cost of the agent's consumption.

The agent enjoys a payoff from a consumption stream $\left(c_{t}\right)_{t=1}^{T}$
which is determined through a per-period utility function $u:\mathbb{R}_{+}\rightarrow\mathbb{R\cup\left\{ -\infty\right\} }$
and the realisation of a sequence of discount factors $\left(\delta_{t}\right)_{t=1}^{T-1}$,
with $\delta_{t}>0$ for each $t$. In particular, the agent's realised
utility is 
\[
u\left(c_{1}\right)+\sum_{t=2}^{T}\Pi_{s=1}^{t-1}\delta_{s}u\left(c_{t}\right).
\]
We assume that $u$ is continuous, strictly increasing, strictly concave,
twice continuously differentiable on $(0,\infty)$, and that $\lim_{c\rightarrow\infty}u\left(c\right)=\infty$
while $\lim_{c\rightarrow\infty}u'\left(c\right)=0$. In Section \ref{sec:Degenerate_Impatience},
we impose that $u\left(0\right)=0$, so there is a finite lower bound
on the utility that may be obtained (that the minimum utility is zero
is a convenient normalisation, without loss of generality). Our motivation
here is to bring our environment as close as possible to HK and especially
GZ, to show that our framework captures similar economic logic and
to highlight differences. In Section \ref{sec:RICHER}, we then allow
also the case where $\lim_{c\downarrow0}u\left(c\right)=-\infty$
(and in such cases we may as well specify $u\left(0\right)=-\infty$).
This will have the advantage of guaranteeing that there are no corner
solutions.

\textbf{Agent discount factors (types). }Our general model of agent
preferences (although markedly specialised in Section \ref{sec:Degenerate_Impatience})
is as follows. In each period $t$, the agent privately learns his
discount factor $\delta_{t}$ that determines how future payoffs are
discounted relative to the present. There is a finite number $N\geq2$
of possible values for $\delta_{t}$ in each period (the agent's date-$t$
``type''), given by $\delta^{(n)}$ for $n=1,\dots,N$. Here, $\delta^{(n)}$
is strictly increasing in $n$. Let $\Delta$ be the collection of
the $N$ possible types. The value $\delta_{t}$ is the realisation
of a random variable $\tilde{\delta}_{t}$ that is commonly known
to be distributed i.i.d. across periods. Firms correctly perceive
the probability of type $\delta^{(n)}$ to be $p_{n}\geq0$ whereas
the agent perceives the probability to be $q_{n}\geq0$ (with agent
beliefs known to firms). Naturally, $\sum_{n=1}^{N}p_{n}=\sum_{n=1}^{N}q_{n}=1$.
We have in mind an agent who believes all types are possible and so
assume $q_{n}>0$ for all $n$ (we refer to this requirement as agent
beliefs having ``full support''). We assume that the agent is ``more
optimistic'' than the firms about patience in the sense of first-order
stochastic dominance: i.e., for all $m$, $\sum_{n=1}^{m}p_{n}\geq\sum_{n=1}^{m}q_{n}$.
We say the agent is ``strictly more optimistic'' if he is more optimistic
and the two distributions differ. The reason we focus on cases where
the agent is more optimistic than the firms is consistency with the
focus of the literature on present bias, the agent believes he will
on average act more patiently in future than he actually does. However,
because we allow that $p_{n}=q_{n}$ for all $n$, we nest the case
with common (and hence correct) beliefs.

A history of types up to date $t\leq T-1$ is denoted $H^{t}\in\Delta^{t}$.
We identify with any history $H^{T}$ the $T-1$-tuple corresponding
to the first $T-1$ realisations. In general, we write $H_{t}=\delta_{t}$
for the $t^{th}$ element of the history and write $H_{T}=\delta_{T}=\emptyset$,
understanding that $\left(H_{1},H_{2},\dots,H_{T-1},\emptyset\right)=\left(H_{1},H_{2},\dots,H_{T-1}\right)$.
We write for $\tau\geq t$, $H_{t}^{\tau}=\left(H_{t},H_{t+1},\dots,H_{\tau}\right)$
(and again $H_{t}^{T}=H_{t}^{T-1}$, the restriction to the first
$T-t$ elements). 

\textbf{Feasible mechanisms.}\footnote{In defining the available mechanisms, and subsequently restricting
equilibria, we ensure that equilibria are characterised through optimisation
problems that have a similar structure to the ones in HK and GZ. In
particular, the agent's date-1 expected payoff is maximised (for each
initial type $\delta_{1}=\delta^{\left(n\right)}$) subject to incentive
constraints and a firm break-even condition. This has the apparent
benefit of aiding comparability with HK and GZ. }\textbf{ }Each firm $j$ makes a one-time simultaneous initial move
by offering a dynamic mechanism $M_{j}$. The agent can accept only
one mechanism, and if he accepts the mechanism of firm $j$, then
both are committed to it. The agent sends a message in each period
and enjoys the resulting consumption. We assume that contractual offers
are private to the agent and not contractible/verifiable, so mechanisms
that are not accepted play no further role.\footnote{The assumption ensures, in particular, that the terms of the contract
under a given contractual offer do not depend on the set of other
contractual offers he receives. Therefore, the payoffs available to
the agent from a given offer will not depend on the other offers he
receives. } We also restrict attention to deterministic mechanisms with a finite
set of messages in each period. While we believe these latter restrictions
can be relaxed without affecting the results, they simplify the analysis
and seem justifiable on grounds of realism. 

\textbf{Equilibrium restrictions. }We restrict attention to equilibria
in which the agent plays a pure message strategy and breaks indifferences
in favour of the firm: If, at any history, the agent is indifferent
over a set of messages in the mechanism he has accepted from a given
firm, then he selects a message which maximises the firm's continuation
payoff (expected NPV of profits). We also restrict attention to equilibria
in which firms choose pure strategies. 

Given the restriction of feasible mechanisms and the equilibrium restrictions
above, there is no loss of generality in restricting attention to
equilibria involving only deterministic incentive-compatible direct
mechanisms, with the set of messages in each period $t\leq T-1$ equal
to $\Delta$. By this, we mean that, for any equilibrium in which
a firm or firms offer a mechanism that is not direct, there is another
equilibrium with the same distribution over outcomes where all offered
mechanisms are direct.\footnote{\label{fn:DIRECT_MECH}For all $j$ and $n$, any collection of mechanisms
$\left(M_{1},\dots,M_{J}\right)$, let $\sigma_{j,n}^{orig}\left(M_{1},\dots,M_{J}\right)$
denote the probability the agent of type $\delta_{1}=\delta^{\left(n\right)}$
accepts the mechanism of Firm $j$ in the original equilibrium. Supposing
Firm 1 chooses some (possibly indirect) mechanism $M_{1}^{orig}$
in the original equilibrium, define a new equilibrium in which Firm
1 chooses the outcome equivalent direct mechanism $M_{1}^{new}$ (i.e.,
it induces, for each agent type history, the same values of consumption
given that the agent reports truthfully), while all other firms' choices
remain unchanged. Then define a modified agent strategy such that,
for all $j$ and $n$, and all $\left(M_{1},\dots,M_{J}\right)$,
$\sigma_{j,n}^{new}\left(M_{1},\dots,M_{J}\right)=\sigma_{j,n}^{orig}\left(M_{1},\dots,M_{J}\right)$
whenever $M_{1}\neq M_{1}^{new}$, while $\sigma_{j,n}^{new}\left(M_{1}^{new},\dots,M_{J}\right)=\sigma_{j,n}^{orig}\left(M_{1}^{orig},\dots,M_{J}\right)$.
Moreover, the agent reports truthfully in Firm 1's mechanism $M_{1}^{new}$,
while his reporting strategy is otherwise unchanged. Then we have
a new equilibrium in which Firm 1's mechanism is direct. The same
argument can be applied iteratively to all firms, establishing the
claim.} A deterministic direct mechanism $M^{D}$ is simply (though slightly
abusively) a sequence of functions $\left(c_{t}\right)_{t=1}^{T}$,
with $c_{t}:\Delta^{t}\rightarrow\mathbb{R}_{+}$ for $t\leq T-1$
and $c_{T}:\Delta^{T-1}\rightarrow\mathbb{R}_{+}$. Since our interest
is in direct mechanisms in which the agent reports truthfully in equilibrium,
we narrow the usual definition of a direct mechanism, requiring the
following. For a mechanism to be called direct, at any history such
that the agent is indifferent over messages, truthful revelation must
be a report that maximises firm profits.\footnote{Put differently, if there are multiple optimal reporting strategies
for the agent, truth-telling must maximise firm profits among them.} It is then without loss of generality to restrict attention to equilibria
in which the agent reports truthfully in any incentive-compatible
direct mechanism. 

\textbf{Timing. }The timing of the game can then be summarised as
follows. In the initial period:
\begin{enumerate}
\item First, the agent learns his date-1 type $\delta_{1}=\delta^{(n)}$,
determining the discounting of the future relative to date 1.
\item Next, the $J$ firms simultaneously offer a mechanism subject to the
feasibility requirements.
\item The agent of realised initial type $\delta^{(n)}$ chooses among the
$J$ offered mechanisms, accepting at most one of them, and sends
an initial message.
\item The mechanism determines the agent's date-$1$ consumption.
\end{enumerate}
In a generic period $t\geq2$, the mechanism has already been determined
and a history of reports $H^{t-1}$ have been made. The timing is
then as follows: 
\begin{enumerate}
\item The agent learns his date $t$ type $\delta_{t}\in\Delta$ (unless
$t=T$, in which case he learns nothing).
\item The agent makes a report $\hat{\delta}_{t}$ to the mechanism (unless
$t=T$, in which case there is no report).
\item The mechanism provides the agent with consumption $c_{t}\left(H^{t-1},\hat{\delta}_{t}\right)$
(or consumption $c_{t}\left(H^{t-1}\right)$ if $t=T$). 
\end{enumerate}
\textbf{Equilibrium concept.} Formally, the equilibrium notion we
apply is Perfect Bayesian Equilibrium, but since firms make an initial
one-shot move, posterior beliefs never need to be calculated (the
updating according to Bayes' rule requirement is trivially satisfied).
Recall that we are restricting to equilibria where:
\begin{itemize}
\item The agent chooses a pure message strategy and breaks indifferences
among messages in a firm's mechanism by selecting the most profitable.
\item The agent reports truthfully in any incentive-compatible direct mechanism.
\item Firms make pure-strategy offers of deterministic incentive-compatible
direct mechanisms.
\end{itemize}

\subsection{Discussion of relationship to quasi-hyperbolic discounting}

An important motivation for studying models with present-biased agents
is apparent ``choice reversals'' of the kind described by Ainslie
(1991, p 334): ``a majority of adults report that they would rather
have \$50 immediately than \$100 in 2 years, but almost no one prefers
\$50 in 4 years over \$100 in 6 years, even though this is the same
choice seen at 4 years' greater distance''.\footnote{Ainslie's example comes from research by Ainslie and Haendel (1983).}
These choice reversals can be explained using quasi-hyperbolic discounting,
in which the agent is viewed as maximising preferences at each date
$t$ given by 
\[
u\left(c_{t}\right)+\beta\bar{\delta}\sum_{\tau>t}\bar{\delta}^{\tau-1-t}u\left(c_{\tau}\right)
\]
where $\bar{\delta}\in\left(0,1\right)$ represents long-run discounting
and $\beta\in\left(0,1\right)$ represents present bias.

But quasi-hyperbolic discounting is not the only model of preferences
that can account for the choice reversals above, and in fact our model
can do so. To see this, consider our model with two discount factors
($N=2$), and with $\delta^{\left(1\right)}=0.4$ and $\delta^{\left(2\right)}=0.9$.
Suppose initially that $q_{2}=p_{2}=1/4$ so that $q_{1}=p_{1}=3/4$.
Suppose the agent, having learned his discount factor $\delta_{t}$
at date $t$, has to choose between utility of consumption equal to
$50$ at date $t$ or utility equal to $100$ at date $t+1$ (Problem
A). From the same perspective (i.e., at date $t$ and given knowledge
of $\delta_{t}$), for some $s\geq1$, the agent has to choose between
utility of $50$ at date $t+s$ or utility of $100$ at date $t+s+1$
(Problem B). Then, the agent chooses the immediate reward $50$ in
Problem A when $\delta_{t}=0.4$ and hence with probability $3/4$.
But he always chooses the later reward of $100$ in Problem B. The
possibility that the agent overestimates future patience strengthens
the reversal in choice: e.g. if instead $p_{2}=0$ and $p_{1}=1$,
then the agent certainly has the low discount factor at date $t$
($\delta_{t}=0.4$) and so the agent chooses the immediate reward
for sure in Problem A, while it does not affect his choice in Problem
B. Conversely, increasing $q_{2}$ strengthens the agent's preference
for the delayed reward in Problem B.

It helps to understand our model of preferences through the lens of
Lu and Saito (2018), which provides an axiomatic foundation for a
``Random Discounting Representation'' underlying consumer choice.
They focus on a one-shot ``ex ante'' decision problem, with stochastic
choice over possibly random consumption streams over an infinite horizon.
Extending our model of preferences to the infinite horizon, and considering
choice at any fixed date $t$, our model of agent preferences generates
stochastic choice that is representable by random discounting in the
sense of Lu and Saito.

For choice at date $t$, we need to define a distribution over discount
factors $D\left(\tau\right)$, for $\tau\geq t$, where $D\left(\tau\right)$
denotes the total discounting of date-$\tau$ utility and $D\left(t\right)=1$.
For $i\in\left\{ 1,2\right\} $, let $D^{(i)}\left(t\right)=1$; and,
for $s\geq1$, let
\[
D^{(i)}\left(t+s\right)=\beta^{\left(i\right)}\bar{\delta}^{s},
\]
where $\bar{\delta}=q_{1}\delta^{\left(1\right)}+q_{2}\delta^{\left(2\right)}$
and $\beta^{\left(i\right)}=\frac{\delta^{\left(i\right)}}{\bar{\delta}}$.
Let $D^{(i)}$ be the discounting that occurs with probability $p_{i}$.
The date-$t$ choice behaviour of our agent -- restricted to choices
over consumption streams as in Lu and Saito (2018) -- is generated
by this stochastic discounting representation. 

The representation here is reminsicent of what Lu and Saito define
(p 787) as a ``random quasi-hyperbolic representation'', involving
a distribution over deterministic quasi-hyperbolic discounting functions.
However, if $p_{2}\in\left(0,1\right)$, then their definition is
not satisfied, as $\beta^{\left(2\right)}>1$. This conclusion extends
immediately to the case of $N>2$ discount factors when each occurs
with positive probability (i.e., $p_{n}>0$ for all $n$). In contrast,
in Section \ref{sec:Degenerate_Impatience}, we assume that $p_{2}=0$,
and in this case date-$t$ choice \emph{is} representable by deterministic
quasi-hyperbolic discounting as $\beta^{(1)}<1$. 

Relative to the above discussion, it is worth emphasizing that the
focus of Lu and Saito (2018) is one-shot choice over streams of (possibly
random) consumption. In our contracting model, by contrast, the agent
makes decisions that affect future choices under later discount factor
realisations. An important determinant of this dynamic setting is
the agent's beliefs about how discounting will evolve and affect future
behaviour. Therefore, agent behaviour in our model would generally
differ if beliefs about future discounting were instead taken directly
from the stochastic discounting representation above. For this reason,
even when the stochastic discounting representation introduced above
is quasi-hyperbolic (as in Section \ref{sec:Degenerate_Impatience}),
our model differs from quasi-hyperbolic models in the contracting
literature.

Finally, consider $N=2$ and suppose, contrary to what we have assumed,
that $q_{2}=1$. Then if, in addition, $p_{1}=1$, agent preferences
and beliefs coincide with those of the ``fully naive'' agent in the
baseline environment of GZ. Therefore, quasi-hyperbolic discounting
in the sense understood by the contracting literature does emerge
as a limiting case of our model. 

\section{Agent with degenerate impatience\label{sec:Degenerate_Impatience}}

The central objective of this section is to argue that our model with
stochastic discount factors can accommodate a logic close to the one
articulated in HK and GZ: the agent is offered the possibility of
strongly backloaded consumption plans which he may believe he will
likely use, but it turns out that he in fact opts for more front-loaded
options. The difference in our setting is that the agent is uncertain
about future levels of patience and understands that both patient
and impatient future states are possible. In our model, there appears
to be no obvious notion of the preferences of a more patient ``long-run
self'', and we discuss the implications of using realised preferences
(equivalently, the firm's beliefs which are understood to be correct)
as the welfare benchmark. 

Our analysis in this section makes several restrictions of the model
to bring it closest to the earlier work. There are two types: $N=2$.
The agent is certainly impatient: $p_{1}=1$. As noted above, agent
utility satisfies $u\left(0\right)=0$.  (Other restrictions on the
environment and on equilibria are as above.)

In this setting, firms compete for an agent whose initial type $\delta_{1}$
is certainly $\delta^{(1)}$. There is then no loss of generality
in defining direct mechanisms to omit the date-1 report. The possible
message histories in a direct mechanism are then ${\cal H}^{R}$,
all $t$-tuples $(\delta_{1},\delta_{2},...,\delta_{t})$, $t\leq T-1$,
such that $\delta_{1}=\delta^{(1)}$.\footnote{Noting that agent beliefs have full support, mechanisms are thus assumed
to be defined at all type histories the agent believes possible. } 

Suppose the maximum expected discounted payoff that can be delivered
to the agent at date 1 in an incentive-compatible direct mechanism
that generates non-negative profits for the firm exists and is finite
(as will be shown below -- see Lemma \ref{lem:Objective_Naive}).
Then, by a familiar undercutting argument, in any equilibrium, the
agent must accept with probability one a mechanism that solves this
optimisation. Moreover, an equilibrium exists. A full proof follows
the arguments in the proof of Proposition \ref{prop:EQUILIBRIUM_GENERAL}.\footnote{The key arguments are as follows. Maintaining the equilibrium restrictions
in Section \ref{sec:MODEL}, now with direct mechanisms omitting date-1
reports, it is easy to see that no firm can earn strictly positive
equilibrium profits. Indeed, the firm earning the lowest profit could
deviate by offering a mechanism that is accepted with positive probability,
but with slightly higher date-1 consumption. The agent must then earn
the maximum expected payoff in an incentive-compatible direct mechanism
generating non-negative profits for the firm (see, e.g., Problem I
below). Otherwise, given continuity of the value to this problem in
the income $I_{T}$ (from Lemma \ref{lem:Objective_Naive} and the
Theorem of the Maximum), a firm can deviate to a mechanism that provides
the agent with an expected payoff slightly smaller than the maximised
value, while earning positive expected profit. }

Understanding equilibrium therefore reduces to solving the problem
(call it Problem I) of maximising by choice of $c_{t}\left(H^{t}\right)\geq0$,
$1\leq t\leq T$ and $H^{t}\in{\cal H}^{R}$, the agent expected utility
\[
\mathbb{E}_{A}\left[u\left(c_{1}\left(\delta^{(1)}\right)\right)+\delta^{(1)}\sum_{t=2}^{T}\Pi_{s=2}^{t-1}\tilde{\delta}_{s}u\left(c_{t}\left(\delta^{(1)},\tilde{\delta}_{2},\dots,\tilde{\delta}_{t}\right)\right)\right]
\]
(where we follow the convention that $\Pi_{s=2}^{1}\delta_{s}\equiv1$)
subject to agent incentive compatibility constraints and to the firm
break-even condition (introduced formally below).\footnote{The expectation $\mathbb{E}_{A}$ is the expectation given the agent
beliefs, i.e. placing probability $q_{n}$ on each type $\delta^{(n)}$.} 

Now consider the agent incentive constraints. These are constraints
at dates $t\in\left\{ 2,\dots,T-1\right\} $ at histories in $H^{t-1}\in{\cal H}^{R}$.
The first kind of constraint ensures that high types do not imitate
low types:

\begin{align*}
 & u\left(c_{t}\left(H^{t-1},\delta^{(2)}\right)\right)+\delta^{(2)}\mathbb{E}_{A}\left[\sum_{\tau=t+1}^{T}\Pi_{s=t+1}^{\tau-1}\tilde{\delta}_{s}u\left(c_{\tau}\left(H^{t-1},\delta^{(2)},\tilde{H}_{t+1}^{\tau}\right)\right)\right]\\
\geq & u\left(c_{t}\left(H^{t-1},\delta^{(1)}\right)\right)+\delta^{(2)}\mathbb{E}_{A}\left[\sum_{\tau=t+1}^{T}\Pi_{s=t+1}^{\tau-1}\tilde{\delta}_{s}u\left(c_{\tau}\left(H^{t-1},\delta^{(1)},\tilde{H}_{t+1}^{\tau}\right)\right)\right].
\end{align*}
The second kind ensures that low types do not imitate high types:
\begin{align}
 & u\left(c_{t}\left(H^{t-1},\delta^{(1)}\right)\right)+\delta^{(1)}\mathbb{E}_{A}\left[\sum_{\tau=t+1}^{T}\Pi_{s=t+1}^{\tau-1}\tilde{\delta}_{s}u\left(c_{\tau}\left(H^{t-1},\delta^{(1)},\tilde{H}_{t+1}^{\tau}\right)\right)\right]\nonumber \\
\geq & u\left(c_{t}\left(H^{t-1},\delta^{(2)}\right)\right)+\delta^{(1)}\mathbb{E}_{A}\left[\sum_{\tau=t+1}^{T}\Pi_{s=t+1}^{\tau-1}\tilde{\delta}_{s}u\left(c_{\tau}\left(H^{t-1},\delta^{(2)},\tilde{H}_{t+1}^{\tau}\right)\right)\right].\label{eq:LOW-IC}
\end{align}
These incentive constraints are necessary for an incentive-compatible
direct mechanism.\footnote{Noting that agent beliefs have full support, an incentive-compatible
direct mechanism is defined to provide incentives for truth-telling
at all histories the agent believes possible. From the usual replication
argument, this is without loss of generality. } That they are also sufficient follows a standard backward induction
argument. If the agent knows he will be truthful at all histories
at dates $t+1,\dots,T-1$, then the specified incentive constraints
ensure he is also willing to be truthful at date $t$.

Now, for $1\leq t\leq T-1$, let $L^{t}=\left(\delta^{(1)},\delta^{(1)},\dots,\delta^{(1)}\right)$
be the $t-$tuple comprising only elements equal to $\delta^{(1)}$.
Let $L^{T}=L^{T-1}$. Because the firm has degenerate beliefs on histories
$H^{t}=L^{t}$, the non-negative profit constraint is
\[
\sum_{t=1}^{T}\frac{c_{t}\left(L^{t}\right)}{R^{t-1}}\leq I_{T}.
\]

Following the approach also to be used in Section \ref{sec:RICHER},
it convenient to state Problem I as maximising instead over the utilities,
i.e. the values $v_{t}\left(H^{t}\right)\equiv u\left(c_{t}\left(H^{t}\right)\right)\geq u\left(0\right)$,
$1\leq t\leq T$, where $H^{t}\in{\cal H}^{R}$. Letting $\phi=u^{-1}$
denote the inverse of $u$, a direct mechanism can now be written
as the collection of functions $M^{D}=\left(\phi\left(v_{t}\right)\right)_{t=1}^{T}$.
Since a mechanism can be represented by the utilities it provides,
we will generally also write simply $M^{D}=\left(v_{t}\right)_{t=1}^{T}$.

The updated (but equivalent) problem is now Problem II, where the
objective is 
\[
\mathbb{E}_{A}\left[v_{1}\left(\delta^{(1)}\right)+\delta^{(1)}\sum_{t=2}^{T}\Pi_{s=2}^{t-1}\tilde{\delta}_{s}v_{t}\left(\delta^{(1)},\tilde{\delta}_{2},\dots,\tilde{\delta}_{t}\right)\right]
\]
and the incentive constraints are (all $2\leq t\leq T-1$, all $H^{t-1})$

\begin{align}
 & v_{t}\left(H^{t-1},\delta^{(2)}\right)+\delta^{(2)}\mathbb{E}_{A}\left[\sum_{\tau=t+1}^{T}\Pi_{s=t+1}^{\tau-1}\tilde{\delta}_{s}v_{\tau}\left(H^{t-1},\delta^{(2)},\tilde{H}_{t+1}^{\tau}\right)\right]\nonumber \\
\geq & v_{t}\left(H^{t-1},\delta^{(1)}\right)+\delta^{(2)}\mathbb{E}_{A}\left[\sum_{\tau=t+1}^{T}\Pi_{s=t+1}^{\tau-1}\tilde{\delta}_{s}v_{\tau}\left(H^{t-1},\delta^{(1)},\tilde{H}_{t+1}^{\tau}\right)\right],\label{eq:HIGH-IC}
\end{align}
and
\begin{align}
 & v_{t}\left(H^{t-1},\delta^{(1)}\right)+\delta^{(1)}\mathbb{E}_{A}\left[\sum_{\tau=t+1}^{T}\Pi_{s=t+1}^{\tau-1}\tilde{\delta}_{s}v_{\tau}\left(H^{t-1},\delta^{(1)},\tilde{H}_{t+1}^{\tau}\right)\right]\nonumber \\
\geq & v_{t}\left(H^{t-1},\delta^{(2)}\right)+\delta^{(1)}\mathbb{E}_{A}\left[\sum_{\tau=t+1}^{T}\Pi_{s=t+1}^{\tau-1}\tilde{\delta}_{s}v_{\tau}\left(H^{t-1},\delta^{(2)},\tilde{H}_{t+1}^{\tau}\right)\right].\label{eq:LOW-IC-NEW}
\end{align}
The firm's non-negative profit constraint is
\[
\sum_{t=1}^{T}\frac{\phi\left(v_{t}\left(L^{t}\right)\right)}{R^{t-1}}\leq I_{T}.
\]
 In relation to this problem, it is worth noting that $\phi$ inherits
convenient properties from $u$. It is strictly convex, strictly increasing
and twice continuously differentiable on $\mathbb{R}_{+}$, with $\lim_{v\rightarrow\infty}\phi^{\prime}\left(v\right)=\infty$.

In solving the above problem, a key observation is that we can focus
on mechanisms such that (\ref{eq:LOW-IC-NEW}) holds with equality.
In particular, we show the following.
\begin{lem}
\label{lem:IC_Section3}Consider any direct mechanism $M^{D}=\left(v_{t}\right)_{t=1}^{T}$
that, if chosen by the agent, results in non-negative profits for
the firm. Then there is another mechanism $\bar{M}^{D}=\left(\bar{v}_{t}\right)_{t=1}^{T}$
that gives the firm non-negative profits in which, for all $t\geq2$
and all $H^{t-1}$, (\ref{eq:LOW-IC-NEW}) holds as equality. Moreover,
the agent's expected payoff is no lower than under $M^{D}$.
\end{lem}
The perturbation to the mechanism $M^{D}$ to arrive at $\bar{M}^{D}$
involves decreasing the agent's utility $v_{t}\left(H^{t-1},\delta^{(1)}\right)$
and increasing $v_{t}\left(H^{t-1},\delta^{(2)}\right)$ whenever
(\ref{eq:LOW-IC-NEW}) is slack at $H^{t-1}$ ($2\le t\leq T-1$).
This is done in such a way as to leave agent expected payoffs, from
the perspective of dates before $t$, unchanged. We then find that
firm profits increase if $H^{t-1}$ is the sequence of $\delta^{(1)}$
realisations and are unchanged otherwise. The increase in profits
can be transferred back to the agent through an increase in $v_{1}\left(H_{1}\right)$.
Hence, in any solution to Problem II, the inequality (\ref{eq:LOW-IC-NEW})
must hold as equality when $H^{t-1}=L^{t-1}$ (for $2\leq t\leq T-1$).

Our next result writes the objective in Problem II after substituting
in the incentive constraints that hold with equality.
\begin{lem}
\label{lem:Objective_Naive}Any solution $M^{D}=\left(v_{t}\right)_{t=1}^{T}$
to Problem II must be such that, for all $t\in\left\{ 2,3,\dots,T-1\right\} $,
$v_{t}\left(L^{t-1},\delta^{(2)}\right)=u\left(0\right)$. Moreover,
for $t\in\left\{ 1,2,\dots,T\right\} $, the values $v_{t}\left(L^{t}\right)$
are the unique maximisers of the expression
\begin{align}
\sum_{t=1}^{T-1}\left(q_{1}\delta^{(1)}+q_{2}\delta^{(2)}\right)^{t-1}v_{t}\left(L^{t}\right)+\left(q_{1}\delta^{(1)}+q_{2}\delta^{(2)}\right)^{T-2}\delta^{(1)}v_{T}\left(L^{T}\right)\label{eq:EQUILIBRIUM_PAYOFF_SEC3}
\end{align}
subject to the firm breaking even. The maximised value represents
the agent's equilibrium payoff. Moreover, any mechanism the agent
accepts in equilibrium yields him the aforementioned utilities at
the specified histories (i.e., respectively, at $\left(L^{t-1},\delta^{(2)}\right)$
with $t\in\left\{ 2,3,\dots,T-1\right\} $, and at $L^{t}$ with $t\in\left\{ 1,2,\dots,T\right\} $).
\end{lem}
The result can be compared to the findings of GZ, especially their
Lemma 2 (which applies given that $u$ is unbounded, see Citanna et
al. (2023)), which indeed follows from a similar argument. Although
we deliberately exclude this case, our arguments and findings apply
also when $q_{2}=1$:\footnote{Excluding this case clarifies the distinction of our model from GZ,
with interpretation discussed in the Introduction. } indeed, agent preferences and beliefs then coincide with those of
the ``fully naive'' agent in GZ. In this sense, our findings can be
viewed as generalising GZ's for the case where the agent is fully
naive.

It is worth highlighting that the values $v_{t}\left(L^{t}\right)$
found in Lemma \ref{lem:Objective_Naive} vary continuously in $q_{2}$
by an application of the Theorem of the Maximum, and hence our predictions
for the realised agent utilities approach the ones predicted by GZ
for the limiting case as $q_{2}\rightarrow1$. Moreover, the mechanism
has the same structure in which, at histories $\left(L^{t-1}\right)$
with $t\in\left\{ 2,3,\dots,T-1\right\} $, the agent has the option
of highly backloaded consmption, with zero consumption in the present
period. This shows how our model can approximate the behavioural predictions
of GZ, but in a setting where the agent assigns positive probability
to being impatient.

The form of the agent's payoff derived in Lemma 2, suggests that the
timing of consumption will often be close to a benchmark contract
that maximises the agent's payoff subject to the firm breaking even,
and where the agent's commonly known discount factor is deterministic
and constant at $q_{1}\delta^{(1)}+q_{2}\delta^{(2)}$. That is, the
timing of consumption is similar in equilibria of our model with mis-specified
beliefs to what would be predicted in a simpler setting with correct
beliefs where all parties know the agent's discount factor is $q_{1}\delta^{(1)}+q_{2}\delta^{(2)}$.
Notably, the agent is more patient in the benchmark than the true
preferences of the agent in our model, whose intertemporal preferences
turn out to be given by the discount factor $\delta^{(1)}$.

To make the above argument precise, we can follow the argument in
the Theorem 1 of Citanna et al. (2023) as follows. Recall that $I_{T}$
is increasing and converges to $I_{\infty}$. We say that the problem
of maximising the agent's payoff with the benchmark discount factor
(i.e., $q_{1}\delta^{(1)}+q_{2}\delta^{(2)}$) is ``well-posed'' if:
\[
\sup_{\left(v_{t}\right)_{t=1}^{\infty}}\left\{ \sum_{t=1}^{\infty}\left(q_{1}\delta^{(1)}+q_{2}\delta^{(2)}\right)^{t-1}v_{t}:\sum_{t=1}^{\infty}\frac{\phi\left(v_{t}\right)}{R^{t-1}}\leq I_{\infty}\right\} <\infty.
\]
If the length of the horizon is $T$, then let $W_{T}^{A}$ denote
the maximum agent welfare, subject to the firm breaking even, and
given that the agent has the benchmark discount factor. We have
\[
W_{T}^{A}=\max_{\left(v_{t}\right)_{t=1}^{T}}\left\{ \sum_{t=1}^{T}\left(q_{1}\delta^{(1)}+q_{2}\delta^{(2)}\right)^{t-1}v_{t}:\sum_{t=1}^{T}\frac{\phi\left(v_{t}\right)}{R^{t-1}}\leq I_{T}\right\} .
\]
Let $v_{t}^{E,T}\left(L^{t}\right)$, $t=1,\dots,T$, be the values
referred to in Lemma \ref{lem:Objective_Naive} and let 
\[
W_{T}^{E}=\sum_{t=1}^{T}\left(q_{1}\delta^{(1)}+q_{2}\delta^{(2)}\right)^{t-1}v_{t}^{E,T}\left(L^{t}\right).
\]
Then $W_{T}^{E}$ is the agent's discounted payoff in the setting
with the benchmark discount factor, but evaluated at the equilibrium
choices of our original model, as characterised in Lemma \ref{lem:Objective_Naive}.
We have the following result:
\begin{prop}
\label{Prop:AVERAGE_DF}Suppose that the problem of maximising the
agent's payoff with the benchmark discount factor $q_{1}\delta^{(1)}+q_{2}\delta^{(2)}$
is well-posed. Then $\lim_{T\rightarrow\infty}\left(W_{T}^{A}-W_{T}^{E}\right)=0$.
\end{prop}
Proposition \ref{Prop:AVERAGE_DF} generalises in a specific way the
finding of GZ and Citanna et al. (2023) by providing a benchmark discount
factor ($q_{1}\delta^{(1)}+q_{2}\delta^{(2)}$) against which equilibrium
consumption streams perform well, at least when the horizon is sufficiently
long. (Their result implies the statement in the proposition for $q_{2}=1$.)
As noted, this benchmark is higher than the agent's true discount
factor $\delta^{\left(1\right)}.$ As noted in the Introduction, use
of the latter would appear consistent with a number of papers on models
with mis-specified beliefs, so we now investigate a comparison with
this benchmark in detail. 

We now term the ``efficient policies'' those which maximise the agent's
payoff subject to firm break-even when the agent's discount factor
is commonly known to be equal to $\delta^{(1)}$. We index these policies
with a $B$ (for the second \emph{benchmark}), so that the efficient
path of utilities is given by $\left(v_{t}^{B,T}\left(L^{t}\right)\right)_{t=1}^{T}$.
Using standard Lagrangian arguments and the result in Lemma \ref{lem:Objective_Naive},
we find that the equilibrium policies are more backloaded than the
efficient policies in a sense made precise in the following proposition.
\begin{prop}
\label{Prop:More_Backloaded}Consider the efficient utility path $\left(v_{t}^{B,T}\left(L^{t}\right)\right)_{t=1}^{T}$
and equilibrium utility path $\left(v_{t}^{E,T}\left(L^{t}\right)\right)_{t=1}^{T}$.
Then $v_{1}^{B,T}\left(L\right)\geq v_{1}^{E,T}\left(L\right).$ Moreover,
if $v_{t}^{E,T}\left(L^{t}\right)\geq v_{t}^{B,T}\left(L^{t}\right)$
and $v_{t}^{E,T}\left(L^{t}\right)>0$ for some $t$, then $v_{s}^{E,T}\left(L^{s}\right)\geq v_{s}^{B,T}\left(L^{s}\right)$
for all $s>t$. The inequality is strict where $t\leq T-2$ and $v_{s}^{E,T}\left(L^{s}\right)>0$.
In addition, if $v_{t}^{B,T}\left(L^{t}\right)>v_{t}^{E,T}\left(L^{t}\right)$
for some $t$, then $v_{t'}^{B,T}\left(L^{t'}\right)<v_{t'}^{E,T}\left(L^{t'}\right)$
for some $t'>t$, and hence indeed $v_{s}^{E,T}\left(L^{s}\right)\geq v_{s}^{B,T}\left(L^{s}\right)$
for all $s\geq t'$.
\end{prop}
The characterisation is more straightforward when utilities are always
strictly greater than zero.
\begin{cor}
\label{cor:Backloading}Suppose that $v_{t}^{B,T}\left(L^{t}\right)>0$
for all $t$, which is guaranteed if $u_{+}^{\prime}\left(0\right)=\infty$.
Then there is a $t^{*}\in\left\{ 2,3,\dots,T-1\right\} $ such that
$v_{t}^{E,T}\left(L^{t}\right)<v_{t}^{B,T}\left(L^{t}\right)$ for
all $t<t^{*}$ and $v_{t}^{E,T}\left(L^{t}\right)>v_{t}^{B,T}\left(L^{t}\right)$
for all $t>t^{*}$. 
\end{cor}
Proposition \ref{Prop:More_Backloaded} and Corollary \ref{cor:Backloading}
appears in conflict with HK, who suggest that equilibrium involves
overborrowing. Naturally, the source of this conflict is the different
welfare criterion, which is based on the agent's realised preferences
(equivalently, firms' beliefs about those preferences). It is also
suggests a potential conflict with GZ, who as noted above find approximate
efficiency in long horizons.

To demonstrate the divergence from GZ, let

\[
V_{T}^{B}=\max_{\left(v_{t}\right)_{t=1}^{T}}\left\{ \sum_{t=1}^{T}\left(\delta^{(1)}\right)^{t-1}v_{t}:\sum_{t=1}^{T}\frac{\phi\left(v_{t}\right)}{R^{t-1}}\leq I_{T}\right\} =\sum_{t=1}^{T}\left(\delta^{(1)}\right)^{t-1}v_{t}^{B,T}\left(L^{t}\right)
\]
be the maximised (efficient) surplus: i.e. the agent's maximal utility
subject to the firm breaking even given true discount factor $\delta^{\left(1\right)}$.
Let
\[
V_{T}^{E}=\sum_{t=1}^{T}\left(\delta^{(1)}\right)^{t-1}v_{t}^{E,T}\left(L^{t}\right),
\]
i.e. the value of surplus evaluated with the discount factor $\delta^{(1)}$
for the utilities realised in equilibrium. To simplify, we restrict
attention to the case where $\delta^{(1)}R\leq1$, and hence (from
Equation (\ref{eq:B_Optimality}) in the proof of Proposition \ref{Prop:More_Backloaded})
$\left(v_{t}^{B,T}\left(L^{t}\right)\right)_{t=1}^{T}$ is a weakly
decreasing sequence. We also assume that, as $T$ increases, $I_{T}$
grows according to a fixed additional income $w>0$ per period; i.e.,
$I_{T}=\sum_{t=1}^{T}\frac{w}{R^{t-1}}$. We show the following.
\begin{prop}
\label{prop:LONG_RUN_INEFFICIENT}Suppose that $\delta^{(1)}R\leq1$
and $\left(q_{1}\delta^{(1)}+q_{2}\delta^{(2)}\right)R\phi^{\prime}\left(u\left(I_{\infty}\right)\right)>\phi_{+}^{\prime}\left(0\right)$,
and that $I_{T}$ is determined by a constant income $w$. Then there
is $\varepsilon>0$ and $\bar{T}>0$ such that $V_{T}^{B}-V_{T}^{E}\geq\varepsilon$
for all $T\geq\bar{T}$. 
\end{prop}
\medskip{}
The above result show that although the agent's initial discount factor
is commonly known to be $\delta^{(1)}$, and although equilibrium
involves a contract that maximises the agent's expected payoff conditional
on the firm breaking even, equilibrium often involves inefficiency.
In particular, the agent consumes inefficiently late. While the divergence
with the earlier literature reflects the choice of benchmark, it is
worth commenting on the reason for our result, focusing on the case
where $q_{2}$ is arbitrarily close to 1, i.e. close to the special
case of GZ. A firm's offer aims at increasing the agent's anticipated
payoff, and hence increasing the date-2 continuation payoff in case
$\delta_{2}=\delta^{\left(2\right)}$, to which the agent attaches
high probability. However, this is subject to the agent's upward date-2
incentive constraint, which requires the agent to value sufficiently
his anticipated consumption when $\delta_{2}=\delta^{\left(1\right)}$.
Anticipated consumption comes largely from two channels: date-2 utility
$v_{2}\left(L^{2}\right)$ and the expected continuation utility from
date $3$ given that $\delta_{3}=\delta^{\left(2\right)}$. Raising
the latter is subject to the date-3 upward incentive constraint at
history $L^{2}$, and in turn calls for a higher value of $v_{3}\left(L^{3}\right)$,
explaining why consumption is more back-loaded than efficient between
dates 2 and 3 given $\delta_{2}=\delta^{\left(1\right)}$.

\section{Findings for the richer model\label{sec:RICHER}}

While Section 3 deliberately considers an environment where the screening
logic is close to HK and GZ, we now consider the results for richer
settings. There are now $N\geq2$ possible discount factors. Firm
beliefs $\left(p_{n}\right)_{n=1}^{N}$ (i.e., true probabilities)
are assumed to have full support. We allow that either $u\left(0\right)$
is finite or $\lim_{c\rightarrow0}u\left(c\right)=-\infty$. While
the backloading logic in the model of Section \ref{sec:Degenerate_Impatience}
necessitates the corner solution of zero consumption in some states
(see Lemma \ref{lem:Objective_Naive}), fully interior equilibrium
mechanisms now become possible and indeed will be guaranteed in case
$\lim_{c\rightarrow0}u\left(c\right)=-\infty$. Interiority is an
important restriction for some results.

We begin by briefly characterising and commenting on what may be viewed
as efficient consumption in this environment. Since we view firms
as having the correct beliefs, this maximises 

\[
\mathbb{E}_{F}\left[\sum_{t=1}^{T}\Pi_{s=1}^{t-1}\tilde{\delta}_{s}u\left(c_{t}\left(\tilde{H}^{t}\right)\right)\right],
\]
where $\mathbb{E}_{F}$ is the expectation under firm beliefs, subject
to the firm break-even condition 
\[
\mathbb{E}_{F}\left[\sum_{t=1}^{T}\frac{c_{t}\left(\tilde{H}^{t}\right)}{R^{t-1}}\right]\leq I_{T}.
\]
We assume that the solution to this problem is interior, i.e. consumption
is always strictly positive, which is guaranteed for example if $\lim_{c\rightarrow0}u\left(c\right)=-\infty$
or, if $u\left(0\right)$ is finite, by the Inada condition $u_{+}^{\prime}\left(0\right)=\infty$.
Following the treatment in the previous section, we write the characterisation
in terms of the utilities: for all $t$, all $H^{t}$, $v_{t}\left(H^{t}\right)=u\left(c_{t}\left(H^{t}\right)\right)$. 
\begin{prop}
\label{prop:EFFICIENT_POLICY}The efficient policy $\left(v_{t}^{B,T}\left(H^{t}\right)\right)_{t=1}^{T}$
exists and is unique. If the solution is interior then, for each $t\geq2$,
$v_{t}^{B,T}\left(H^{t}\right)$ is strictly increasing in each of
the first $t-1$ arguments.
\end{prop}
Proposition \ref{prop:EFFICIENT_POLICY} has the immediate implication,
of interest in light of the discussion in Section \ref{sec:Degenerate_Impatience},
that the efficient policy is not implementable by a dynamic mechanism.
In any mechanism seeking to implement the efficient policy, the agent
will simply send the messages intended for the highest realisations
of the discount factor. Thus inefficiency appears ``hard-wired'' given
that the true discount factor realisations are now uncertain.

Now consider equilibrium mechanisms. A direct mechanism gives rise
to specific outcomes $\left(c_{t}\left(\delta_{1},H_{2}^{t}\right)\right)_{t=1}^{T}$
for an agent with initial type $\delta_{1}$ and continuation histories
$H_{2}^{t}$. It will be seen that, for each initial agent type $\delta_{1}$,
in any equilibrium, these outcomes must maximise
\[
\mathbb{E}_{A}\left[u\left(c_{1}\left(\delta_{1}\right)\right)+\delta_{1}\sum_{t=2}^{T}\Pi_{s=2}^{t-1}\tilde{\delta}_{s}u\left(c_{t}\left(\delta_{1},\tilde{H}_{2}^{t}\right)\right)\right]
\]
subject to the firm break-even condition
\[
\mathbb{E}_{F}\left[\sum_{t=1}^{T}\frac{c_{t}\left(\delta_{1},\tilde{H}_{2}^{t}\right)}{R^{t-1}}\right]\leq I_{T}
\]
and to incentive constraints for the truthful reporting of $\delta_{t}$
at all $t\in\left\{ 2,\dots,T-1\right\} $ and all histories $\left(\delta_{1},H_{2}^{t-1}\right)$.
In particular, for each such history, each $\delta_{t}$ and each
$\hat{\delta}_{t}$, we must have 
\begin{align}
 & u\left(c_{t}\left(\delta_{1},H_{2}^{t-1},\delta_{t}\right)\right)+\delta_{t}\mathbb{E}_{A}\left[\sum_{\tau=t+1}^{T}\Pi_{s=t+1}^{\tau-1}\tilde{\delta}_{s}u\left(c_{\tau}\left(\delta_{1},H_{2}^{t-1},\delta_{t},\tilde{H}_{t+1}^{\tau}\right)\right)\right]\nonumber \\
\geq & u\left(c_{t}\left(\delta_{1},H_{2}^{t-1},\hat{\delta}_{t}\right)\right)+\delta_{t}\mathbb{E}_{A}\left[\sum_{\tau=t+1}^{T}\Pi_{s=t+1}^{\tau-1}\tilde{\delta}_{s}u\left(c_{\tau}\left(\delta_{1},H_{2}^{t-1},\hat{\delta}_{t},\tilde{H}_{t+1}^{\tau}\right)\right)\right].\label{eq:IC_GEN}
\end{align}
Call this Problem III. That the incentive constraints ensure incentive
compatibility from date $2$ onwards follows by the same backwards
induction argument as for the previous section.

The following result confirms that Problem III is central to understanding
equilibrium.
\begin{prop}
\label{prop:EQUILIBRIUM_GENERAL}For each $\delta_{1}$, there is
a unique solution to Problem III. This solution represents the outcomes
in the dynamic mechanism accepted by an agent with initial type $\delta_{1}$
in any equilibrium. An equilibrium with such outcomes for any $\delta_{1}$
exists.
\end{prop}
In equilibrium, each initial type of agent $\delta_{1}$ has his expected
payoff maximised subject to incentive compatibility from date 2 onwards,
and subject to firm break-even. In particular, the agent's consumption
process is given by the solution to Problem III. Moreover, there exists
an equilibrium in which all firms offer the direct mechanism defined
by solving Problem III for every initial type $\delta_{1}$. Note
that, while Problem III does not make reference to incentive compatibility
of the agent's date-1 report, this is satisfied because the constraints
of the maximisation for each initial type $\delta_{1}$ are identical.
Because a firm earns zero discounted profits in the solution to Problem
III (as defined for any initial type $\delta_{1}$), and because these
profits do not depend on the agent's true initial type, incentive
compatibility of the date-1 report follows in particular from the
following observation. If a type $\delta_{1}=\delta^{\left(n\right)}$
can gain by deviating to mimic some initial type $\delta^{\left(n'\right)}$,
then the solution to Problem III for $\delta^{\left(n'\right)}$ must
represent an improvement on the solution to Problem III as parameterised
by $\delta_{1}=\delta^{\left(n\right)}$, a contradiction.

As in the proof of Proposition \ref{prop:EQUILIBRIUM_GENERAL}, we
can rewrite Problem III so that the problem is to determine agent
utilities $v_{t}\left(\delta_{1},H_{2}^{t}\right)=u\left(c_{t}\left(\delta_{1},H_{2}^{t}\right)\right)$,
call this Problem IV. Let $\left(v_{t}^{E,T}\left(\delta_{1},H_{2}^{t}\right)\right)_{t=1}^{T}$
denote the solution. Our objective is to characterise this solution.

We begin by showing that optimal policies satisfy a particular form
of monotonicity, and that only local incentive constraints are relevant.
By local incentive constraints, we mean, at any date $t\in\left\{ 2,\dots,T-1\right\} $,
any history $\left(\delta_{1},H_{2}^{t-1}\right)$, the agent with
discount factor $\delta_{t}=\delta^{\left(n\right)}$ prefers to report
truthfully (and then follow truthful reporting thereafter), rather
than report an adjacent discount factor (i.e., $\delta^{\left(n-1\right)}$
if $n\geq2$ or $\delta^{n+1}$ if $n\leq N-1$).
\begin{lem}
\label{lem:IC_Section4}For any date $t\geq2$, any history $\left(\delta_{1},H_{2}^{t-1}\right)$,
$v_{t}^{E,T}\left(\delta_{1},H_{2}^{t-1},\delta_{t}\right)$ is non-increasing
in $\delta_{t}$, while 
\[
\mathbb{E}\left[\sum_{\tau=t+1}^{T}\Pi_{s=t+1}^{\tau-1}\tilde{\delta}_{s}v_{\tau}^{E,T}\left(\delta_{1},H_{2}^{t-1},\delta_{t},\tilde{H}_{t+1}^{\tau}\right)\right]
\]
is non-decreasing. The solution to Problem IV is unchanged if we omit
all incentive constraints other than local incentive constraints.
\end{lem}
The monotonicity finding in Lemma \ref{lem:IC_Section4} follows from
incentive compatibility of the solution to Problem IV. For higher
values of the discount factor, the agent relatively favours higher
future consumption streams, and is willing to give up current consumption.
Thus the monotonicity properties are crucial to sorting (in fact,
they are necessary for incentive compatibility in general, not only
a feature of the optimal mechanism).

Now considering Problem IV but only with local incentive constraints,
we show that upward incentive constraints hold with equality.
\begin{lem}
\label{lem:UPWARD_BINDS_GEN}In the solution to Problem IV, all local
upward incentive constraints hold with equality.
\end{lem}
That upward incentive constraints hold with equality in the solution
to Problem IV can be explained by the value of increasing the payoffs
of higher types. In particular, from the perspective of date 1, the
agent particularly values consumption in the periods following high
discount factor realisations. But consumption after such realisations
can be increased only subject to incentive constraints, and in particular
that the agent does not want to mimic an upward adjacent discount
factor. Higher types are thus given as much consumption as possible,
subject to upward incentive constraints. 

We now consider how contractual terms vary with changes in the agent's
date-$t$ discount factor $\delta_{t}$. In particular, we are interested
in how the firm's expected discounted continuation cost of the agent's
consumption, at any date $t\in\left\{ 2,\dots,T-1\right\} $ and any
$\left(\delta_{1},H_{2}^{t-1}\right)$, varies with the agent's realized
date-$t$ discount factor. This cost is given by
\begin{equation}
\mathbb{E}_{F}\left[\sum_{\tau=t}^{T}\frac{\phi\left(v_{\tau}\left(\delta_{1},H_{2}^{t-1},\delta_{t},\tilde{H}_{t+1}^{\tau}\right)\right)}{R^{t-1}}\right].\label{eq:CON TINUATION_COST}
\end{equation}
In providing the result, we use the notion of the ``continuation policy
at history $\left(\delta_{1},H_{2}^{t-1}\right)$ for discount factor
$\delta^{\left(n\right)}$'' to mean the sequence $\left(v_{\tau}\left(\delta_{1},H_{2}^{t-1},\delta^{\left(n\right)},\cdot\right)\right)_{\tau=t}^{T}$.
The continuation policies at history $\left(\delta_{1},H_{2}^{t-1}\right)$
for two distinct discount factors $\delta^{\left(n'\right)}$ and
$\delta^{\left(n''\right)}$ are distinct if the continuation policies
are different: $\left(v_{\tau}\left(\delta_{1},H_{2}^{t-1},\delta^{\left(n^{\prime}\right)},\cdot\right)\right)_{\tau=t}^{T}\neq\left(v_{\tau}\left(\delta_{1},H_{2}^{t-1},\delta^{\left(n^{\prime\prime}\right)},\cdot\right)\right)_{\tau=t}^{T}$.
\begin{cor}
\label{cor:WORSE_TERMS}In Problem IV, for any date $t\in\left\{ 2,\dots,T-1\right\} $,
any history $\left(\delta_{1},H_{2}^{t-1}\right)$, the firm's expected
discounted continuation cost (as given by Equation (\ref{eq:CON TINUATION_COST}))
is weakly increasing in $\delta_{t}$. For two discount factors $\delta_{t}=\delta^{\left(n'\right)}$
and $\delta_{t}=\delta^{\left(n''\right)}$, $\delta^{\left(n'\right)}<\delta^{\left(n''\right)}$,
if the continuation policies at $\left(\delta_{1},H_{2}^{t-1}\right)$
are distinct, then the firm's expected discounted continuation cost
is strictly higher for $\delta_{t}=\delta^{\left(n''\right)}$.
\end{cor}
The result shows that when the agent realises a lower discount factor,
at dates $t\in\left\{ 2,\dots,T-1\right\} $, then he receives worse
terms in the sense of a lower expected NPV payout, where the NPV is
calculated according to the firm gross interest rate $R$. The result
has value when there is at least some separation in the terms received
by the agent when realising different discount factors from date 2
onwards. The following example illustrates in a tractable setting
that indeed full separation can occur.
\begin{example}
\label{exa:Squareroot}Let $T=3$ and $u\left(c\right)=\sqrt{c}$.
Suppose that discount factors are commonly believed to be uniformly
distributed over values $\delta^{\left(n\right)}=\underline{\delta}+\frac{\left(n-1\right)}{N-1}\left(\bar{\delta}-\underline{\delta}\right)$,
for $n=1,\dots,N$, where $0<\underline{\delta}<\bar{\delta}$. Hence,
$q_{n}=p_{n}=1/N$ for all $n$. If $N$ is sufficiently large, then
for all $\delta_{1}$, $v_{2}\left(\delta_{1},\delta_{2}\right)$
is strictly positive and strictly decreasing in $\delta_{2}$, while
$v_{3}\left(\delta_{1},\delta_{2}\right)$ is strictly positive and
strictly increasing in $\delta_{2}$. 
\end{example}
The example confirms that, for a plausible setting with agreement
on beliefs -- a ``neoclassical'' setting in the language of HK --
the agent finds reporting to be less patient, and thus taking earlier
consumption, more expensive. This higher expense is judged using firms'
intertemporal preferences as captured by the gross interest rate $R$.
As noted in the Introduction, this differs from the implications discussed
by HK of a possible alternative neoclassical model (p 2300). 

To illustrate Example \ref{exa:Squareroot}, we compute numerically
a particular case, displayed in Figures \ref{fig:CONTINUATION_COST}-\ref{fig:DATE-3-CONSUMPTION}.
Figure \ref{fig:CONTINUATION_COST} shows that when the agent has
a higher realised discount factor $\delta_{2}$, he receives in equilibrium
better terms in the sense that the NPV of consumption from date-2
onwards, using the firm's gross interest rate $R$, is higher. Thus,
if the agent chooses front-loaded consumption at date 2 (as he does
when $\delta_{2}$ is small), he is penalised through worse terms
(according to the gross interest rate $R$). Consumption is naturally
then lower at date 3. Compared to the efficient policy (as studied
at the beginning of this section), the dispersion in consumption is
reduced due to the need to respect incentive constraints. In other
words, higher types consume too little and lower types consume too
much. This reflects that incentive compatibility limits the extent
to which the firm can favour more patient types (as demanded by the
efficient policy, see Proposition \ref{prop:EFFICIENT_POLICY}). 
\begin{figure}[H]
\caption{NPV of firm's date-2 continuation cost\label{fig:CONTINUATION_COST}}

\begin{centering}
\includegraphics[scale=0.5]{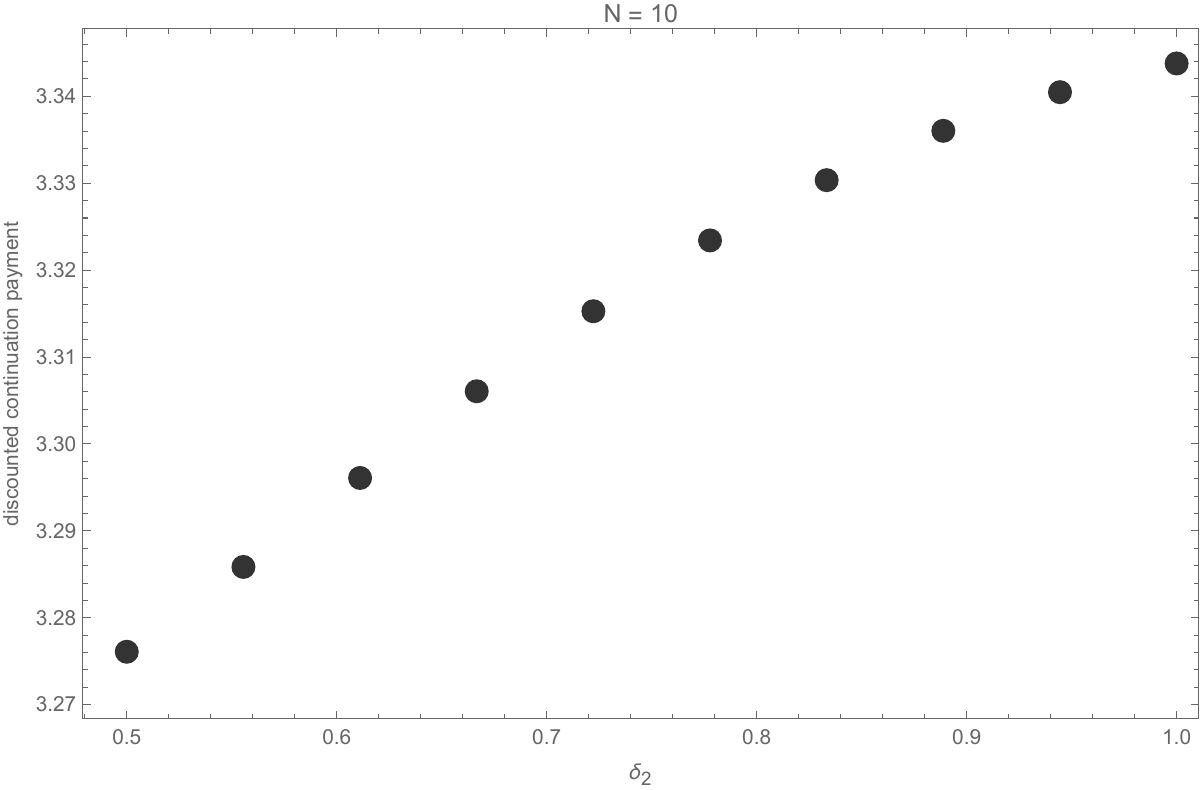}
\par\end{centering}
\centering{}NPV of date-2 continuation cost $c_{2}\left(1,\delta_{2}\right)+\frac{c_{3}\left(1,\delta_{2}\right)}{R}$
as function of $\delta_{2}$ for Example \ref{exa:Squareroot}. Parameter
values: $\underline{\delta}=0.5$, $\bar{\delta}=1$, $N=10$, $R=3/2,$
$I_{3}=3$.
\end{figure}

\begin{figure}[H]
\caption{Equilibrium and efficient date-2 consumption\label{fig:DATE-2-CONSUMPTION}}

\begin{centering}
\includegraphics[scale=0.5]{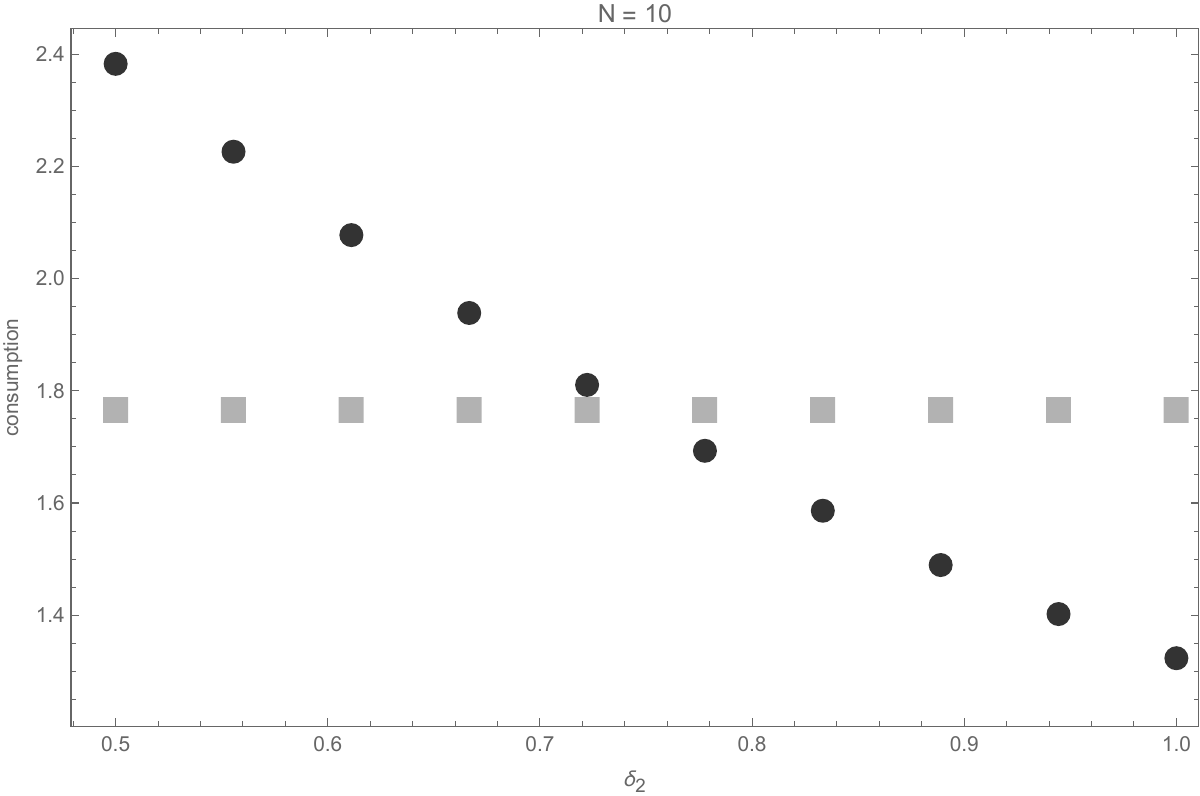}
\par\end{centering}
\centering{}Date-2 consumption $c_{2}\left(1,\delta_{2}\right)$ as
function of $\delta_{2}$: efficient (grey) and equilibrium (black).
Parameterisation as in Example \ref{exa:Squareroot} with $\underline{\delta}=0.5$,
$\bar{\delta}=1$, $N=10$, $R=3/2,$ $I_{3}=3$.
\end{figure}

\begin{figure}[H]
\caption{Equilibrium and efficient date-2 consumption\label{fig:DATE-3-CONSUMPTION}}

\begin{centering}
\includegraphics[scale=0.5]{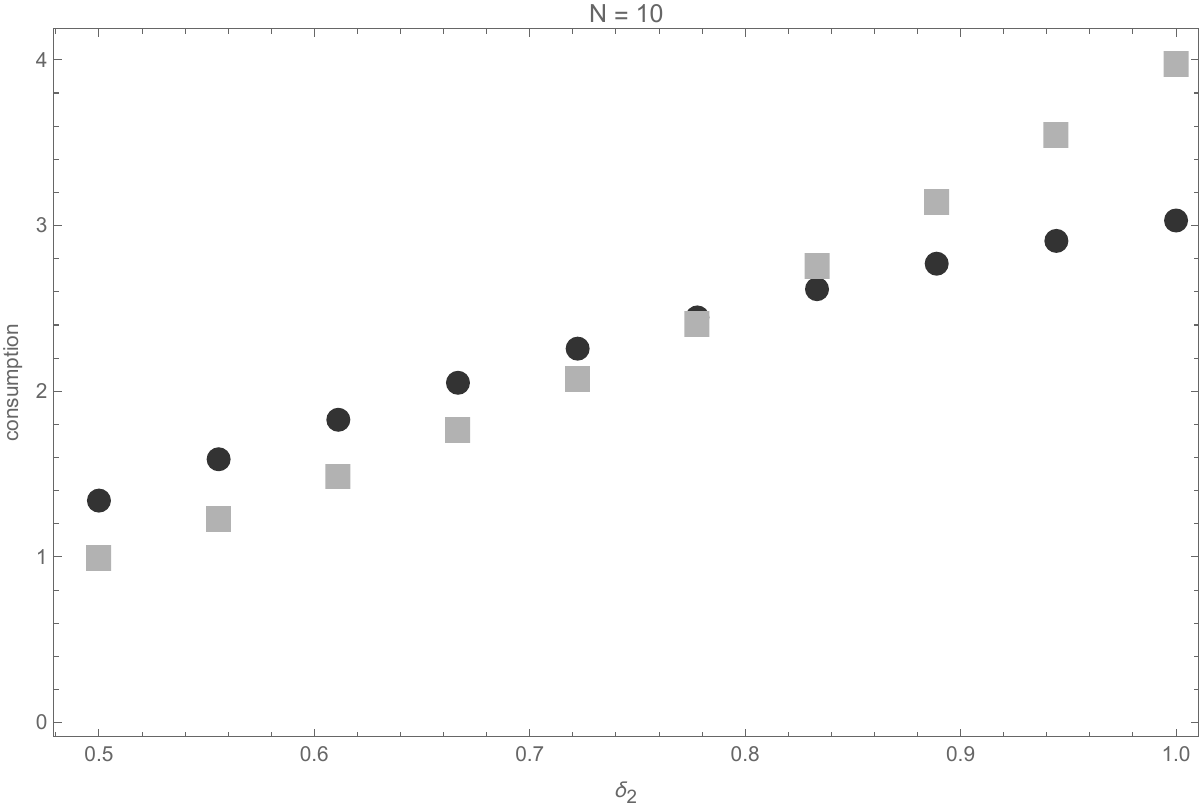}
\par\end{centering}
\centering{}Date-3 consumption $c_{3}\left(1,\delta_{2}\right)$ as
function of $\delta_{2}$: efficient (grey) and equilibrium (black).
Parameterisation as in Example \ref{exa:Squareroot} with $\underline{\delta}=0.5$,
$\bar{\delta}=1$, $N=10$, $R=3/2,$ $I_{3}=3$.
\end{figure}

One important feature of the contracts predicted in HK is the large
magnitude of the penalty for electing earlier consumption. They emphasise
their prediction of ``large penalties for falling behind {[}the{]}
front-loaded repayment {[}i.e., delayed consumption{]} schedule''
(p 2280). In our setting where agent and firm beliefs agree, and with
many types, the equilibrium mechanism is flexible in the sense that
a small reduction in patience corresponds to a small change in contract
terms. Large or discontinuous changes in payment terms are possible
with belief agreement, but then seem to be a feature of specifications
with a small number of types. For instance, for the parameterisation
considered in the figures, but setting $N=2$, the firm's discounted
continuation cost of agent consumption after an initial discount factor
$\delta_{1}=1$, i.e. $c_{2}\left(1,\delta_{2}\right)+\frac{c_{3}\left(1,\delta_{2}\right)}{R}$,
moves from 3.57 if $\delta_{2}=1$ to 3.10 if $\delta_{2}=0.5$, a
13 per cent fall in the NPV of consumption. The key reason for HK's
finding of large penalties for delayed repayment (which can be read
as early consumption in our model) is different and seems sourced
in the agent's misprediction of his future preferences and behaviour.
Because the specification in Section \ref{sec:Degenerate_Impatience}
replicates the logic in the models with quasi-hyperbolic discounting,
our model can also account for situations where small differences
in preferences (e.g., where $\delta^{\left(2\right)}-\delta^{\left(1\right)}$
is small in Section \ref{sec:Degenerate_Impatience}) is associated
with discretely different consumption behaviour, with a large effective
penalty (evaluated according to the gross interest rate $R$) for
early consumption. 

To further analyse the structure of equilibrium repayment, we refer
to the argument in the proof of Example \ref{exa:Squareroot}, where
we derive necessary conditions for optimality in Problem IV when $T=3$
generally. These conditions apply provided the solution is interior
(consumption is strictly positive) and separating (higher values of
$\delta_{2}$ receive lower date-2 consumption and higher date-3 consumption).
Let $\lambda$ be the multiplier on the firm's break-even condition,
let $v_{1}$ be the date-1 utility, $\left(w_{n}\right)_{n=1}^{N}$
be the date-2 utilities at each discount factor $\delta_{2}=\delta^{\left(n\right)}$,
and let $\left(z_{n}\right)_{n=1}^{N}$ be the corresponding date-3
utilities. Let $\bar{Q}_{n}=\sum_{m=n}^{N}q_{n}$. The necessary conditions
can be written:
\begin{equation}
1=\lambda\phi^{\prime}\left(v_{1}\right);\label{eq:DATE_1_NECESSARY}
\end{equation}
\begin{equation}
\delta_{1}=\frac{\lambda}{R}\sum_{n=1}^{N}p_{n}\phi^{\prime}\left(w_{n}\right);\label{eq:OPTIMAL_W1-BODY}
\end{equation}
for $n\geq2$,
\begin{align}
0= & \left(\delta^{\left(n\right)}-\delta^{\left(n-1\right)}\right)\left(\delta_{1}\bar{Q}_{n}-\frac{\lambda}{R}\sum_{m=n}^{N}p_{m}\phi'\left(w_{m}\right)\right)\nonumber \\
 & +\frac{\lambda p_{n}\delta^{\left(n\right)}}{R}\phi^{\prime}\left(w_{n}\right)-\frac{\lambda p_{n}}{R^{2}}\phi'\left(z_{n}\right);\label{eq:EULER-BODY}
\end{align}
and
\begin{equation}
\delta^{\left(1\right)}R\phi^{\prime}\left(w_{1}\right)=\phi^{\prime}\left(z_{1}\right).\label{eq:EFFICIENCY_BOTTOM-BODY}
\end{equation}

One set of insights suggested by these conditions relates to distortions
in the timing of date 2 and date 3 consumption. In particular, note
that for a fixed amount of consumption (in NPV terms, according to
the firm's gross discount factor $R$) to be efficiently distributed
between date 2 and date 3 for a given $\delta_{2}=\delta^{\left(n\right)}$
requires 
\[
\delta^{\left(n\right)}R\phi^{\prime}\left(w_{n}\right)=\phi^{\prime}\left(z_{n}\right).
\]
This condition is necessary and sufficient for the agent's payoff
not to be increased by any movement of consumption between dates 2
and 3 in state $\delta_{2}=\delta^{\left(n\right)}$, subject to such
change not reducing the firm's profit. Then Equation (\ref{eq:EFFICIENCY_BOTTOM-BODY})
provides a sense in which there are ``no distortions at the bottom''.
Using Equation (\ref{eq:OPTIMAL_W1-BODY}), and the assumption that
agent beliefs ($q_{n}$) first-order stochastically dominate firm
beliefs ($p_{n}$), the first term on the right side of Equation (\ref{eq:EULER-BODY})
is strictly positive, for $n\ge2$. This implies 
\begin{equation}
\delta^{\left(n\right)}R\phi^{\prime}\left(w_{n}\right)<\phi^{\prime}\left(z_{n}\right)\label{eq:BACKLOADED_N}
\end{equation}
and so consumption is ``too backloaded'' relative to efficiency. These
observations can be summarised as follows.
\begin{prop}
\label{prop:BACKLOADING}Let $T=3$ and fix $\delta_{1}$. Suppose
the solution to Problem IV is interior (all consumption is strictly
positive) and that the mechanism is separating (higher values of $\delta_{2}$
receive strictly lower date-2 consumption but strictly higher date-3
consumption). Then for all values $\delta_{2}=\delta^{\left(n\right)}$,
$n\geq2$, the agent's date-2 continuation payoff given $\delta_{2}$
can be increased in the absence of incentive constraints by shifting
consumption earlier (i.e., to date 2), while leaving the NPV of firm
profits unchanged. In this sense, equilibrium consumption at dates
2 and 3 is inefficiently backloaded.
\end{prop}
Proposition \ref{prop:BACKLOADING} provides a sense in which inefficient
backloading of consumption can occur in equilibrium, including when
the agent and firms have common beliefs ($q_{n}=p_{n}$ for all $n$).
Notice that the statement is conditional on the realisation of the
discount factor $\delta_{2}$ and so the statement is not sensitive
to the beliefs over discount factors one may apply to evaluate welfare.
It is also worth observing that the distortion described in Proposition
\ref{prop:BACKLOADING} will often vanish for the highest type $\delta_{2}=\delta^{\left(N\right)}$
as the number of types $N$ becomes large (we might say there is ``approximate
efficiency at the top''). The calculations in the proof of Example
\ref{exa:Squareroot} imply that, at least in this example, we have
$\phi^{\prime}\left(z_{N}\right)-\delta^{\left(N\right)}R\phi^{\prime}\left(w_{N}\right)\rightarrow0$,
and so there is approximate effiency in the timing of consumption
for the highest type.

The reason consumption is more backloaded than efficient in the sense
of Proposition \ref{prop:BACKLOADING} is simply that increasing the
date-3 utility of a given value of $\delta_{2}>\delta^{\left(1\right)}$
relaxes the upward incentive constraint pointing to this type. Indeed,
higher types have a relative preference for more backloaded consumption,
compared to those types below. The effect can be understood by considering
an envelope representation of the agent's date-2 continuation payoff:
\[
w_{n}+\delta^{\left(n\right)}z_{n}=w_{1}+\delta^{\left(1\right)}z_{1}+\sum_{i=2}^{n}\left(\delta^{\left(i\right)}-\delta^{\left(i-1\right)}\right)z_{i}.
\]
 This expression is derived using that upward incentive constraints
bind. Increases in the date-3 consumption $z_{i}$, for some $i\geq2$,
increase the date-2 continuation utility $w_{n}+\delta^{\left(n\right)}z_{n}$
of all higher types~$\delta_{2}=\delta^{\left(n\right)}\geq\delta^{\left(i\right)}$,
while leaving the utilities of types below $\delta^{\left(i\right)}$
unchanged. This increases the extent to which the mechanism can favour
higher realisations of $\delta_{2}$ at date 3 (which, as we have
seen, is associated with greater efficiency) while respecting incentive
compatibility. This effect on the dispersion of date-2 continuation
payoffs does not apply however for $\delta_{2}=\delta^{\left(1\right)}$:
an increase in $z_{1}$ raises the date-2 continuation utility of
all types equally. This helps to explain why we predict the efficent
level of backloading between dates 2 and 3, conditional on $\delta_{2}=\delta^{\left(1\right)}$.

Another useful perturbation that can shed light on the dynamics of
equilibrium and efficient consumption is to consider shifting utility
into the subsequent period by a constant that does not respond to
type. This has the advantage of preserving incentive constraints and
is therefore a relevant perturbation both in Problem IV and in studying
the efficient policy. For $T\geq3$, we consider such perturbations
at dates $t\in\left\{ 2,\dots,T-1\right\} $ conditional on history
$\left(\delta_{1},H_{2}^{t-1}\right)$, as well as between the first
and second period. We find the following.
\begin{prop}
\label{prop:CONSTANT_PERTURBATION}Consider any $T\geq3$ and suppose
both equilibrium and efficient consumption are interior -- i.e.,
all consumption values are strictly positive. For any $t\in\left\{ 2,\dots,T-1\right\} $,
any history $\left(\delta_{1},H_{2}^{t-1}\right)$, the equilibrium
utilities satisfy
\[
\mathbb{E}_{F}\left[\phi^{\prime}\left(v_{t+1}^{E,T}\left(\delta_{1},H_{2}^{t-1},\tilde{\delta}_{t},\tilde{\delta}_{t+1}\right)\right)\right]=R\mathbb{E}_{A}\left[\tilde{\delta}_{t}\right]\mathbb{E}_{F}\left[\phi^{\prime}\left(v_{t}^{E,T}\left(\delta_{1},H_{2}^{t-1},\tilde{\delta}_{t}\right)\right)\right]
\]
while the efficient policies satisfy 
\[
\mathbb{E}_{F}\left[\phi^{\prime}\left(v_{t+1}^{B,T}\left(\delta_{1},H_{2}^{t-1},\tilde{\delta}_{t},\tilde{\delta}_{t+1}\right)\right)\right]=R\mathbb{E}_{F}\left[\tilde{\delta}_{t}\right]\mathbb{E}_{F}\left[\phi^{\prime}\left(v_{t}^{B,T}\left(\delta_{1},H_{2}^{t-1},\tilde{\delta}_{t}\right)\right)\right].
\]
Similarly,
\[
\mathbb{E}_{F}\left[\phi^{\prime}\left(v_{2}^{E,T}\left(\delta_{1},\tilde{\delta}_{2}\right)\right)\right]=R\delta_{1}\phi^{\prime}\left(v_{1}^{E,T}\left(\delta_{1}\right)\right)
\]
and 
\[
\mathbb{E}_{F}\left[\phi^{\prime}\left(v_{2}^{B,T}\left(\delta_{1},\tilde{\delta}_{2}\right)\right)\right]=R\delta_{1}\phi^{\prime}\left(v_{1}^{B,T}\left(\delta_{1}\right)\right).
\]
\end{prop}
We have the following corollary to Proposition \ref{prop:CONSTANT_PERTURBATION}.
\begin{cor}
\label{Cor:LOG-CASE}Suppose $T\geq3$. If $v\left(c\right)=\ln\left(c\right)$,
the interiority assumption in Proposition \ref{prop:CONSTANT_PERTURBATION}
holds and we have $\phi^{\prime}\left(u\left(c\right)\right)=c$.
Therefore, for $t\in\left\{ 1,2,\dots,T-1\right\} $, the ratio $\frac{\mathbb{E}_{F}\left[c_{t+1}\left(\delta_{1},\tilde{H}_{2}^{t+1}\right)\right]}{\mathbb{E}_{F}\left[c_{t}\left(\delta_{1},\tilde{H}_{2}^{t+1}\right)\right]}$
is weakly larger for equilibrium consumption than efficient consumption.
The ratio is the same in both cases if firm and agent beliefs agree.
For $t\geq2$, the ratio is strictly larger for equilibrium consumption
than efficient consumption if the agent is strictly more optimistic
about patience than the firm.
\end{cor}
These results give a sense of how agent beliefs affect the rate of
growth, or rate of decline, of equilibrium consumption relative to
efficiency. The corollary shows that it is possible to compare the
rates of growth (or decline) of expected consumption in the logarithmic
case, with equilibrium consumption more backloaded than for the efficient
policy, strictly so if the agent is more optimistic than firms. The
reason for this result is that, when the agent is more optimistic
about future discount factors, he attaches more weight to future consumption.
In equilibrium, firms pander to the agent's beliefs by delaying consumption,
thus raising the agent's expected payoffs from the contract at the
time of contracting.

It is worth noting here that \ref{prop:BACKLOADING} and \ref{prop:CONSTANT_PERTURBATION}
represent distinct though related exercises, calling for careful interpretation.
Consider the case of $T=3$ and the case of common beliefs. Suppose
both results apply to the equilibrium contract, i.e. the equilibrium
consumption values are interior and the contract fully separates types
$\delta_{2}$, as in the case of Example \ref{exa:Squareroot} for
large enough $N$. Then, from the arguments surrounding Proposition
\ref{prop:BACKLOADING} (see Equation (\ref{eq:BACKLOADED_N}), we
have
\begin{equation}
R\sum_{n=1}^{N}p_{n}\delta^{\left(n\right)}\phi^{\prime}\left(v_{2}^{E,3}\left(\delta_{1},\delta^{\left(n\right)}\right)\right)<\sum_{n=1}^{N}p_{n}\phi^{\prime}\left(v_{3}^{E,3}\left(\delta_{1},\delta^{\left(n\right)}\right)\right).\label{eq:BACKLOADING_AVERAGE_3_PERIODS}
\end{equation}
From Proposition \ref{prop:CONSTANT_PERTURBATION}, we have
\begin{equation}
R\left(\sum_{i=1}^{N}p_{i}\delta^{\left(i\right)}\right)\sum_{n=1}^{N}p_{n}\phi^{\prime}\left(v_{2}^{E,3}\left(\delta_{1},\delta^{\left(n\right)}\right)\right)=\sum_{n=1}^{N}p_{n}\phi^{\prime}\left(v_{3}^{E,3}\left(\delta_{1},\delta^{\left(n\right)}\right)\right).\label{eq:BALANCED_AVERAGE_3_PERIODS}
\end{equation}
The difference in (\ref{eq:BACKLOADING_AVERAGE_3_PERIODS}) and (\ref{eq:BALANCED_AVERAGE_3_PERIODS})
is explained by the fact that $v_{2}^{E,3}\left(\delta_{1},\delta^{\left(n\right)}\right)$
is decreasing in $\delta^{\left(n\right)}$. Thus, in Proposition
\ref{prop:BACKLOADING}, we find that consumption is inefficiently
backloaded to 3 from date 2, while in Proposition \ref{prop:CONSTANT_PERTURBATION}
we find that the ratio of the firm's expected marginal cost of date-3
utility to the expected marginal cost of date-2 utility is the same
as for the efficient policy. These findings are not contradictory.

We now further interpret Equation (\ref{eq:BALANCED_AVERAGE_3_PERIODS}),
and because the point is more general, consider now the case with
$T\ge3$. Suppose still that the agent has correct beliefs. Then,
for $t\in\left\{ 2,\dots,T-1\right\} $,
\[
\mathbb{E}_{F}\left[\phi^{\prime}\left(v_{t+1}^{E,T}\left(\delta_{1},H_{2}^{t-1},\tilde{\delta}_{t},\tilde{\delta}_{t+1}\right)\right)\right]=R\mathbb{E}_{F}\left[\tilde{\delta}_{t}\right]\mathbb{E}_{F}\left[\phi^{\prime}\left(v_{t}^{E,T}\left(\delta_{1},H_{2}^{t-1},\tilde{\delta}_{t}\right)\right)\right].
\]
Consider changing date-$t+1$ utility by $\varepsilon$ and date-$t$
utility by $-\varepsilon\mathbb{E}_{F}\left[\tilde{\delta}_{t}\right]$,
following history $\left(\delta_{1},H_{2}^{t-1}\right)$. This keeps
agent expected payoffs unchanged conditional on history $\left(\delta_{1},H_{2}^{t-1}\right)$.
The corresponding change in firm payments conditional on history $\left(\delta_{1},H_{2}^{t-1}\right)$
is
\begin{align*}
 & \mathbb{E}_{F}\left[\phi\left(v_{t}^{E,T}\left(\delta_{1},H_{2}^{t-1},\tilde{\delta}_{t}\right)-\varepsilon\mathbb{E}_{F}\left[\tilde{\delta}_{t}\right]\right)\right]+\mathbb{E}_{F}\left[\phi\left(v_{t+1}^{E,T}\left(\delta_{1},H_{2}^{t-1},\tilde{\delta}_{t},\tilde{\delta}_{t+1}\right)+\varepsilon\right)\right]\\
- & \left(\mathbb{E}_{F}\left[\phi\left(v_{t}^{E,T}\left(\delta_{1},H_{2}^{t-1},\tilde{\delta}_{t}\right)\right)\right]+\mathbb{E}_{F}\left[\phi\left(v_{t+1}^{E,T}\left(\delta_{1},H_{2}^{t-1},\tilde{\delta}_{t},\tilde{\delta}_{t+1}\right)\right)\right]\right).
\end{align*}
The expression is convex in $\varepsilon$, showing that the cost-minimising
choice of $\varepsilon$ over feasible values (i.e., such that the
resulting utilities are interior) is $\varepsilon=0$. It is then
immediate that, considering shifting utilities at date $t$ and date-$t+1$
after history $\left(\delta_{1},H_{2}^{t-1}\right)$ by a constant
independent of $\left(\tilde{\delta}_{t},\tilde{\delta}_{t+1}\right)$
cannot raise the agent's expected payoff without increasing the expected
cost to the firm.

Referring back to Equations (\ref{eq:BACKLOADING_AVERAGE_3_PERIODS})
and (\ref{eq:BALANCED_AVERAGE_3_PERIODS}) for the case of $T=3$
and $t=2$, the above implies that, in case solutions are interior
and under full separation, utilities are too backloaded state by state
(i.e., for each value of $\delta_{2}$), but neither too backloaded
nor too frontloaded if we instead consider shifting per-period utilities
by a constant independent of the state. In the above exercise, taking
$\varepsilon<0$ and hence shifting utility earlier, we see that the
change results in a compression of expected discounted payoffs and
a relative worsening of the terms of higher types in terms of the
NPV of firm expenditure. This helps to explains why shifting utility
by a constant yields a different effect compared to shifting utility
state by state.

Away from the case where agent beliefs are correct, suppose that the
agent is more optimistic than the principal. Then, for $T\geq3$ and
$t\in\left\{ 2,\dots,T-1\right\} $,
\[
\mathbb{E}_{F}\left[\phi^{\prime}\left(v_{t+1}^{E,T}\left(\delta_{1},H_{2}^{t-1},\tilde{\delta}_{t},\tilde{\delta}_{t+1}\right)\right)\right]>R\mathbb{E}_{F}\left[\tilde{\delta}_{t}\right]\mathbb{E}_{F}\left[\phi^{\prime}\left(v_{t}^{E,T}\left(\delta_{1},H_{2}^{t-1},\tilde{\delta}_{t}\right)\right)\right].
\]
It then follows that decreasing utilities by a constant at date $t+1$
and increasing utilities by a constant at date $t$, following history
$\left(\delta_{1},H_{2}^{t-1}\right)$, such that agent expected payoffs
under correct beliefs are unchanged, decreases the NPV of firm expected
payments. Similarly, shifting utility earlier in this way permits
an increase in agent expected utility under correct beliefs that holds
firm expected profits constant.

\section{Conclusions}

A body of work in behavioural contract theory studies present-biased
agents by specifying models with quasi-hyperbolic discounting. This
paper suggests an alternative framework based on stochastic discount
factors that could be a useful complement to the existing literature.
Given that the fully naive agent with quasi-hyperbolic discounting
arises as a limiting case of our model, our framework can be viewed
as an extension (in a particular direction) of the workhorse model
in the literature. This extension includes an agent with correct beliefs.
It is then of interest to understand which intuitions and logic from
models of quasi-hyperbolic discounting with naive agents survive as
we move towards the case of correct beliefs.

In the present paper we have applied our alternative framework to
a model of competitive credit markets, following closely the work
of HK and GZ. In Section \ref{sec:Degenerate_Impatience}, we studied
a setting with only two realisations of the discount factor, where
firms know that the low value always occurs. This specialisation highlights
that a similar logic to the existing literature must apply, although
the agent believes the discount factor is stochastic. While the forces
shaping equilibrium contracts are essentially the same, a natural
choice of welfare benchmark implies that consumption is our model
is inefficiently delayed. In Section \ref{sec:RICHER}, we consider
a setting which allows more values of the discount factor and assumes
full support over these values, nesting the case of correct agent
beliefs. We then show that the agent being penalised for choosing
options with earlier consumption is a robust feature of equilibrium
mechanisms, including in the case where the agent has correct beliefs.
We also provide results showing how equilibrium consumption is often
inefficiently backloaded -- we show that this holds in a three-period
model when the contract involves the appropriate separation of types.
Separation is verified in a natural special case with correct beliefs.
While contractual distortions do arise in the case with correct agent
beliefs, regulatory intervention cannot improve agent welfare. In
this case, the appropriate beliefs with which to measure welfare are
unambiguous, and it is not possible to raise agent expected payoffs
without violating either firm participation or incentive-compatibility
constraints. From the perspective of a regulator that does not know
the true model, observing that credit contracts penalise consumers
for early consumption (or late repayment) does not imply that regulatory
intervention can improve consumer welfare.

\appendix

\section*{Appendix:\ \ Omitted proofs}

\subsection*{Proofs of results in Section \ref{sec:Degenerate_Impatience}}

\begin{proof}[Proof of Lemma~\ref{lem:IC_Section3}]

We begin by observing that in an incentive-compatible mechanism, for
any $2\leq t\leq T-1$, and any $H^{t-1}\in{\cal H}^{R}$,
\begin{equation}
v_{t}\left(H^{t-1},\delta^{(1)}\right)\geq v_{t}\left(H^{t-1},\delta^{(2)}\right)\label{eq:CURRENT-INEQUALITY}
\end{equation}
 and 
\begin{equation}
\mathbb{E}_{A}\left[\sum_{\tau=t+1}^{T}\Pi_{s=t+1}^{\tau-1}\tilde{\delta}_{s}v_{\tau}\left(H^{t-1},\delta^{(2)},\tilde{H}_{t+1}^{\tau}\right)\right]\geq\mathbb{E}_{A}\left[\sum_{\tau=t+1}^{T}\Pi_{s=t+1}^{\tau-1}\tilde{\delta}_{s}v_{\tau}\left(H^{t-1},\delta^{(1)},\tilde{H}_{t+1}^{\tau}\right)\right],\label{eq:LATER-INEQUALITY}
\end{equation}
with both inequalities strict if and only if at least one of (\ref{eq:HIGH-IC})
and (\ref{eq:LOW-IC-NEW}) hold as a strict inequality. To see this,
first we can write (\ref{eq:LOW-IC-NEW}) as 
\begin{equation}
v_{t}\left(H^{t-1},\delta^{(1)}\right)-v_{t}\left(H^{t-1},\delta^{(2)}\right)\geq\delta^{(1)}\left(\begin{array}{c}
\mathbb{E}_{A}\left[\sum_{\tau=t+1}^{T}\Pi_{s=t+1}^{\tau-1}\tilde{\delta}_{s}v_{\tau}\left(H^{t-1},\delta^{(2)},\tilde{H}_{t+1}^{\tau}\right)\right]\\
-\mathbb{E}_{A}\left[\sum_{\tau=t+1}^{T}\Pi_{s=t+1}^{\tau-1}\tilde{\delta}_{s}v_{\tau}\left(H^{t-1},\delta^{(1)},\tilde{H}_{t+1}^{\tau}\right)\right]
\end{array}\right).\label{eq:LOW-IC-RESTATED}
\end{equation}
If the inequality (\ref{eq:LATER-INEQUALITY}) fails to hold, then
\begin{equation}
v_{t}\left(H^{t-1},\delta^{(1)}\right)-v_{t}\left(H^{t-1},\delta^{(2)}\right)>\delta^{(2)}\left(\begin{array}{c}
\mathbb{E}_{A}\left[\sum_{\tau=t+1}^{T}\Pi_{s=t+1}^{\tau-1}\tilde{\delta}_{s}v_{\tau}\left(H^{t-1},\delta^{(2)},\tilde{H}_{t+1}^{\tau}\right)\right]\\
-\mathbb{E}_{A}\left[\sum_{\tau=t+1}^{T}\Pi_{s=t+1}^{\tau-1}\tilde{\delta}_{s}v_{\tau}\left(H^{t-1},\delta^{(1)},\tilde{H}_{t+1}^{\tau}\right)\right]
\end{array}\right)\label{eq:HIGH-IC-VIOLATED}
\end{equation}
and so (\ref{eq:HIGH-IC}) is violated. So we can conclude (\ref{eq:LATER-INEQUALITY})
is satisfied. Then if (\ref{eq:CURRENT-INEQUALITY}) fails to hold,
both types strictly prefer to report $\delta^{\left(2\right)}$.

Now, note that, when both of (\ref{eq:CURRENT-INEQUALITY}) and (\ref{eq:LATER-INEQUALITY})
hold as equality, both types are indifferent across both reports;
i.e., both incentive constraints (\ref{eq:HIGH-IC}) and (\ref{eq:LOW-IC-NEW})
hold as equality. It is easy to see we cannot have exactly one of
these two inequalities hold as equalities, and if both are strict
inequalities then at least one of the two incentive constraints is
strict. (For instance, if (\ref{eq:LOW-IC-NEW}) holds as equality
and hence equivalently (\ref{eq:LOW-IC-RESTATED}) holds as equality,
then the reverse strict inequality as in inequality (\ref{eq:HIGH-IC-VIOLATED})
holds, i.e. (\ref{eq:HIGH-IC}) holds strictly.)

Suppose now that (\ref{eq:LOW-IC-NEW}) holds as a strict inequality
at date $t$ s.t. $2\leq t\leq T-1$, and history $H^{t-1}\in{\cal H}^{R}$,
with the difference between left and right sides equal to $x$. Increase
$v_{t}\left(H^{t-1},\delta^{(2)}\right)$ by $x\left(1-q_{2}\right)$
and decrease $v_{t}\left(H^{t-1},\delta^{(1)}\right)$ by $xq_{2}$.
Now, (\ref{eq:LOW-IC-NEW}) holds as an equality and since (\ref{eq:LATER-INEQUALITY})
holds strictly, we have (\ref{eq:HIGH-IC}) holds strictly. Moreover,
for any period $t'<t$, the expected agent continuation payoff from
$t'+1$ onwards conditional on having reported $H^{t'}\in{\cal H}^{R}$,
as determined by 
\[
\mathbb{E}_{A}\left[\sum_{\tau=t'+1}^{T}\Pi_{s=t'+1}^{\tau-1}\tilde{\delta}_{s}v_{\tau}\left(H^{t'},\tilde{H}_{t'+1}^{\tau}\right)\right],
\]
remains unchanged. Hence all incentive constraints at dates before
$t$ are unaffected, and the same is clearly true for incentive constraints
after $t$. Thus, it is easy to see that all incentive constraints
remain intact.

Finally, note that the agent's expected payoff has not changed and
firm profits are either unchanged or increase if $H^{t-1}$ is the
sequence comprising only $\delta^{(1)}$. The change can be implemented
at every $t$, $2\leq t\leq T-1$, and $H^{t-1}$ such that (\ref{eq:LOW-IC-NEW})
holds as a strict inequality.

\end{proof}

\bigskip{}

\begin{proof}[Proof of Lemma~\ref{lem:Objective_Naive}]

Suppose now that all incentive constraints (\ref{eq:LOW-IC-NEW})
hold with equality. By substituting this equality for $t=2$ into
the agent's expected payoff, we obtain
\begin{align*}
 & v_{1}\left(\delta^{(1)}\right)+\delta^{(1)}\mathbb{E}_{A}\left[v_{2}\left(\delta^{(1)},\tilde{\delta}_{2}\right)+\sum_{t=3}^{T}\Pi_{s=2}^{t-1}\tilde{\delta}_{s}v_{t}\left(\delta^{(1)},\tilde{\delta}_{2},\dots,\tilde{\delta}_{t}\right)\right]\\
= & v_{1}\left(\delta^{(1)}\right)-q_{2}\left(\delta^{(2)}-\delta^{(1)}\right)v_{2}\left(\delta^{(1)},\delta^{(2)}\right)+q_{2}\delta^{(2)}\left[\begin{array}{c}
v_{2}\left(\delta^{(1)},\delta^{(2)}\right)\\
+\delta^{(1)}\mathbb{E}_{A}\left[\sum_{t=3}^{T}\Pi_{s=3}^{t-1}\tilde{\delta}_{s}u\left(c_{t}\left(\delta^{(1)},\delta^{(2)},\tilde{\delta}_{3},\dots,\tilde{\delta}_{t}\right)\right)\right]
\end{array}\right]\\
 & +q_{1}\delta^{(1)}\left(v_{2}\left(\delta^{(1)},\delta^{(1)}\right)+\delta^{(1)}\mathbb{E}_{A}\left[\sum_{t=3}^{T}\Pi_{s=3}^{t-1}\tilde{\delta}_{s}v_{t}\left(\delta^{(1)},\delta^{(1)},\tilde{\delta}_{3},\dots,\tilde{\delta}_{t}\right)\right]\right)\\
= & v_{1}\left(\delta^{(1)}\right)-q_{2}\left(\delta^{(2)}-\delta^{(1)}\right)v_{2}\left(\delta^{(1)},\delta^{(2)}\right)+q_{2}\delta^{(2)}\left[\begin{array}{c}
v_{2}\left(\delta^{(1)},\delta^{(1)}\right)\\
+\delta^{(1)}\mathbb{E}_{A}\left[\sum_{t=3}^{T}\Pi_{s=3}^{t-1}\tilde{\delta}_{s}u\left(c_{t}\left(\delta^{(1)},\delta^{(1)},\tilde{\delta}_{3},\dots,\tilde{\delta}_{t}\right)\right)\right]
\end{array}\right]\\
 & +q_{1}\delta^{(1)}\left(v_{2}\left(\delta^{(1)},\delta^{(1)}\right)+\delta^{(1)}\mathbb{E}_{A}\left[\sum_{t=3}^{T}\Pi_{s=3}^{t-1}\tilde{\delta}_{s}v_{t}\left(\delta^{(1)},\delta^{(1)},\tilde{\delta}_{3},\dots,\tilde{\delta}_{t}\right)\right]\right)\\
= & v_{1}\left(\delta^{(1)}\right)-q_{2}\left(\delta^{(2)}-\delta^{(1)}\right)v_{2}\left(\delta^{(1)},\delta^{(2)}\right)+\left(q_{1}\delta^{(1)}+q_{2}\delta^{(2)}\right)v_{2}\left(\delta^{(1)},\delta^{(1)}\right)\\
 & +\left(q_{1}\delta^{(1)}+q_{2}\delta^{(2)}\right)\delta^{(1)}\mathbb{E}_{A}\left[\sum_{t=3}^{T}\Pi_{s=3}^{t-1}\tilde{\delta}_{s}v_{t}\left(\delta^{(1)},\delta^{(1)},\tilde{\delta}_{3},\dots,\tilde{\delta}_{t}\right)\right].
\end{align*}
We then observe that, for $\tau\in\left\{ 3,4,\dots,T-1\right\} $,
\begin{align*}
 & \delta^{(1)}\mathbb{E}_{A}\left[\sum_{t=\tau}^{T}\Pi_{s=\tau}^{t-1}\tilde{\delta}_{s}v_{t}\left(L^{\tau-1},\tilde{\delta}_{\tau},\dots,\tilde{\delta}_{t}\right)\right]\\
= & q_{2}\delta^{(2)}\left(v_{\tau}\left(L^{\tau-1},\delta^{(2)}\right)+\delta^{(1)}\mathbb{E}_{A}\left[\sum_{t=\tau+1}^{T}\Pi_{s=\tau+1}^{t-1}\tilde{\delta}_{s}v_{t}\left(L^{\tau-1},\delta^{(2)},\tilde{\delta}_{\tau+1},\dots,\tilde{\delta}_{t}\right)\right]\right)\\
 & -\left(\delta^{(2)}-\delta^{(1)}\right)q_{2}v_{\tau}\left(L^{\tau-1},\delta^{(2)}\right)\\
 & +q_{1}\delta^{(1)}\left(v_{\tau}\left(L^{\tau}\right)+\delta^{(1)}\mathbb{E}_{A}\left[\sum_{t=\tau+1}^{T}\Pi_{s=\tau+1}^{t-1}\tilde{\delta}_{s}v_{t}\left(L^{\tau},\tilde{\delta}_{\tau+1},\dots,\tilde{\delta}_{t}\right)\right]\right)\\
= & q_{2}\delta^{(2)}\left(v_{\tau}\left(L^{\tau}\right)+\delta^{(1)}\mathbb{E}_{A}\left[\sum_{t=\tau+1}^{T}\Pi_{s=\tau+1}^{t-1}\tilde{\delta}_{s}v_{t}\left(L^{\tau},\tilde{\delta}_{\tau+1},\dots,\tilde{\delta}_{t}\right)\right]\right)\\
 & -\left(\delta^{(2)}-\delta^{(1)}\right)q_{2}v_{\tau}\left(L^{\tau-1},\delta^{(2)}\right)\\
 & +q_{1}\delta^{(1)}\left(v_{\tau}\left(L^{\tau}\right)+\delta^{(1)}\mathbb{E}_{A}\left[\sum_{t=\tau+1}^{T}\Pi_{s=\tau+1}^{t-1}\tilde{\delta}_{s}v_{t}\left(L^{\tau},\tilde{\delta}_{\tau+1},\dots,\tilde{\delta}_{t}\right)\right]\right)\\
= & \left(q_{1}\delta^{(1)}+q_{2}\delta^{(2)}\right)\left(v_{\tau}\left(L^{\tau}\right)+\delta^{(1)}\mathbb{E}_{A}\left[\sum_{t=\tau+1}^{T}\Pi_{s=\tau+1}^{t-1}\tilde{\delta}_{s}v_{t}\left(L^{\tau},\tilde{\delta}_{\tau+1},\dots,\tilde{\delta}_{t}\right)\right]\right)\\
 & -\left(\delta^{(2)}-\delta^{(1)}\right)q_{2}v_{\tau}\left(L^{\tau-1},\delta^{(2)}\right).
\end{align*}
Substituting iteratively, we find that 
\begin{align}
 & v_{1}\left(\delta^{(1)}\right)+\delta^{(1)}\mathbb{E}_{A}\left[v_{2}\left(\delta^{(1)},\tilde{\delta}_{2}\right)+\sum_{t=3}^{T}\Pi_{s=2}^{t-1}\tilde{\delta}_{s}v_{t}\left(\delta^{(1)},\tilde{\delta}_{2},\dots,\tilde{\delta}_{t}\right)\right]\nonumber \\
= & \sum_{t=1}^{T-1}\left(q_{1}\delta^{(1)}+q_{2}\delta^{(2)}\right)^{t-1}v_{t}\left(L^{t}\right)+\left(q_{1}\delta^{(1)}+q_{2}\delta^{(2)}\right)^{T-2}\delta^{(1)}v_{T}\left(L^{T}\right)\nonumber \\
 & -\sum_{t=2}^{T-1}\left(q_{1}\delta^{(1)}+q_{2}\delta^{(2)}\right)^{t-2}q_{2}\left(\delta^{(2)}-\delta^{(1)}\right)v_{t}\left(L^{t-1},\delta^{(2)}\right).\label{eq:OBJECTIVE_REWRITTEN}
\end{align}

Now, consider a relaxed program in which we maximize (\ref{eq:OBJECTIVE_REWRITTEN})
subject only to the constraints that the firm makes non-negative profit
and that utilities are no less than $u\left(0\right)$. First note
that, for any $t\in\left\{ 2,\dots,T-1\right\} $, the obective is
decreasing in $v_{t}\left(L^{t-1},\delta^{(2)}\right)$, so these
must be equal to $u\left(0\right)=0$. Second, because $T$ is finite,
satisfaction of the non-negative profit constraint requires that utilities
are uniformly bounded across histories $L^{t}$. Hence, a maximum
in the problem exists. It is easy to see that the firm makes zero
profit. Because these will correspond to the equilibrium utilities,
we index the optimal values with a superscript $E$. In particular,
the utilities maximising (\ref{eq:OBJECTIVE_REWRITTEN}) are given
by $v_{t}^{E,T}\left(L^{t}\right)$, $t=1,\dots,T$, and $v_{t}^{E,T}\left(L^{t-1},\delta^{(2)}\right)$,
$t=2,\dots,T-1$. Using strict convexity of $\phi$, these values
are uniquely determined. 

Now we can determine an incentive-compatible mechanism with the utilities
$\hat{v}_{t}\left(H^{t}\right)$, $H^{t}\in{\cal H}^{R}$. For $t=1,\dots,T$,
$\hat{v}_{t}\left(L^{t}\right)=v_{t}^{E,T}\left(L^{t}\right)$. For
$t=2,\dots,T-1$, let $\hat{v}_{t}\left(L^{t-1},\delta^{(2)}\right)=v_{t}^{E,T}\left(L^{t-1},\delta^{(2)}\right)=u\left(0\right)$.
We complete the specification by setting, for any period $t\in\left\{ 2,\dots,T-1\right\} $
and history such that $H^{t}=\left(L^{t-1},\delta^{(2)}\right)$,
all subsequent utilities to be constant; i.e. for all $H_{t+1}^{\tau}$
with $\tau\geq t+1$, let $\hat{v}_{\tau}\left(L^{t-1},\delta^{(2)},H_{t+1}^{\tau}\right)=\bar{v}\left(L^{t-1},\delta^{(2)}\right)$.
Then, all incentive constraints occurring after the history $\left(L^{t-1},\delta^{(2)}\right)$,
i.e. at date $t+1$ or any later period, are satisfied trivially.
We only need to ensure incentive constraints occuring at $t\in\left\{ 2,\dots,T-1\right\} $
at histories $H^{t-1}=L^{t-1}$, and we do so such that the constraints
(\ref{eq:LOW-IC-NEW}) hold with equality. This can be achieved iteratively
working backwards. For $t=T-1$, we choose $\hat{v}_{T}\left(L^{T-2},\delta^{(2)}\right)=\bar{v}\left(L^{T-2},\delta^{(2)}\right)$,
with $\bar{v}\left(L^{T-2},\delta^{(2)}\right)$ satisfying
\[
v_{T-1}^{E,T}\left(L^{T-1}\right)+\delta^{(1)}v_{T}^{E,T}\left(L^{T-1}\right)=u\left(0\right)+\delta^{(1)}\bar{v}\left(L^{T-2},\delta^{(2)}\right).
\]
So we have (\ref{eq:LOW-IC-NEW}) holds with equality at $H^{T-2}=L^{T-2}$,
while (\ref{eq:HIGH-IC}) is satisfied because $\hat{v}_{T}\left(L^{T-1}\right)=v_{T}^{E,T}\left(L^{T-1}\right)\leq\bar{v}\left(L^{T-2},\delta^{(2)}\right)=\hat{v}_{T}\left(L^{T-2},\delta^{(2)}\right)$.

Now, for any $t\in\left\{ 2,\dots,T-2\right\} $, if $\hat{v}_{\tau}\left(L^{t},H_{t+1}^{\tau}\right)$,
all $H_{t+1}^{\tau}\in\Delta^{\tau-t}$, have been determined for
$\tau>t$, specify $\bar{v}\left(L^{t-1},\delta^{(2)}\right)$ as
solving
\[
v_{t}^{E,T}\left(L^{t}\right)+\delta^{(1)}\mathbb{E}_{A}\left[\sum_{\tau^{\prime}=t+1}^{T}\Pi_{s=t+1}^{\tau^{\prime}-1}\tilde{\delta}_{s}\hat{v}_{\tau^{\prime}}\left(L^{t},\tilde{H}_{t+1}^{\tau^{\prime}}\right)\right]=u\left(0\right)+\delta^{(1)}\sum_{\tau^{\prime}=t+1}^{T}\left(q_{1}\delta^{(1)}+q_{2}\delta^{(2)}\right)^{\tau^{\prime}-\left(t+1\right)}\bar{v}\left(L^{t-1},\delta^{(2)}\right).
\]
Thus, we have specified, for $\tau>t$, all $H_{t+1}^{\tau}\in\Delta^{\tau-t}$,
$\hat{v}_{\tau}\left(L^{t-1},\delta^{(2)},H_{t+1}^{\tau}\right)=\bar{v}\left(L^{t-1},\delta^{(2)}\right)$.
And hence, $\hat{v}_{\tau}\left(L^{t-1},H_{t}^{\tau}\right)$ is determined
for all $\tau>t-1$, all $H_{t}^{\tau}\in\Delta^{\tau-t+1}$. Thus,
proceeding iteratively backwards, we specify the entire mechanism.
The mechanism is constructed to ensure constraint (\ref{eq:LOW-IC-NEW})
holds with equality at histories $L^{t-1}$, $2\leq t\leq T-1$. That
(\ref{eq:HIGH-IC}) is then satisfied follows because at these histories,
\begin{align*}
\mathbb{E}_{A}\left[\sum_{\tau=t+1}^{T}\Pi_{s=t+1}^{\tau-1}\tilde{\delta}_{s}\hat{v}_{\tau}\left(L^{t},\tilde{H}_{t+1}^{\tau}\right)\right] & \leq\mathbb{E}_{A}\left[\sum_{\tau=t+1}^{T}\Pi_{s=t+1}^{\tau-1}\tilde{\delta}_{s}\hat{v}_{\tau}\left(L^{t-1},\delta^{(2)},\tilde{H}_{t+1}^{\tau}\right)\right]\\
 & =\sum_{\tau=t+1}^{T}\left(q_{1}\delta^{(1)}+q_{2}\delta^{(2)}\right)^{\tau-\left(t+1\right)}\bar{v}\left(L^{t-1},\delta^{(2)}\right).
\end{align*}
This completes the proof.

\end{proof}

\bigskip{}

\begin{proof}[Proof of Proposition~\ref{Prop:AVERAGE_DF}]

The argument exactly follows Citanna et al. (2023), reproduced here
for convenience. Consider the programme:
\[
V_{T}\left(\beta,I_{T}\right)=\max_{\left\{ v_{t}\right\} _{t=1}^{T}}\left\{ \sum_{t=1}^{T-1}\left(q_{1}\delta^{(1)}+q_{2}\delta^{(2)}\right)^{t-1}v_{t}+\beta\left(q_{1}\delta^{(1)}+q_{2}\delta^{(2)}\right)^{T-2}v_{T}\right\} 
\]
subject to 
\[
\sum_{t=1}^{T}\frac{\phi\left(v_{t}\right)}{R^{t-1}}\leq I_{T}.
\]
By Lemma \ref{lem:Objective_Naive}, $V_{T}\left(\delta^{(1)},I_{T}\right)$
is the agent's equilibrium payoff in the $T$-period version of our
original model when discounted income is $I_{T}$. Then, because the
benchmark-discounting problem is well-posed, $\left(V_{T}\left(q_{1}\delta^{(1)}+q_{2}\delta^{(2)},I_{T}\right)\right)_{T=3}^{\infty}$
is bounded from above and it is a non-decreasing sequence. Hence it
has a finite limit and so the sequence is Cauchy:
\[
\lim_{T\rightarrow\infty}\left(V_{T}\left(q_{1}\delta^{(1)}+q_{2}\delta^{(2)},I_{T}\right)-V_{T-1}\left(q_{1}\delta^{(1)}+q_{2}\delta^{(2)},I_{T-1}\right)\right)=0.
\]

For any $T\geq3$, we have 
\[
W_{T}^{A}-W_{T}^{E}=V_{T}\left(\left(q_{1}\delta^{(1)}+q_{2}\delta^{(2)}\right),I_{T}\right)-\sum_{t=1}^{T}\left(q_{1}\delta^{(1)}+q_{2}\delta^{(2)}\right)^{t-1}v_{t}^{E,T}\left(L^{t}\right)\geq0.
\]
But then 
\[
\sum_{t=1}^{T}\left(q_{1}\delta^{(1)}+q_{2}\delta^{(2)}\right)^{t-1}v_{t}^{E,T}\left(L^{t}\right)\geq V_{T}\left(\delta^{(1)},I_{T}\right)\geq V_{T-1}\left(\left(q_{1}\delta^{(1)}+q_{2}\delta^{(2)}\right),I_{T-1}\right).
\]
This implies that 
\begin{align*}
V_{T}\left(\left(q_{1}\delta^{(1)}+q_{2}\delta^{(2)}\right),I_{T}\right)-V_{T-1}\left(\left(q_{1}\delta^{(1)}+q_{2}\delta^{(2)}\right),I_{T-1}\right) & \geq W_{T}^{A}-W_{T}^{E}\ge0
\end{align*}
as desired.

\end{proof}

\bigskip{}

\begin{proof}[Proof of Proposition~\ref{Prop:More_Backloaded}]

The proof of Lemma \ref{lem:Objective_Naive} follows from an analysis
of Karush-Kuhn-Tucker (KKT) conditions. First consider the programme
determining the equilibrium utilities $\left(v_{t}^{E,T}\left(L^{t}\right)\right)_{t=1}^{T}$.
The constraints are the zero lower bounds on utility and the firm's
break-even condition. That KKT conditions are valid for the optimum
$\left(v_{t}^{E,T}\left(L^{t}\right)\right)_{t=1}^{T}$ follows from
a standard constraint qualification. For instance, noting that $I_{T}>0$
and hence $\left(v_{t}^{E,T}\left(L^{t}\right)\right)_{t=1}^{T}$
is not degenerately zero, the gradients of the active constraints
evaluated at $\left(v_{t}^{E,T}\left(L^{t}\right)\right)_{t=1}^{T}$
are linearly independent.

The Lagrangean is
\begin{align*}
{\cal L}^{E,T}= & \sum_{t=1}^{T-1}\left(q_{1}\delta^{(1)}+q_{2}\delta^{(2)}\right)^{t-1}v_{t}\left(L^{t}\right)+\left(q_{1}\delta^{(1)}+q_{2}\delta^{(2)}\right)^{T-2}\delta^{(1)}v_{T}\left(L^{T}\right)\\
 & +\lambda\left(I_{T}-\sum_{t=1}^{T}\frac{\phi\left(v_{t}\left(L^{t}\right)\right)}{R^{t-1}}\right)+\sum_{t=1}^{T}\mu_{t}v_{t}\left(L^{t}\right).
\end{align*}
Here, $\left(v_{t}\left(L^{t}\right)\right)_{t=1}^{T}$ are the utilities
to be chosen, $\lambda$ represents the multiplier on the firm's break-even
condition, and $\left(\mu_{t}\right)_{t=1}^{T}$ are the multipliers
on the non-negative utility constraints. Necessary conditions for
an optimum are: (i) for $t\leq T-1$,
\[
\left(q_{1}\delta^{(1)}+q_{2}\delta^{(2)}\right)^{t-1}+\mu_{t}=\frac{\lambda}{R^{t-1}}\phi'\left(v_{t}\left(L^{t}\right)\right),
\]
(ii)
\[
\left(q_{1}\delta^{(1)}+q_{2}\delta^{(2)}\right)^{T-2}\delta^{(1)}+\mu_{T}=\frac{\lambda}{R^{T-1}}\phi'\left(v_{T}\left(L^{T}\right)\right),
\]
(iii)
\[
\lambda\left(I_{T}-\sum_{t=1}^{T}\frac{\phi\left(v_{t}\left(L^{t}\right)\right)}{R^{t-1}}\right)=0,
\]
and (iv) for $t\in\left\{ 1,2,\dots,T\right\} $,
\[
\mu_{t}v_{t}\left(L^{t}\right)=0.
\]
In addition, the constraints must be satisfied ($I_{T}-\sum_{t=1}^{T}\frac{\phi\left(v_{t}\left(L^{t}\right)\right)}{R^{t-1}}\geq0$
and $v_{t}\left(L^{t}\right)\geq u\left(0\right)$ for all $t$) and
multipliers must be non-negative. We let $\left(\left(v_{t}^{E,T}\left(L^{t}\right)\right)_{t=1}^{T},\lambda^{E,T},\left(\mu_{t}^{E,T}\right)_{t=1}^{T}\right)$
denote a KKT solution and note that, by convexity of $\phi$, $\left(v_{t}^{E,T}\left(L^{t}\right)\right)_{t=1}^{T}$
maximises the Lagrangean given the multipliers $\left(\lambda^{E,T},\left(\mu_{t}^{E,T}\right)_{t=1}^{T}\right)$.
Note that (i) and (ii) above imply $\lambda^{E,T}>0$.

Now, let $\rho$ be the inverse of $\phi^{\prime}$. We can write,
for $t\leq T-2$,
\begin{equation}
v_{t}^{E,T}\left(L^{t}\right)=\begin{cases}
\rho\left(\frac{\left(q_{1}\delta^{(1)}+q_{2}\delta^{(2)}\right)^{t-1}R^{t-1}}{\lambda^{E,T}}\right) & \text{if \ensuremath{\frac{\left(q_{1}\delta^{(1)}+q_{2}\delta^{(2)}\right)^{t-1}R^{t-1}}{\lambda^{E,T}}}\ensuremath{\ensuremath{\geq\phi_{+}^{\prime}\left(0\right)}}}\\
0 & \text{otherwise.}
\end{cases}\label{eq:E_Optimality}
\end{equation}
Also, 
\begin{equation}
v_{T}^{E,T}\left(L^{T}\right)=\begin{cases}
\rho\left(\frac{\left(q_{1}\delta^{(1)}+q_{2}\delta^{(2)}\right)^{T-2}\delta^{(1)}R^{T-1}}{\lambda^{E,T}}\right) & \text{if \ensuremath{\frac{\left(q_{1}\delta^{(1)}+q_{2}\delta^{(2)}\right)^{T-2}\delta^{(1)}R^{T-1}}{\lambda^{E,T}}}\ensuremath{\ensuremath{\geq\phi_{+}^{\prime}\left(0\right)}}}\\
0 & \text{otherwise.}
\end{cases}\label{eq:E_Optimality_last}
\end{equation}

Now consider the efficient policy, which recall solves the problem
of maximising the agent's utility subject to the firm breaking even
when the agent's discount factor is $\delta^{(1)}$. The problem and
analysis is the same as above if we set $q_{1}=1$ and $q_{2}=0$.
A corresponding KKT solution is denoted $\left(\left(v_{t}^{B,T}\left(L^{t}\right)\right)_{t=1}^{T},\lambda^{B,T},\left(\mu_{t}^{B,T}\right)_{t=1}^{T}\right)$,
where $\lambda^{B,T}>0$. The efficient policy is thus given, for
$1\leq t\leq T$,
\begin{equation}
v_{t}^{B,T}\left(L^{t}\right)=\begin{cases}
\rho\left(\frac{\left(\delta^{(1)}\right)^{t-1}R^{t-1}}{\lambda^{B,T}}\right) & \text{if \ensuremath{\frac{\left(\delta^{(1)}\right)^{t-1}R^{t-1}}{\lambda^{B,T}}}\ensuremath{\ensuremath{\geq\phi}'\ensuremath{\left(0\right)}}}\\
0 & \text{otherwise}.
\end{cases}\label{eq:B_Optimality}
\end{equation}

Now let us prove the claims in the proposition. Because the firm's
break-even condition is the same across the two cases, and because
$\rho$ is strictly increasing, we must have $\lambda^{B,T}\leq\lambda^{E,T}$
and thus $v_{1}^{B,T}\left(\delta^{(1)}\right)\geq v_{1}^{E,T}\left(\delta^{(1)}\right)$.

Consider next $t\leq T-2$, and suppose $v_{t}^{E,T}\left(L^{t}\right)\geq v_{t}^{B,T}\left(L^{t}\right)$
and $v_{t}^{E,T}\left(L^{t}\right)>0$. Then
\[
\frac{\left(q_{1}\delta^{(1)}+q_{2}\delta^{(2)}\right)^{t-1}R^{t-1}}{\lambda^{E,T}}\geq\frac{\left(\delta^{(1)}\right)^{t-1}R^{t-1}}{\lambda^{B,T}}
\]
 for $t+1\leq s\leq T-1$
\[
\frac{\left(q_{1}\delta^{(1)}+q_{2}\delta^{(2)}\right)^{s-1}R^{s-1}}{\lambda^{E,T}}>\frac{\left(\delta^{(1)}\right)^{s-1}R^{s-1}}{\lambda^{B,T}}
\]
and 
\[
\frac{\left(q_{1}\delta^{(1)}+q_{2}\delta^{(2)}\right)^{T-2}\delta^{(1)}R^{T-1}}{\lambda^{E,T}}>\frac{\left(\delta^{(1)}\right)^{T-1}R^{T-1}}{\lambda^{B,T}}.
\]
This shows that, for all $s>t$, $v_{s}^{E,T}\left(L^{s}\right)\geq v_{s}^{B,T}\left(L^{s}\right)$;
and $v_{s}^{E,T}\left(L^{s}\right)>v_{s}^{B,T}\left(L^{s}\right)$
if $v_{s}^{E,T}\left(L^{s}\right)>0$. Considering $t=T-1$, $v_{T-1}^{E,T}\left(L^{T-1}\right)\geq v_{T-1}^{B,T}\left(L^{T-1}\right)$
and $v_{T-1}^{E,T}\left(L^{T-1}\right)>u\left(0\right)$ implies
\[
\frac{\left(q_{1}\delta^{(1)}+q_{2}\delta^{(2)}\right)^{T-2}R^{T-2}}{\lambda^{E,T}}\geq\frac{\left(\delta^{(1)}\right)^{T-2}R^{T-2}}{\lambda^{B,T}}
\]
and hence 
\[
\frac{\left(q_{1}\delta^{(1)}+q_{2}\delta^{(2)}\right)^{T-2}\delta^{(1)}R^{T-1}}{\lambda^{E,T}}\geq\frac{\left(\delta^{(1)}\right)^{T-1}R^{T-1}}{\lambda^{B,T}}.
\]
Therefore, $v_{T}^{E,T}\left(L^{T}\right)\geq v_{T}^{B,T}\left(L^{T}\right)$.

\end{proof}

\bigskip{}

\begin{proof}[Proof of Corollary~\ref{cor:Backloading}]

Note that $u_{+}^{\prime}\left(0\right)$ is a standard Inada condition
guaranteeing interiority of solutions by a simple perturbation argument
(any policy that specifies zero utility at some date can be strictly
improved by slightly increasing utility at that date and decreasing
it at some other date). If $v_{1}^{E,T}\left(L\right)\geq v_{1}^{B,T}\left(L\right)$,
then recalling $T\geq3$ and the Equations (\ref{eq:E_Optimality})-(\ref{eq:B_Optimality}),
we must have $v_{t}^{E,T}\left(L^{t}\right)>v_{t}^{B,T}\left(L^{t}\right)$
for $t>1$. Since $\left(v_{t}^{B,T}\left(L^{t}\right)\right)_{t=1}^{T}$
generates zero firm profits, $\left(v_{t}^{E,T}\left(L^{t}\right)\right)_{t=1}^{T}$
must generate negative firm profits, a contradiction. Therefore, $v_{1}^{E,T}\left(L\right)<v_{1}^{B,T}\left(L\right)$.
Proposition \ref{Prop:More_Backloaded} implies that if $v_{t^{\prime}}^{E,T}\left(L^{t^{\prime}}\right)\geq v_{t^{\prime}}^{B,T}\left(L^{t^{\prime}}\right)$
for $t^{\prime}\leq T-2$, then $v_{s}^{E,T}\left(L^{s}\right)>v_{s}^{B,T}\left(L^{s}\right)$
for all $s>t^{\prime}$, and in this case we can take $t^{*}=t^{\prime}$
in the corollary. From Equations (\ref{eq:E_Optimality})-(\ref{eq:B_Optimality})
and the fact that the firm earns zero profits under both $\left(v_{t}^{B,T}\left(L^{t}\right)\right)_{t=1}^{T}$
and $\left(v_{t}^{E,T}\left(L^{t}\right)\right)_{t=1}^{T}$, we can
see that the only other possibility is that $v_{T-1}^{E,T}\left(L^{T-1}\right)>v_{T-1}^{B,T}\left(L^{T-1}\right)$
and $v_{T}^{E,T}\left(L^{T}\right)>v_{T}^{B,T}\left(L^{T-1}\right)$
and in this case $t^{*}=T-1$.

\end{proof}\bigskip{}

\begin{proof}[Proof of Proposition~\ref{prop:LONG_RUN_INEFFICIENT}]

We begin with a preliminary observation: \emph{If $\delta^{(1)}R=1$,
then, for any $T$, $v_{t}^{B,T}\left(L^{t}\right)=w$ for all $1\leq t\leq T$.
If $\delta^{(1)}R<1$, then $v_{t}^{B,T}\left(L^{t}\right)$ converges
pointwise as $T\rightarrow\infty$ to a weakly decreasing sequence
$\left(v_{t}^{B,\infty}\left(L^{t}\right)\right)_{t=1}^{\infty}$.}

The observation for $\delta^{(1)}R=1$ is because, for each $T$,
efficient consumption is constant over time with the firm earning
zero profit. Consider then the case where $\delta^{(1)}R<1$. There
exists a $\underline{\lambda}^{B}>0$ such that $\lambda_{T}^{B}\geq\underline{\lambda}^{B}$
for every $T$-period efficient problem. In particular, we have 
\[
\phi'\left(u\left(I_{\infty}\right)\right)\geq\phi'\left(v_{1}^{B}\left(L\right)\right)=\frac{1}{\lambda^{B,T}}
\]
implies
\[
\lambda^{B,T}\geq\frac{1}{\phi'\left(u\left(I_{\infty}\right)\right)},
\]
as period 1 consumption does not exceed $I_{\infty}$.

This implies that utility in period $t$, $v_{t}^{B,T}\left(L^{t}\right)$,
for any horizon $T$, is bounded above by 
\[
\begin{cases}
\rho\left(\frac{\left(\delta^{(1)}\right)^{t-1}R^{t-1}}{\underline{\lambda}^{B}}\right) & \text{if \ensuremath{\frac{\left(\delta^{(1)}\right)^{t-1}R^{t-1}}{\underline{\lambda}^{B}}}\ensuremath{\ensuremath{\geq\phi_{+}^{\prime}\left(0\right)}}}\\
0 & \text{otherwise}.
\end{cases}
\]
There is then $\bar{T}$ such that this bound is $0$ for all $t\geq\bar{T}$;
i.e., there is no consumption after $\bar{T}$. It is easy to see
that, at each horizon length $T$, the budget constraint is satisfied
with equality, hence $\lambda^{B,T}$ is in fact uniquely determined
and must eventually be decreasing with $T$. Hence, for each $t$,
$v_{t}^{B,T}\left(L^{t}\right)$ is non-decreasing in $T$ for all
$T$ sufficiently large.

Moreover, noting that $v_{1}^{B,T}\left(L\right)>0$ for all $T$,
we have $\lambda^{B,T}=\frac{1}{\phi'\left(v_{1}^{B,T}\left(L\right)\right)}$.
As $T\rightarrow\infty$, we must have $\lambda^{B,T}\rightarrow\lambda^{B,\infty}$
for some $\lambda^{B,\infty}>0$, and utility converges pointwise
to 
\begin{equation}
v_{t}^{B,\infty}\left(L^{t}\right)=\begin{cases}
\rho\left(\frac{\left(\delta^{(1)}\right)^{t-1}R^{t-1}}{\lambda^{B,\infty}}\right) & \text{if \ensuremath{\frac{\left(\delta^{(1)}\right)^{t-1}R^{t-1}}{\lambda^{B,\infty}}}\ensuremath{\ensuremath{\geq\phi'_{+}\left(0\right)}}}\\
0 & \text{otherwise}.
\end{cases}\label{eq:LIMITING_SEQUENCE}
\end{equation}
In order that the firm earns zero profit for $\left(v_{t}^{B,T}\left(L^{t}\right)\right)_{t=1}^{T}$
for large $T$, and because $I_{T}\rightarrow I_{\infty}$ as $T\rightarrow\infty$,
we must have zero profit also in the infinite-horizon problem: i.e.,
$\sum_{t=1}^{\infty}\frac{\phi\left(v_{t}^{B,\infty}\left(L^{t}\right)\right)}{R^{t-1}}=I_{\infty}$.
Therefore, $v_{1}^{B,\infty}\left(L\right)>w$.

\medskip{}

We now prove the result in three cases, according to the value of
$\left(q_{1}\delta^{(1)}+q_{2}\delta^{(2)}\right)$.\medskip{}

\textbf{Case 1: $\left(q_{1}\delta^{(1)}+q_{2}\delta^{(2)}\right)R<1$.
}In this case, by the same argument as argued in the observation above,
the utility $v_{t}^{E,T}\left(L^{t}\right)$ converges to 
\begin{equation}
v_{t}^{E,\infty}\left(L^{t}\right)=\begin{cases}
\rho\left(\frac{\left(q_{1}\delta^{(1)}+q_{2}\delta^{(2)}\right)^{t-1}R^{t-1}}{\lambda^{E,\infty}}\right) & \text{if \ensuremath{\frac{\left(\delta^{(1)}\right)^{t-1}R^{t-1}}{\lambda^{E,\infty}}}\ensuremath{\ensuremath{\geq\phi{}_{+}^{\prime}\left(0\right)}}}\\
0 & \text{otherwise},
\end{cases}\label{eq:LIMITING_SEQUENCE-E}
\end{equation}
as $T\rightarrow\infty$, where $\lambda^{E,\infty}=\lim_{T\rightarrow\infty}\frac{1}{\phi'\left(v_{1}^{E,T}\left(L\right)\right)}>0$.
As above, utilities $\left(v_{t}^{E,\infty}\left(L^{t}\right)\right)_{t=1}^{\infty}$
generate zero profit for the firm over the infinite horizon: $\sum_{t=1}^{\infty}\frac{\phi\left(v_{t}^{E,\infty}\left(L^{t}\right)\right)}{R^{t-1}}=I_{\infty}$.

We now show it is not the case that $v_{1}^{E,\infty}\left(L\right)=u\left(I_{\infty}\right)$
by supposing that it were. Then $\lambda^{E,\infty}=\frac{1}{\phi'\left(u\left(I_{\infty}\right)\right)}$
and so
\begin{align*}
\phi^{\prime}\left(v_{2}^{E,\infty}\left(L^{2}\right)\right) & =\frac{\left(q_{1}\delta^{(1)}+q_{2}\delta^{(2)}\right)R}{\lambda^{E,\infty}}\\
 & =\left(q_{1}\delta^{(1)}+q_{2}\delta^{(2)}\right)R\phi'\left(u\left(I_{\infty}\right)\right)\\
 & >\phi{}_{+}^{\prime}\left(0\right),
\end{align*}
where the inequality is assumed in the statement of the proposition.
Hence, $v_{2}^{E,\infty}\left(L^{2}\right)>0$ and so the firm earns
strictly negative profit from $\left(v_{t}^{E,\infty}\left(L^{t}\right)\right)_{t=1}^{\infty}$
in the infinite horizon, in contradiction to our previous claim. We
have therefore shown that $v_{1}^{E,\infty}\left(L\right)<u\left(I_{\infty}\right)$.
By Equation (\ref{eq:LIMITING_SEQUENCE-E}) and because the firm earns
zero profit in the infinite horizon problem with utilities $\left(v_{t}^{E,\infty}\left(L^{t}\right)\right)_{t=1}^{\infty}$,
we have $v_{2}^{E,\infty}\left(L^{2}\right)>0$.

That the firm earns zero profits in the infinite horizon under both
$\left(v_{t}^{E,\infty}\left(L^{t}\right)\right)_{t=1}^{\infty}$
and $\left(v_{t}^{B,\infty}\left(L^{t}\right)\right)_{t=1}^{\infty}$,
together with Equations (\ref{eq:LIMITING_SEQUENCE}) and (\ref{eq:LIMITING_SEQUENCE-E}),
then imply that $\lambda^{E,\infty}>\lambda^{B,\infty}$ and in particular,
$v_{1}^{E,\infty}\left(L\right)<v_{1}^{B,\infty}\left(L\right)$.
As a consequence, there exists a $\varepsilon>0$ and $\bar{T}$ sufficiently
large that, for all $T\geq\bar{T}$, 
\begin{align}
v_{1}^{E,T}\left(L\right)-\lambda^{B,T}\phi\left(v_{1}^{E,T}\left(L\right)\right) & \leq v_{1}^{B,T}\left(L\right)-\lambda^{B,T}\phi\left(v_{1}^{B,T}\left(L\right)\right)+\varepsilon\nonumber \\
 & =\max_{v_{1}\left(L\right)}\left\{ v_{1}\left(L\right)-\lambda^{B,T}\phi\left(v_{1}\left(L\right)\right)\right\} +\varepsilon\label{eq:FIRST_PERIOD_INEQUALITY}
\end{align}
where we use strict convexity of $\phi$ and the fact that $v_{1}^{B,T}\left(L\right)$
maximises $v_{1}\left(L\right)-\lambda^{B,T}\phi\left(v_{1}\left(L\right)\right)$,
while $v_{1}^{E,T}\left(L\right)$ remains bounded away from $v_{1}^{B,T}\left(L\right)$
for all large enough $T$.

Therefore, for all $T\geq\bar{T}$,
\begin{align*}
\sum_{t=1}^{T}\left(\delta^{(1)}\right)^{t-1}v_{t}^{B,T}\left(L^{t}\right) & =\max_{\left(v_{t}\left(L^{t}\right)\right)_{t=1}^{T}}\left\{ \sum_{t=1}^{T}\left(\delta^{(1)}\right)^{t-1}v_{t}\left(L^{t}\right)+\lambda^{B,T}\left(I_{T}-\sum_{t=1}^{T}\frac{\phi\left(v_{t}\left(L^{t}\right)\right)}{R^{t-1}}\right)\right\} \\
 & =\max_{\left(v_{t}\left(L^{t}\right)\right)_{t=1}^{T}}\left\{ \sum_{t=1}^{T}\left(\left(\delta^{(1)}\right)^{t-1}v_{t}\left(L^{t}\right)-\lambda^{B,T}\frac{\phi\left(v_{t}\left(L^{t}\right)\right)}{R^{t-1}}\right)+\lambda^{B,T}I_{T}\right\} \\
 & \geq\sum_{t=1}^{T}\left(\left(\delta^{(1)}\right)^{t-1}v_{t}^{E,T}\left(L^{t}\right)-\lambda^{B,T}\frac{\phi\left(v_{t}^{E,T}\left(L^{t}\right)\right)}{R^{t-1}}\right)+\lambda^{B,T}I_{T}+\varepsilon\\
 & =\sum_{t=1}^{T}\left(\delta^{(1)}\right)^{t-1}v_{t}^{E,T}\left(L^{t}\right)+\lambda^{B,T}\left(I_{T}-\sum_{t=1}^{T}\frac{\phi\left(v_{t}^{E,T}\left(L^{t}\right)\right)}{R^{t-1}}\right)+\varepsilon\\
 & =\sum_{t=1}^{T}\left(\delta^{(1)}\right)^{t-1}v_{t}^{E,T}\left(L^{t}\right)+\varepsilon,
\end{align*}
where the maximisation is over $v_{t}\left(L^{t}\right)\geq0$. The
first equality holds because $\left(v_{t}^{B,T}\left(L^{t}\right)\right)_{t=1}^{T}$
maximises the Lagrangian expression above, with multiplier $\lambda^{B,T}$,
with the zero-profit condition satisfied with equality. The inequality
follows by the inequality (\ref{eq:FIRST_PERIOD_INEQUALITY}). The
fourth equality holds because $\left(v_{t}^{E,T}\left(L^{t}\right)\right)_{t=1}^{T}$
satisfies the zero-profit condition with equality. This shows precisely
the claim that $V_{T}^{B}-V_{T}^{E}\geq\varepsilon$ for all $T\geq\bar{T}$. 

\textbf{Case 2: $\left(q_{1}\delta^{(1)}+q_{2}\delta^{(2)}\right)R=1$.
}In this case, for each $T$, we have 
\[
\phi^{\prime}\left(v_{t}^{E,T}\left(L^{t}\right)\right)=\frac{\left(q_{1}\delta^{(1)}+q_{2}\delta^{(2)}\right)^{t-1}R^{t-1}}{\lambda^{B,T}}=\frac{1}{\lambda^{B,T}}
\]
for $t\leq T-1$, while
\begin{align*}
\phi^{\prime}\left(v_{T}^{E,T}\left(L^{T}\right)\right) & =\max\left\{ \frac{\left(q_{1}\delta^{(1)}+q_{2}\delta^{(2)}\right)^{T-2}\delta^{(1)}R^{T-1}}{\lambda^{B,T}},\phi_{+}^{\prime}\left(0\right)\right\} \\
 & =\max\left\{ \frac{\delta^{(1)}}{q_{1}\delta^{(1)}+q_{2}\delta^{(2)}}\frac{1}{\lambda^{B,T}},\phi^{\prime}\left(u\left(0\right)\right)\right\} .
\end{align*}
Then, because $\left(v_{t}^{E,T}\left(L^{t}\right)\right)_{t=1}^{T}$
satisfies the zero-profit condition with equality, we must have $\phi\left(v_{t}^{E,T}\left(L^{t}\right)\right)>w$
for all $t\leq T-1$. Let $\check{u}^{T}\equiv v_{t}^{E,T}\left(L^{t}\right)$
for $t\leq T-1$. Then
\[
\sum_{t=1}^{T-1}\frac{\phi\left(\check{u}^{T}\right)}{R^{t-1}}\leq\sum_{t=1}^{T}\frac{w}{R^{t-1}}.
\]
We conclude that $\phi\left(\check{u}^{T}\right)\rightarrow w$ as
$T\rightarrow\infty$.

However, as noted above, given $\delta^{(1)}R<1$, we have $\phi\left(v_{1}^{B,\infty}\left(L\right)\right)>w$.
This shows that there is a $\varepsilon>0$ and large enough $\bar{T}$
such that, for $T\geq\bar{T}$, the inequality (\ref{eq:FIRST_PERIOD_INEQUALITY})
is satisfied. The rest of the argument is then as in Case 1.

\textbf{Case 3. $\left(q_{1}\delta^{(1)}+q_{2}\delta^{(2)}\right)R>1$.
}In this case, there is $\eta>0$ and $\bar{T}$ sufficiently large
that, for all $T\geq\bar{T}$, $v_{1}^{E,T}\left(L\right)<w-\eta$.
Otherwise, there is a subsequence $\left(T_{m}\right)$ along which
$v_{1}^{E,T_{m}}\left(L\right)\geq w-\frac{1}{m}$. For $m$ sufficiently
large (for instance, using the expressions for $v_{t}^{E,T}\left(L^{t}\right)$
in Equations (\ref{eq:E_Optimality}) and (\ref{eq:E_Optimality_last})),
the policy $\left(v_{t}^{E,T_{m}}\left(L^{t}\right)\right)_{t=1}^{T_{m}}$
violates the zero-profit constraint. The remainder of the argument
is then as in Case 1 and Case 2. In particular, using $\phi\left(v_{1}^{B,\infty}\left(L\right)\right)\geq w$,
we have that there is a $\varepsilon>0$ and large enough $\bar{T}$
such that, for $T\geq\bar{T}$, the inequality (\ref{eq:FIRST_PERIOD_INEQUALITY})
is satisfied. Then the rest of the argument follows Case 1.

\end{proof}

\subsection*{Proofs of results in Section \ref{sec:RICHER}}

\begin{proof}[Proof of Proposition~\ref{prop:EFFICIENT_POLICY}]

In utility space, the problem can be written as maximising
\[
\mathbb{E}_{F}\left[\sum_{t=1}^{T}\Pi_{s=1}^{t-1}\tilde{\delta}_{s}v_{t}\left(\tilde{H}^{t}\right)\right]
\]
subject to
\begin{equation}
\mathbb{E}_{F}\left[\sum_{t=1}^{T}\frac{\phi\left(v_{t}\left(\tilde{H}^{t}\right)\right)}{R^{t-1}}\right]\leq I_{T}.\label{eq:ZERO_PROFIT_GENERAL}
\end{equation}
 First consider existence. Utilities are bounded above owing to the
firm's break-even condition. If $u\left(0\right)$ is finite, the
problem amounts to maximising a continuous function over the compact
set of values $v_{t}\left(H^{t}\right)\geq u\left(0\right)$ (where
$t\in\left\{ 1,2,\dots,T\right\} $ and $H^{t}\in\Delta^{t}$) satisfying
the zero-profit condition (\ref{eq:ZERO_PROFIT_GENERAL}). In the
alternative case where $\lim_{c\rightarrow0}u\left(c\right)=-\infty$,
we may without loss of optimality consider the $v_{t}\left(H^{t}\right)$
to be bounded below by some suitably small lower bound, so again we
are maximising over a compact set.

Uniqueness is similarly straightforward using strict convexity of
$\phi$. Were there two optimal policies, a convex combination would
yield the same agent expected utility but a strictly positive firm
profit. This policy could be strictly improved by increasing the payment
to the agent at any date and history, contradicting the optimality
of the original policies.

Assuming interiority, the optimum --- call it $\left(v_{t}^{B,T}\left(H^{t}\right)\right)_{t=1}^{T}$
--- maximises the Lagrangian 
\begin{align*}
 & \sum_{t=1}^{T}\sum_{n_{1}=1}^{N}\dots\sum_{n_{t}=1}^{N}\left(\Pi_{s^{\prime}=1}^{t}p_{n_{s^{\prime}}}\right)\left(\Pi_{s^{\prime\prime}=1}^{t-1}\delta^{\left(n_{s^{\prime\prime}}\right)}\right)v_{t}\left(\delta^{\left(n_{1}\right)},\dots,\delta^{\left(n_{t}\right)}\right)\\
+ & \lambda^{B,T}\left(I_{T}-\sum_{t=1}^{T}\sum_{n_{1}=1}^{N}\dots\sum_{n_{t}=1}^{N}\left(\Pi_{s=1}^{t}p_{n_{s}}\right)\frac{\phi\left(v_{t}\left(\delta^{\left(n_{1}\right)},\dots,\delta^{\left(n_{t}\right)}\right)\right)}{R^{t-1}}\right)
\end{align*}
for some multiplier $\lambda^{B,T}>0$. The stationarity condition
is then, for each $t$, each $\left(\delta^{\left(n_{1}\right)},\dots,\delta^{\left(n_{t}\right)}\right)$,
\[
\phi^{\prime}\left(v_{t}^{B,T}\left(\delta^{\left(n_{1}\right)},\dots,\delta^{\left(n_{t}\right)}\right)\right)=\frac{R^{t-1}\left(\Pi_{s=1}^{t-1}\delta^{\left(n_{s}\right)}\right)}{\lambda^{B,T}}.
\]
Because $\phi^{\prime}$ is increasing, it follows immediately from
this expression that $v_{t}^{B,T}\left(\delta^{\left(n_{1}\right)},\dots,\delta^{\left(n_{t}\right)}\right)$
is strictly increasing in each index $n_{s}$, $1\leq s\leq t$ .

\end{proof}

\bigskip{}

\begin{proof}[Proof of Proposition~\ref{prop:EQUILIBRIUM_GENERAL}]

Given $\delta_{1}$, Problem III can be written as a choice of utilities
$\left(v_{t}\left(H^{t}\right)\right)_{t=1}^{T}$ maximising 
\begin{equation}
\mathbb{E}_{A}\left[v_{1}\left(\delta_{1}\right)+\delta_{1}\sum_{t=2}^{T}\Pi_{s=2}^{t-1}\tilde{\delta}_{s}v_{t}\left(\delta_{1},\tilde{H}_{2}^{t}\right)\right]\label{eq:OBJECTIVE-REWRITTEN}
\end{equation}
subject to the firm break-even condition
\[
\mathbb{E}_{F}\left[\sum_{t=1}^{T}\frac{\phi\left(v_{t}\left(\delta_{1},\tilde{H}_{2}^{t}\right)\right)}{R^{t-1}}\right]\leq I_{T}
\]
and to incentive constraints for the truthful reporting of $\delta_{t}$
at all $t\geq2$ and all histories $\left(\delta_{1},H_{2}^{t-1}\right)$.
(In the main text, we subsequently refer to the problem in utility
space as ``Problem IV'', but in the proof below we refer to this still
as Problem III.) In particular, for each such history, each $\delta_{t}$
and each $\hat{\delta}_{t}$, we must have 
\begin{align*}
 & v_{t}\left(\delta_{1},H_{2}^{t-1},\delta_{t}\right)+\delta_{t}\mathbb{E}_{A}\left[\sum_{\tau=t+1}^{T}\Pi_{s=t+1}^{\tau-1}\tilde{\delta}_{s}v_{\tau}\left(\delta_{1},H_{2}^{t-1},\delta_{t},\tilde{H}_{t+1}^{\tau}\right)\right]\\
\geq & v_{t}\left(\delta_{1},H_{2}^{t-1},\hat{\delta}_{t}\right)+\delta_{t}\mathbb{E}_{A}\left[\sum_{\tau=t+1}^{T}\Pi_{s=t+1}^{\tau-1}\tilde{\delta}_{s}v_{\tau}\left(\delta_{1},H_{2}^{t-1},\hat{\delta}_{t},\tilde{H}_{t+1}^{\tau}\right)\right].
\end{align*}
Either $u\left(0\right)$ is finite or $\lim_{c\rightarrow0}u\left(c\right)=-\infty$.
In the latter case, it is easy to see that the optimal utilies must
be bounded below. We are therefore optimising a continuous function
over a compact set, and a maximum exists. Moreover, the optimum is
unique: were there two optima, the convex combination would satisfy
the incentive constraints and satisfy the break-even condition strictly.
It would then be possible to slightly raise $v_{1}\left(\delta\right)$
yielding a strictly higher agent expected payoff -- this contradicts
the optimality of the original policies.

Now consider why any equilibrium dynamic mechanism accepted by the
agent with initial type $\delta_{1}$ must solve Problem III. Throughout
the argument, we can associate with each date-1 message a uniquely
determined continuation mechanism (a mapping from messages to consumption),
and assume an agent message strategy that maximises firm profits among
those optimal for the agent. Hence, because the firm anticipates i.i.d.
types, there is associated with the date-1 message a unique value
of discounted expected profits (independent of the agent's true initial
type).

The central part of the argument is that no firm earns positive discounted
expected profit in equilibrium. Suppose not. It is enough to show
that any firm can deviate to a mechanism that generates arbitrarily
close to the sum of expected profits across all firms. For each type
$\delta_{1}=\delta^{\left(n\right)}$ that accepts with positive probability
a mechanism of a Firm $j\left(n\right)$ and sends a message $\hat{\delta}_{1}=\delta^{\left(n\right)}$
in the putative equilibrium, outcomes will be determined by utilities
$\left(v_{t}^{j\left(n\right)}\left(\delta^{\left(n\right)},H_{2}^{t}\right)\right)_{t=1}^{T}$.
In case of agent randomisation, consider only the firm $j\left(n\right)$
for which firm expected discounted profits are highest (and take the
smallest firm index if multiple). Then consider the mechanism defined
by
\[
v_{t}\left(\delta^{\left(n\right)},H_{2}^{t}\right)=v_{t}^{j\left(n\right)}\left(\delta^{\left(n\right)},H_{2}^{t}\right)
\]
for the initial types $\delta_{1}=\delta^{\left(n\right)}$ that accept
some mechanism with positive probability in the putative equilibrium
(thus the date-1 reports in the new mechanism are a subset of $\Delta$).
This mechanism is incentive compatible in the following restricted
sense. Any agent type that accepted a mechanism with positive probability
in the putative equilibrium is willing to report truthfully if accepting
the new mechanism. Now make two changes: (i) omit any report $\delta_{1}$
that fails to generate strictly positive expected profit, and (ii)
increase the initial utility $v_{1}\left(\delta_{1}\right)$ by an
arbitrarily small constant amount, uniform over the initial messages
$\delta_{1}$. Now, for all values of the initial type $\delta_{1}$
that accept a mechanism with positive probability in the putative
equilibrium but not omitted in (i), the agent strictly prefers to
accept this mechanism than any of the others offered, and is willing
to report truthfully. This shows that the profit from each such type
is at least arbitrarily close to that generated from this type by
all firms in the putative equilibrium.\footnote{While the agent is willing to report truthfully, according to the
prescribed strategy, he may not do so if there is some other message
that generates strictly higher profits for the firm. However, this
only raises the profitability of the constructed deviation.} Other agent types may accept the mechanism and choose some (untruthful)
report or not accept the mechanism. However, irrespective of what
the agent's strategy specifies for the other types, the new mechanism
obtains arbitrarily close to the expected total profit among all firms.
This is a profitable deviation at least for a firm earning the lowest
expected profit.

We now note that if a type $\delta_{1}=\delta^{\left(n\right)}$ accepts
a firm's mechanism with positive probability in equilibrium, then
the consumption following truthful reporting (and given the future
i.i.d. evolution of types as anticipated by each firm) must generate
exactly zero expected profit. Suppose instead that a strict loss is
generated. Then the firm must be generating positive expected profits
from the acceptance of the mechanism by some other type $\delta_{1}\neq\delta^{\left(n\right)}$.
There is then a profitable deviation. In particular, similar to above,
the firm can omit any initial reports $\delta_{1}$ not generating
strictly positive expected profit, while increasing the date-1 utility
for the remaining reports by an arbitrarily small and constant amount.
Any agent type that was generating positive expected profits for the
firm continues to generate at least arbitrarily close to this amount,
while any negative profits are eliminated.

We now argue that any initial type $\delta_{1}$ has equilibrium outcomes
given by the solution to Problem III. Note that each type must be
obtaining an (agent)-anticipated discounted payoff equal to the value
in Problem III. Obtaining more would imply negative profit for the
report $\delta_{1}$ (which we saw in the previous paragraph was not
possible in equilibrium). If a given agent type obtains less, then
a firm can earn strictly positive profit by offering a mechanism that
is the solution to Problem III for this agent type (requiring no date-1
message), for a value of income $I_{T}$ that is below but arbitrarily
close to the true value. Indeed, by the Theorem of the Maximum (using
that the constraint set can be considered compact-valued, as argued
above, and because it is continuous in $I_{T}$), the value of Problem
III is continuous in $I_{T}$. Given that we argued equilibrium must
yield firms zero profits, this is a profitable deviation. We can conclude
that each initial type $\delta_{1}$ anticipates discounted payoffs
equal to the value of Problem III, while generating any firm whose
mechanism the agent accepts zero profits. By uniqueness of the solution
to Problem III, the equilibrium outcomes of type $\delta_{1}$ are
then unique.

Finally, it is an equilibrium for each firm to symmetrically offer
the mechanism $\left(c_{t}\left(\delta_{1},H_{2}^{t}\right)\right)_{t=1}^{T}$,
with $\left(\delta_{1},H_{2}^{t}\right)\in\Delta^{t}$ for each $t\leq T$
and $\left(\delta_{1},H_{2}^{T}\right)\in\Delta^{T-1}$, with values
uniquely solving Problem III, and for the agent to randomise symmetrically
over the equilibrium mechanisms and to report truthfully. Agent play
following any off-path selection of mechanisms satisfies the equilibrium
requirements: the agent follows a deterministic message strategy that
maximises the firm's discounted profits among those optimal for the
agent (since there are finitely many messages in each period, such
a strategy exists). When the mechanism satisfies our notion of a direct
mechanism (see the model set-up), the agent reports truthfully. Then
firms have no incentive to deviate: if a firm offers a mechanism where
discounted profits associated with a given date-1 message are strictly
positive, the agent does not send this message; indeed, he obtains
a lower expected payoff than the value of Problem III associated with
his true initial type $\delta_{1}$, while this value in Problem III
is available from other firm(s). Agent incentive compatibility in
the equilibrium mechanism from date-2 onwards is implied by the incentive
constraints (\ref{eq:IC_GEN}), as argued in the main text. Agent
incentive compatibility in the equilibrium mechanism at date 1 is
satisfied because each initial type $\delta_{1}$ receives, by reporting
truthfully, outcomes that solve Problem III, where this problem is
to maximise the agent's discounted payoff given $\delta_{1}$ subject
to constraints that do not depend on $\delta_{1}$.

\end{proof}

\bigskip{}

\begin{proof}[Proof of Lemma~\ref{lem:IC_Section4}]

To see the monotonicity in the lemma, consider any history $\left(\delta_{1},H_{2}^{t-1}\right)$,
and any adjacent discount factors $\delta^{\left(n\right)}$ and $\delta^{\left(n+1\right)}$,
$n\leq N-1$. Analogous to Equations (\ref{eq:HIGH-IC}) and (\ref{eq:LOW-IC}),
we have 
\begin{align}
 & v_{t}\left(\delta_{1},H_{2}^{t-1},\delta^{(n+1)}\right)+\delta^{(n+1)}\mathbb{E}_{A}\left[\sum_{\tau=t+1}^{T}\Pi_{s=t+1}^{\tau-1}\tilde{\delta}_{s}v_{\tau}\left(\delta_{1},H_{2}^{t-1},\delta^{(n+1)},\tilde{H}_{t+1}^{\tau}\right)\right]\nonumber \\
\geq & v_{t}\left(\delta_{1},H_{2}^{t-1},\delta^{(n)}\right)+\delta^{(n+1)}\mathbb{E}_{A}\left[\sum_{\tau=t+1}^{T}\Pi_{s=t+1}^{\tau-1}\tilde{\delta}_{s}v_{\tau}\left(\delta_{1},H_{2}^{t-1},\delta^{(n)},\tilde{H}_{t+1}^{\tau}\right)\right],\label{eq:HIGH-IC-GEN}
\end{align}
and
\begin{align}
 & v_{t}\left(\delta_{1},H_{2}^{t-1},\delta^{(n)}\right)+\delta^{(n)}\mathbb{E}_{A}\left[\sum_{\tau=t+1}^{T}\Pi_{s=t+1}^{\tau-1}\tilde{\delta}_{s}v_{\tau}\left(\delta_{1},H_{2}^{t-1},\delta^{(n)},\tilde{H}_{t+1}^{\tau}\right)\right]\nonumber \\
\geq & v_{t}\left(\delta_{1},H_{2}^{t-1},\delta^{(n+1)}\right)+\delta^{(n)}\mathbb{E}_{A}\left[\sum_{\tau=t+1}^{T}\Pi_{s=t+1}^{\tau-1}\tilde{\delta}_{s}v_{\tau}\left(\delta_{1},H_{2}^{t-1},\delta^{(n+1)},\tilde{H}_{t+1}^{\tau}\right)\right].\label{eq:LOW-IC-GEN}
\end{align}
Then as in Equations (\ref{eq:CURRENT-INEQUALITY}) and (\ref{eq:LATER-INEQUALITY})
in the proof of Lemma \ref{lem:IC_Section3}, we conclude that, for
any $t\in\left\{ 2,\dots,T-1\right\} $, any $\left(\delta_{1},H_{2}^{t-1}\right)$,
\begin{equation}
v_{t}\left(\delta_{1},H_{2}^{t-1},\delta^{(n)}\right)\geq v_{t}\left(\delta_{1},H_{2}^{t-1},\delta^{(n+1)}\right)\label{eq:MON_CURRENT_GEN}
\end{equation}
 and 
\begin{equation}
\mathbb{E}_{A}\left[\sum_{\tau=t+1}^{T}\Pi_{s=t+1}^{\tau-1}\tilde{\delta}_{s}v_{\tau}\left(\delta_{1},H_{2}^{t-1},\delta^{(n+1)},\tilde{H}_{t+1}^{\tau}\right)\right]\geq\mathbb{E}_{A}\left[\sum_{\tau=t+1}^{T}\Pi_{s=t+1}^{\tau-1}\tilde{\delta}_{s}v_{\tau}\left(\delta_{1},H_{2}^{t-1},\delta^{(n)},\tilde{H}_{t+1}^{\tau}\right)\right],\label{eq:MON_LATER_GEN}
\end{equation}
with both inequalities strict if and only if at least one of (\ref{eq:HIGH-IC-GEN})
and (\ref{eq:HIGH-IC-GEN}) hold as a strict inequality.

Now let us show that local incentive constraints imply the remaining
incentive constraints. To see that downward incentive constraints
hold, argue as follows. Consider any $n>2$ and suppose that downward
incentive constraints hold for all lower types. In particular, for
any $n^{\prime}<n-1$, we have
\begin{align*}
 & v_{t}\left(\delta_{1},H_{2}^{t-1},\delta^{(n-1)}\right)+\delta^{(n-1)}\mathbb{E}_{A}\left[\sum_{\tau=t+1}^{T}\Pi_{s=t+1}^{\tau-1}\tilde{\delta}_{s}v_{\tau}\left(\delta_{1},H_{2}^{t-1},\delta^{(n-1)},\tilde{H}_{t+1}^{\tau}\right)\right]\\
\geq & v_{t}\left(\delta_{1},H_{2}^{t-1},\delta^{(n^{\prime})}\right)+\delta^{(n-1)}\mathbb{E}_{A}\left[\sum_{\tau=t+1}^{T}\Pi_{s=t+1}^{\tau-1}\tilde{\delta}_{s}v_{\tau}\left(\delta_{1},H_{2}^{t-1},\delta^{(n^{\prime})},\tilde{H}_{t+1}^{\tau}\right)\right].
\end{align*}
By the monotonicity finding above,
\begin{align*}
 & v_{t}\left(\delta_{1},H_{2}^{t-1},\delta^{(n-1)}\right)+\delta^{(n)}\mathbb{E}_{A}\left[\sum_{\tau=t+1}^{T}\Pi_{s=t+1}^{\tau-1}\tilde{\delta}_{s}v_{\tau}\left(\delta_{1},H_{2}^{t-1},\delta^{(n-1)},\tilde{H}_{t+1}^{\tau}\right)\right]\\
\geq & v_{t}\left(\delta_{1},H_{2}^{t-1},\delta^{(n^{\prime})}\right)+\delta^{(n)}\mathbb{E}_{A}\left[\sum_{\tau=t+1}^{T}\Pi_{s=t+1}^{\tau-1}\tilde{\delta}_{s}v_{\tau}\left(\delta_{1},H_{2}^{t-1},\delta^{(n^{\prime})},\tilde{H}_{t+1}^{\tau}\right)\right].
\end{align*}
By the local incentive constraint,
\begin{align*}
 & v_{t}\left(\delta_{1},H_{2}^{t-1},\delta^{(n)}\right)+\delta^{(n)}\mathbb{E}_{A}\left[\sum_{\tau=t+1}^{T}\Pi_{s=t+1}^{\tau-1}\tilde{\delta}_{s}v_{\tau}\left(\delta_{1},H_{2}^{t-1},\delta^{(n)},\tilde{H}_{t+1}^{\tau}\right)\right]\\
\geq & v_{t}\left(\delta_{1},H_{2}^{t-1},\delta^{(n-1)}\right)+\delta^{(n)}\mathbb{E}_{A}\left[\sum_{\tau=t+1}^{T}\Pi_{s=t+1}^{\tau-1}\tilde{\delta}_{s}v_{\tau}\left(\delta_{1},H_{2}^{t-1},\delta^{(n-1)},\tilde{H}_{t+1}^{\tau}\right)\right].
\end{align*}
Therefore, indeed 
\begin{align*}
 & v_{t}\left(\delta_{1},H_{2}^{t-1},\delta^{(n)}\right)+\delta^{(n)}\mathbb{E}_{A}\left[\sum_{\tau=t+1}^{T}\Pi_{s=t+1}^{\tau-1}\tilde{\delta}_{s}v_{\tau}\left(\delta_{1},H_{2}^{t-1},\delta^{(n)},\tilde{H}_{t+1}^{\tau}\right)\right]\\
\geq & v_{t}\left(\delta_{1},H_{2}^{t-1},\delta^{(n^{\prime})}\right)+\delta^{(n)}\mathbb{E}_{A}\left[\sum_{\tau=t+1}^{T}\Pi_{s=t+1}^{\tau-1}\tilde{\delta}_{s}v_{\tau}\left(\delta_{1},H_{2}^{t-1},\delta^{(n^{\prime})},\tilde{H}_{t+1}^{\tau}\right)\right].
\end{align*}
We conclude that downward incentive constraints hold for all types
$\delta^{\left(n\right)}$ and lower, and hence by induction all downward
incentive constraints hold.

To see that upward incentive constraints hold, argue as follows. Consider
any $n<N-1$ and suppose that upward incentive constraints hold for
all higher types. In particular, for any $n^{\prime}>n+1$, we have
\begin{align*}
 & v_{t}\left(\delta_{1},H_{2}^{t-1},\delta^{(n+1)}\right)+\delta^{(n+1)}\mathbb{E}_{A}\left[\sum_{\tau=t+1}^{T}\Pi_{s=t+1}^{\tau-1}\tilde{\delta}_{s}v_{\tau}\left(\delta_{1},H_{2}^{t-1},\delta^{(n+1)},\tilde{H}_{t+1}^{\tau}\right)\right]\\
\geq & v_{t}\left(\delta_{1},H_{2}^{t-1},\delta^{(n^{\prime})}\right)+\delta^{(n+1)}\mathbb{E}_{A}\left[\sum_{\tau=t+1}^{T}\Pi_{s=t+1}^{\tau-1}\tilde{\delta}_{s}v_{\tau}\left(\delta_{1},H_{2}^{t-1},\delta^{(n^{\prime})},\tilde{H}_{t+1}^{\tau}\right)\right].
\end{align*}
By the aforementioned monotonicity, we have 
\begin{align*}
 & v_{t}\left(\delta_{1},H_{2}^{t-1},\delta^{(n+1)}\right)+\delta^{(n)}\mathbb{E}_{A}\left[\sum_{\tau=t+1}^{T}\Pi_{s=t+1}^{\tau-1}\tilde{\delta}_{s}v_{\tau}\left(\delta_{1},H_{2}^{t-1},\delta^{(n+1)},\tilde{H}_{t+1}^{\tau}\right)\right]\\
\geq & v_{t}\left(\delta_{1},H_{2}^{t-1},\delta^{(n^{\prime})}\right)+\delta^{(n)}\mathbb{E}_{A}\left[\sum_{\tau=t+1}^{T}\Pi_{s=t+1}^{\tau-1}\tilde{\delta}_{s}v_{\tau}\left(\delta_{1},H_{2}^{t-1},\delta^{(n^{\prime})},\tilde{H}_{t+1}^{\tau}\right)\right].
\end{align*}
By the local incentive constraint, we have 
\begin{align*}
 & v_{t}\left(\delta_{1},H_{2}^{t-1},\delta^{(n)}\right)+\delta^{(n)}\mathbb{E}_{A}\left[\sum_{\tau=t+1}^{T}\Pi_{s=t+1}^{\tau-1}\tilde{\delta}_{s}v_{\tau}\left(\delta_{1},H_{2}^{t-1},\delta^{(n)},\tilde{H}_{t+1}^{\tau}\right)\right]\\
\geq & v_{t}\left(\delta_{1},H_{2}^{t-1},\delta^{(n+1)}\right)+\delta^{(n)}\mathbb{E}_{A}\left[\sum_{\tau=t+1}^{T}\Pi_{s=t+1}^{\tau-1}\tilde{\delta}_{s}v_{\tau}\left(\delta_{1},H_{2}^{t-1},\delta^{(n+1)},\tilde{H}_{t+1}^{\tau}\right)\right].
\end{align*}
Therefore,
\begin{align*}
 & v_{t}\left(\delta_{1},H_{2}^{t-1},\delta^{(n)}\right)+\delta^{(n)}\mathbb{E}_{A}\left[\sum_{\tau=t+1}^{T}\Pi_{s=t+1}^{\tau-1}\tilde{\delta}_{s}v_{\tau}\left(\delta_{1},H_{2}^{t-1},\delta^{(n)},\tilde{H}_{t+1}^{\tau}\right)\right]\\
\geq & v_{t}\left(\delta_{1},H_{2}^{t-1},\delta^{(n^{\prime})}\right)+\delta^{(n)}\mathbb{E}_{A}\left[\sum_{\tau=t+1}^{T}\Pi_{s=t+1}^{\tau-1}\tilde{\delta}_{s}v_{\tau}\left(\delta_{1},H_{2}^{t-1},\delta^{(n^{\prime})},\tilde{H}_{t+1}^{\tau}\right)\right].
\end{align*}
This establishes that incentive constraints for $\delta^{\left(n\right)}$
are satisfied upwards, and by induction upwards incentive constraints
are satisfied for all types.

\end{proof}

\bigskip{}

\begin{proof}[Proof of Lemma~\ref{lem:UPWARD_BINDS_GEN}]

Suppose for a contradiction there is $\bar{t}\in\left\{ 2,\dots,T-1\right\} $,
and any history $\left(\bar{\delta}_{1},\bar{H}_{2}^{t-1}\right)$,
and some $\delta^{\left(\bar{n}\right)}$ with $\bar{n}<N$ such that
the local upward incentive constraint fails to hold with equality.
That is, 
\begin{align*}
 & v_{\bar{t}}\left(\bar{\delta}_{1},\bar{H}_{2}^{\bar{t}-1},\delta^{(\bar{n})}\right)+\delta^{(\bar{n})}\mathbb{E}_{A}\left[\sum_{\tau=\bar{t}+1}^{T}\Pi_{s=\bar{t}+1}^{\tau-1}\tilde{\delta}_{s}v_{\tau}\left(\bar{\delta}_{1},\bar{H}_{2}^{t-1},\delta^{(\bar{n})},\tilde{H}_{\bar{t}+1}^{\tau}\right)\right]\\
=x+ & v_{\bar{t}}\left(\bar{\delta}_{1},\bar{H}_{2}^{\bar{t}-1},\delta^{(\bar{n}+1)}\right)+\delta^{(\bar{n})}\mathbb{E}_{A}\left[\sum_{\tau=\bar{t}+1}^{T}\Pi_{s=\bar{t}+1}^{\tau-1}\tilde{\delta}_{s}v_{\tau}\left(\bar{\delta}_{1},\bar{H}_{2}^{\bar{t}-1},\delta^{(\bar{n}+1)},\tilde{H}_{\bar{t}+1}^{\tau}\right)\right]
\end{align*}
for some $x>0$. We can determine a new mechanism with utilities $\left(\check{v}_{t}\left(\bar{\delta}_{1},H_{2}^{t}\right)\right)_{t=1}^{T}$
such that the only changes to utility are at the history $\left(\bar{\delta}_{1},\bar{H}_{2}^{\bar{t}-1}\right)$:
\[
\check{v}_{\bar{t}}\left(\bar{\delta}_{1},\bar{H}_{2}^{\bar{t}-1},\delta^{\left(n\right)}\right)=v_{\bar{t}}\left(\bar{\delta}_{1},\bar{H}_{2}^{\bar{t}-1},\delta^{\left(n\right)}\right)+x\sum_{n=1}^{\bar{n}}q_{n}
\]
for $n>\bar{n}$, and 
\[
\check{v}_{\bar{t}}\left(\bar{\delta}_{1},\bar{H}_{2}^{\bar{t}-1},\delta^{\left(n\right)}\right)=v_{\bar{t}}\left(\bar{\delta}_{1},\bar{H}_{2}^{\bar{t}-1},\delta^{\left(n\right)}\right)-x\sum_{n=\bar{n}+1}^{N}q_{n}
\]
for $n\leq\bar{n}$. The only (adjacent) incentive constraints this
affects are: (i) the upward incentive constraint for $\delta^{\left(\bar{n}\right)}$
now holds with equality, and (ii) the downward incentive constraint
for $\delta^{\left(\bar{n}+1\right)}$ is relaxed. In particular,
note that incentive constraints are unaffected for all earlier types
(as expected continuation payoffs under truthful reporting are unaffected).
All other adjacent incentive constraints at history $\left(\bar{\delta}_{1},\bar{H}_{2}^{\bar{t}-1}\right)$
are unaffected as a constant is added to each side of the constraint.

Note also that, by incentive compatibility, $\check{v}_{\bar{t}}\left(\bar{\delta}_{1},\bar{H}_{2}^{\bar{t}-1},\delta^{\left(n\right)}\right)$
remains non-increasing in $n$, and so the consumption levels prescribed
by the new mechanism are non-negative.

Now consider the effect on firm profits. The change in the cost of
the contract to the firm is proportional to
\begin{align*}
 & \sum_{n=1}^{\bar{n}}p_{n}\left(\phi\left(v_{\bar{t}}\left(\bar{\delta}_{1},\bar{H}_{2}^{\bar{t}-1},\delta^{\left(n\right)}\right)-x\left(\sum_{m=\bar{n}+1}^{N}q_{m}\right)\right)\right)+\sum_{n=\bar{n}+1}^{N}p_{n}\left(\phi\left(v_{\bar{t}}\left(\bar{\delta}_{1},\bar{H}_{2}^{\bar{t}-1},\delta^{\left(n\right)}\right)+x\left(\sum_{m=1}^{\bar{n}}q_{m}\right)\right)\right)\\
 & -\sum_{n=1}^{N}p_{n}\left(\phi\left(v_{\bar{t}}\left(\bar{\delta}_{1},\bar{H}_{2}^{\bar{t}-1},\delta^{\left(n\right)}\right)\right)\right)\\
= & -\sum_{n=1}^{\bar{n}}p_{n}\left(\sum_{m=\bar{n}+1}^{N}q_{m}\right)\int_{0}^{x}\phi^{\prime}\left(v_{\bar{t}}\left(\bar{\delta}_{1},\bar{H}_{2}^{\bar{t}-1},\delta^{\left(n\right)}\right)-r\left(\sum_{m=\bar{n}+1}^{N}q_{m}\right)\right)dr\\
 & +\sum_{n=\bar{n}+1}^{N}p_{n}\left(\sum_{m=1}^{\bar{n}}q_{m}\right)\int_{0}^{x}\phi^{\prime}\left(v_{\bar{t}}\left(\bar{\delta}_{1},\bar{H}_{2}^{\bar{t}-1},\delta^{\left(n\right)}\right)+r\left(\sum_{m=1}^{\bar{n}}q_{m}\right)\right)dr\\
< & -\sum_{n=1}^{\bar{n}}p_{n}\left(\sum_{m=\bar{n}+1}^{N}q_{m}\right)\int_{0}^{x}\phi^{\prime}\left(\check{v}_{\bar{t}}\left(\bar{\delta}_{1},\bar{H}_{2}^{\bar{t}-1},\delta^{\left(\bar{n}\right)}\right)\right)dr\\
 & +\sum_{n=\bar{n}+1}^{N}p_{n}\left(\sum_{m=1}^{\bar{n}}q_{m}\right)\int_{0}^{x}\phi^{\prime}\left(\check{v}_{\bar{t}}\left(\bar{\delta}_{1},\bar{H}_{2}^{\bar{t}-1},\delta^{\left(\bar{n}+1\right)}\right)\right)dr\\
\leq & \left(\left(\sum_{m=1}^{\bar{n}}q_{m}\right)\left(\sum_{n=\bar{n}+1}^{N}p_{n}\right)-\left(\sum_{m=\bar{n}+1}^{N}q_{m}\right)\left(\sum_{n=1}^{\bar{n}}p_{n}\right)\right)\int_{0}^{x}\phi^{\prime}\left(\check{v}_{\bar{t}}\left(\bar{\delta}_{1},\bar{H}_{2}^{\bar{t}-1},\delta^{\left(\bar{n}\right)}\right)\right)dr.\\
\leq & 0.
\end{align*}
The first inequality uses the definition of $\check{v}_{\bar{t}}$,
that $\check{v}_{\bar{t}}\left(\bar{\delta}_{1},\bar{H}_{2}^{\bar{t}-1},\cdot\right)$
is non-increasing, and that $\phi$ is strictly convex. The second
inequality uses that $\check{v}_{\bar{t}}\left(\bar{\delta}_{1},\bar{H}_{2}^{\bar{t}-1},\cdot\right)$
is non-increasing, and that $\phi$ is convex. The final inequality
follows because the agent is weakly more optimistic about patience
than the firm in the sense of first-order stochastic dominance. But
we have thereby obtained a mechanism that is strictly more profitable
than the original, contradicting that the original solves Problem
IV. 

\end{proof}

\bigskip{}

\begin{proof}[Proof of Corollary~\ref{cor:WORSE_TERMS}]

Consider any date $t\in\left\{ 2,\dots,T-1\right\} $, any history
$\left(\delta_{1},H_{2}^{t-1}\right)$. It is enough to show that,
for any $n<N$, the continuation cost in Equation (\ref{eq:CON TINUATION_COST})
is strictly higher for $\delta_{t}=\delta^{\left(n+1\right)}$ than
for $\delta_{t}=\delta^{\left(n\right)}$ whenever the continuation
policy is distinct. Suppose not. By Lemma \ref{lem:UPWARD_BINDS_GEN},
type $\delta^{\left(n\right)}$ is willing to report $\delta^{\left(n+1\right)}$
(and report truthfully thereafter). This single change in reporting
provides a different reporting strategy for the agent, inducing different
outcomes --- this mapping from discount factors to utilities/consumption
can be induced by a revised direct mechanism. The revised mechanism
is distinct from the original, gives the agent the same expected payoff,
but weakly lowers the firm's expected cost. Hence, the new mechanism
respects all constraints and hence represents a distinct solution
to Problem IV. This contradicts uniqueness of the solution to Problem
IV. 

\end{proof}

\bigskip{}

\begin{proof}[Proof of Example~\ref{exa:Squareroot}]

Fix $\delta_{1}$ and write $v_{1}=v_{1}\left(\delta_{1}\right)$,
while we let $w_{n}=v_{2}\left(\delta_{1},\delta^{\left(n\right)}\right)$
and $z_{n}=v_{3}\left(\delta_{1},\delta^{\left(n\right)}\right)$.
We further introduce $U_{n}=w_{n}+\delta^{\left(n\right)}z_{n}$. 

The firm's break-even condition is
\[
\phi\left(v_{1}\right)+\sum_{n=1}^{N}p_{n}\left(\frac{\phi\left(w_{n}\right)}{R}+\frac{\phi\left(z_{n}\right)}{R^{2}}\right)\leq I_{3},
\]
while we use Lemma \ref{lem:UPWARD_BINDS_GEN} to substitute in the
incentive constraints. Given $U_{1}$ and $\left(z_{n}\right)_{n=2}^{N}$,
we use the $N-1$ upward incentive constraints holding with equality
to determine $\left(w_{n}\right)_{n=2}^{N}$. In particular, we observe
that for $k\in\left\{ 2,\dots,N\right\} $, 
\[
U_{k}=U_{1}+\sum_{n=2}^{k}\left(\delta^{\left(n\right)}-\delta^{\left(n-1\right)}\right)z_{n}
\]
and therefore 
\[
w_{k}=U_{1}+\sum_{n=2}^{k}\left(\delta^{\left(n\right)}-\delta^{\left(n-1\right)}\right)z_{n}-\delta^{\left(k\right)}z_{k}.
\]

Problem IV can then be written as maximising 
\[
v_{1}+\delta_{1}\sum_{n=1}^{N}q_{n}U_{n}=v_{1}+\delta_{1}\left(U_{1}+\sum_{n=2}^{N}\bar{Q}_{n}z_{n}\left(\delta^{\left(n\right)}-\delta^{\left(n-1\right)}\right)\right),
\]
where $\bar{Q}_{n}=\sum_{m=n}^{N}q_{m}$, subject to the firm's break-even
condition and non-negativity constraints. Rather than imposing the
negativity constraints, we study a relaxed program where all utilities
may be negative, but the implied payment for an arbitrary utility
$v$ is $\phi\left(v\right)=v^{2}$. We then confirm that, when $N$
is sufficiently large, all utilities are strictly positive and hence
the solution to the relaxed problem is the solution to the problem
of interest.

\textbf{Necessary conditions for an optimum. }Due to the firm's break-even
condition, the feasible utilities come from a compact set. The same
argument as in Proposition \ref{prop:EQUILIBRIUM_GENERAL} then implies
the existence of a unique solution to the relaxed problem. We can
then study a Lagrangian with only the firm's break-even condition
\begin{align}
{\cal L=} & v_{1}+\delta_{1}\left(w_{1}+\delta^{\left(1\right)}z_{1}+\sum_{n=2}^{N}\bar{Q}_{n}z_{n}\left(\delta^{\left(n\right)}-\delta^{\left(n-1\right)}\right)\right)\nonumber \\
 & +\lambda\left(I_{3}-\phi\left(v_{1}\right)-\sum_{n=1}^{N}p_{n}\left(\frac{\phi\left(w_{1}+\delta^{\left(1\right)}z_{1}+\sum_{i=2}^{n}\left(\delta^{\left(i\right)}-\delta^{\left(i-1\right)}\right)z_{i}-\delta^{\left(n\right)}z_{n}\right)}{R}+\frac{\phi\left(z_{n}\right)}{R^{2}}\right)\right).\label{eq:Lagrangean_Example}
\end{align}
The stationarity conditions are as follows. Optimality of $v_{1}$
yields the necessary condition
\[
1=\lambda\phi^{\prime}\left(v_{1}\right).
\]
Optimality of $w_{1}$ yields the necessary condition
\begin{equation}
\delta_{1}=\frac{\lambda}{R}\sum_{n=1}^{N}p_{n}\phi^{\prime}\left(w_{1}+\delta^{\left(1\right)}z_{1}+\sum_{i=2}^{n}\left(\delta^{\left(i\right)}-\delta^{\left(i-1\right)}\right)z_{i}-\delta^{\left(n\right)}z_{n}\right).\label{eq:OPTIMAL_W1}
\end{equation}
Then, optimality of $z_{n}$ for $n\geq2$ yields
\begin{align}
0= & \left(\delta^{\left(n\right)}-\delta^{\left(n-1\right)}\right)\left(\delta_{1}\bar{Q}_{n}-\sum_{m=n}^{N}p_{m}\frac{\lambda}{R}\phi'\left(U_{1}+\sum_{i=2}^{m}\left(\delta^{\left(i\right)}-\delta^{\left(i-1\right)}\right)z_{i}-\delta^{\left(n\right)}z_{n}\right)\right)\nonumber \\
 & +\frac{\lambda p_{n}\delta^{\left(n\right)}}{R}\phi^{\prime}\left(U_{1}+\sum_{i=2}^{n}\left(\delta^{\left(i\right)}-\delta^{\left(i-1\right)}\right)z_{i}-\delta^{\left(n\right)}z_{n}\right)-\frac{\lambda p_{n}}{R^{2}}\phi'\left(z_{n}\right),\label{eq:EULER}
\end{align}
while optimality of $z_{1}$ yields
\begin{equation}
\delta^{\left(1\right)}R\phi^{\prime}\left(w_{1}\right)=\phi^{\prime}\left(z_{1}\right).\label{eq:EFFICIENCY_BOTTOM}
\end{equation}
Given the zero-profit condition is an inequality constraint, we have
$\lambda\geq0$, while stationarity requires $\lambda>0$. Note that
Equation (\ref{eq:OPTIMAL_W1}) and $\phi\left(u\right)=2u$ imply
that, as $\lambda\rightarrow0$, mean date-2 utility diverges, i.e.
$\sum_{n=1}^{N}p_{n}w_{n}\rightarrow\infty$, and so there is a violation
of the firm's break-even condition. Hence, $\lambda$ remains bounded
away from zero as $N\rightarrow\infty$.

\textbf{Taking differences. }We now use the specific values of $q_{n}$,
$p_{n}$ and $\phi$. Considering $n\in\left\{ 1,\dots,N-1\right\} $,
taking differences of the optimality conditions for $z_{n}$ (the
condition for $z_{n+1}$ less the condition for $z_{n}$), we have
\begin{align}
0= & \frac{\bar{\delta}-\underline{\delta}}{N-1}\frac{1}{N}\left(-\delta_{1}+\frac{\lambda}{R}4\left(U_{1}+\sum_{i=2}^{n}\left(\delta^{\left(i\right)}-\delta^{\left(i-1\right)}\right)z_{i}-\delta^{\left(n\right)}z_{n}\right)\right)\nonumber \\
 & -\left(z_{n+1}-z_{n}\right)\left(2\frac{\lambda\frac{1}{N}}{R^{2}}+2\frac{\lambda\frac{1}{N}}{R}\delta^{\left(n+1\right)}\delta^{\left(n\right)}\right).\label{eq:FIRST_DIFFERENCE}
\end{align}
Considering $n\in\left\{ 1,\dots,N-2\right\} $ and taking a further
difference, yields
\begin{equation}
z_{n+2}-z_{n+1}=\left(z_{n+1}-z_{n}\right)\frac{\frac{1}{R}+\delta^{\left(n+1\right)}\delta^{\left(n\right)}-2\frac{\bar{\delta}-\underline{\delta}}{N-1}\delta^{\left(n\right)}}{\frac{1}{R}+\delta^{\left(n+2\right)}\delta^{\left(n+1\right)}}.\label{eq:DIFFERENCES}
\end{equation}
Given $N$ is sufficiently large, the signs of both sides are the
same: they are either both strictly positive, both strictly negative,
or both zero. That is, $z_{n}$ is either increasing, decreasing,
or constant.

\textbf{Proof that $z_{n}$ is strictly increasing. }We now show that
$z_{n}$ is increasing by assuming otherwise and reaching a contradiction
if $N$ is large. We will use that, for $n\in\left\{ 1,\dots,N-1\right\} $,
\begin{align*}
w_{n+1}-w_{n}= & U_{1}+\sum_{i=2}^{n+1}\left(\delta^{\left(i\right)}-\delta^{\left(i-1\right)}\right)z_{i}-\delta^{\left(n+1\right)}z_{n+1}\\
 & -\left(U_{1}+\sum_{i=2}^{n}\left(\delta^{\left(i\right)}-\delta^{\left(i-1\right)}\right)z_{i}-\delta^{\left(n\right)}z_{n}\right)\\
= & \delta^{\left(n\right)}\left(z_{n}-z_{n+1}\right).
\end{align*}
Since $z_{n}$ is weakly decreasing (either constant or decreasing)
$w_{n}$ is weakly increasing (constant or increasing). Therefore,
recalling Equation (\ref{eq:OPTIMAL_W1}),

\[
\delta_{1}\leq2\frac{\lambda}{R}\left(U_{1}+\sum_{i=2}^{N}\left(\delta^{\left(i\right)}-\delta^{\left(i-1\right)}\right)z_{i}-\delta^{\left(N\right)}z_{N}\right),
\]
with a strict inequality is $z_{n}$ is strictly decreasing. Therefore,
by Equation (\ref{eq:FIRST_DIFFERENCE}) evaluated at $n=N-1$,
\begin{align*}
0= & \frac{\bar{\delta}-\underline{\delta}}{N-1}\frac{1}{N}\left(\begin{array}{c}
-\delta_{1}+4\frac{\lambda}{R}\left(U_{1}+\sum_{i=2}^{N}\left(\delta^{\left(i\right)}-\delta^{\left(i-1\right)}\right)z_{i}-\delta^{\left(N\right)}z_{N}\right)\\
+4\frac{\lambda}{R}\delta^{\left(N-1\right)}\left(z_{N}-z_{N-1}\right)
\end{array}\right)\\
 & -\left(z_{N}-z_{N-1}\right)\left(2\frac{\lambda\frac{1}{N}}{R^{2}}+2\frac{\lambda\frac{1}{N}}{R}\delta^{\left(N\right)}\delta^{\left(N-1\right)}\right)\\
\geq & \frac{\bar{\delta}-\underline{\delta}}{N-1}\frac{1}{N}\delta_{1}+\left(z_{N}-z_{N-1}\right)\frac{\lambda}{RN}\left(4\delta^{\left(N-1\right)}\frac{\bar{\delta}-\underline{\delta}}{N-1}-\left(\frac{2}{R}+2\delta^{\left(N\right)}\delta^{\left(N-1\right)}\right)\right).
\end{align*}
Since $z_{N}\leq z_{N-1}$, the right side of this inequality is strictly
positive for $N$ sufficiently large. Hence, we obtain a contradiction
whenever $N$ is large enough.

\textbf{Proof that $z_{n}$ is strictly positive for all $n$. }Suppose
then that $z_{n}$ is strictly increasing. We have by Equation (\ref{eq:OPTIMAL_W1})
that 
\begin{equation}
\delta_{1}<2\frac{\lambda}{R}w_{1}=2\frac{\lambda}{R}\left(U_{1}-\delta^{\left(1\right)}z_{1}\right).\label{eq:OPTIMAL_W1_SIMPLIFIED}
\end{equation}
By Equation (\ref{eq:EFFICIENCY_BOTTOM}),
\[
U_{1}-\delta^{\left(1\right)}z_{1}=\frac{z_{1}}{R\delta^{\left(1\right)}}
\]
or
\[
z_{1}=\frac{U_{1}}{\delta^{\left(1\right)}+\frac{1}{R\delta^{\left(1\right)}}}.
\]
Plugging into Equation (\ref{eq:OPTIMAL_W1_SIMPLIFIED}) and rearranging,
we have
\[
U_{1}>\frac{R^{2}\left(\delta^{\left(1\right)}\right)^{2}+R}{2\lambda}\delta_{1}.
\]
Therefore, $z_{1}>0$, which shows that $z_{n}>0$ for all $n$.

\textbf{Proof that $w_{n}$ is strictly positive for all $n$. }Condition
(\ref{eq:EFFICIENCY_BOTTOM}) requires that utility is efficiently
distributed between date 2 and date 3 conditional on the lowest type
$\delta^{\left(1\right)}$, which implies the positivity of $w_{1}$.
However, $w_{n}$ is decreasing, so we need to check that $w_{N}>0$
also. Here, we use Equation (\ref{eq:EULER}), evaluated at $n=N$.
If the first term is sufficiently small, then this equation also states
efficient distribution of utility between date 2 and date 3 conditional
on the highest type $\delta^{\left(N\right)}$, which would be enough
to conclude the result. We therefore need to show that the first term
in Equation (\ref{eq:EULER}) vanishes sufficiently fast with $N$.
It is enough to find a lower bound on $w_{N}=U_{1}+\sum_{i=2}^{N}\left(\delta^{\left(i\right)}-\delta^{\left(i-1\right)}\right)z_{i}-\delta^{\left(N\right)}z_{N}$
that does not depend on $N$ -- doing this constitutes the bulk of
the argument.

We will use from the previous step the following observations. First,
\begin{equation}
z_{1}>\frac{\frac{R^{2}\left(\delta^{\left(1\right)}\right)^{2}+R}{2\lambda}\delta_{1}}{\delta^{\left(1\right)}+\frac{1}{R\delta^{\left(1\right)}}},\label{eq:LOWER_BOUND_Z1}
\end{equation}
and therefore, second, 
\begin{equation}
w_{1}=U_{1}-\delta^{\left(1\right)}z_{1}=\frac{z_{1}}{R\delta^{\left(1\right)}}>\frac{R}{2\lambda}\delta_{1}.\label{eq:LOWER_BOUND_W1}
\end{equation}

We are considering the maximiser of the Lagrangean in Equation (\ref{eq:Lagrangean_Example})
and therefore the following perturbation must not increase the value
in that problem: for all $n$, add/subtract a constant (independent
of $n$) from the date-3 utility $z_{n}$, all else fixed. Thus $w_{1}$
is unchanged, and the implied values $\left(w_{n}\right)_{n=2}^{N}$
are unchanged. That this perturbation does not raise the value in
Equation (\ref{eq:Lagrangean_Example}) requires
\[
\delta_{1}\sum_{n=1}^{N}\frac{1}{N}\delta^{\left(n\right)}-\sum_{n=1}^{N}\frac{1}{N}\frac{\lambda}{R^{2}}\phi'\left(z_{n}\right)=0.
\]
Since $z_{n}$ is increasing, this implies
\[
\frac{\lambda}{R^{2}}\phi'\left(z_{1}\right)<\delta_{1}\sum_{n=1}^{N}\frac{1}{N}\delta^{\left(n\right)}
\]
 or
\[
z_{1}<\frac{R^{2}\delta_{1}}{2\lambda}\sum_{n=1}^{N}\frac{1}{N}\delta^{\left(n\right)}.
\]
This implies that 
\begin{equation}
w_{1}=U_{1}-\delta^{\left(1\right)}z_{1}=\frac{z_{1}}{R\delta^{\left(1\right)}}<\frac{R\delta_{1}}{2\lambda\delta^{\left(1\right)}}\sum_{n=1}^{N}\frac{1}{N}\delta^{\left(n\right)}.\label{eq:UPPER_BOUND_W1}
\end{equation}

From Equation (\ref{eq:FIRST_DIFFERENCE}) evaluated at $n=1$, we
have 
\begin{align*}
z_{2}-z_{1}= & \frac{\frac{\bar{\delta}-\underline{\delta}}{N-1}\frac{1}{N}\left(-\delta_{1}+\frac{\lambda}{R}4\left(U_{1}-\delta^{\left(1\right)}z_{1}\right)\right)}{2\frac{\lambda\frac{1}{N}}{R^{2}}+2\frac{\lambda\frac{1}{N}}{R}\delta^{\left(2\right)}\delta^{\left(1\right)}}\\
< & \frac{\frac{\bar{\delta}-\underline{\delta}}{N-1}\frac{1}{N}\left(-\delta_{1}+\frac{2\delta_{1}}{\delta^{\left(1\right)}}\sum_{n=1}^{N}\frac{1}{N}\delta^{\left(n\right)}\right)}{2\frac{\lambda\frac{1}{N}}{R^{2}}+2\frac{\lambda\frac{1}{N}}{R}\delta^{\left(2\right)}\delta^{\left(1\right)}},
\end{align*}
where the inequality follows from Equation (\ref{eq:UPPER_BOUND_W1}).
This can be used to provide an upper bound on $z_{n}-z_{n-1}$ also
for $n\in\left\{ 3,\dots,N\right\} $. In particular, from Equation
(\ref{eq:DIFFERENCES}), we have for $N$ sufficiently large that
$z_{n}-z_{n-1}$ has the same (positive) sign as $z_{n-1}-z_{n-2}$,
but is smaller. Therefore, for $n\geq2$, 
\begin{align*}
 & w_{n}-w_{n-1}\\
= & U_{1}+\sum_{i=2}^{n}\left(\delta^{\left(i\right)}-\delta^{\left(i-1\right)}\right)z_{i}-\delta^{\left(n\right)}z_{n}-\left(U_{1}+\sum_{i=2}^{n-1}\left(\delta^{\left(i\right)}-\delta^{\left(i-1\right)}\right)z_{i}-\delta^{\left(n-1\right)}z_{n-1}\right)\\
\geq & -\delta^{\left(n-1\right)}\left(z_{2}-z_{1}\right)\\
> & -\delta^{\left(n-1\right)}\frac{\frac{\bar{\delta}-\underline{\delta}}{N-1}\left(-\delta_{1}+\frac{2\delta_{1}}{\delta^{\left(1\right)}}\sum_{n=1}^{N}\frac{1}{N}\delta^{\left(n\right)}\right)}{2\frac{\lambda}{R^{2}}+2\frac{\lambda}{R}\delta^{\left(2\right)}\delta^{\left(1\right)}}.
\end{align*}
And so,
\begin{align*}
 & w_{N}-w_{1}\\
= & U_{1}+\sum_{i=2}^{N}\left(\delta^{\left(i\right)}-\delta^{\left(i-1\right)}\right)z_{i}-\delta^{\left(N\right)}z_{N}-\left(U_{1}-\delta^{\left(1\right)}z_{1}\right)\\
\geq & -\sum_{n=2}^{N}\delta^{\left(n-1\right)}\left(z_{2}-z_{1}\right)\\
> & -\delta^{\left(N\right)}\frac{\left(\bar{\delta}-\underline{\delta}\right)\left(-\delta_{1}+\frac{2\delta_{1}}{\delta^{\left(1\right)}}\sum_{n=1}^{N}\frac{1}{N}\delta^{\left(n\right)}\right)}{2\frac{\lambda}{R^{2}}+2\frac{\lambda}{R}\delta^{\left(2\right)}\delta^{\left(1\right)}}.
\end{align*}
Recalling that $\lambda$ remains bounded away from zero as $N\rightarrow\infty$,
this provides a lower bound on $w_{N}$ that is invariant to $N$.
In particular, we have 
\begin{align*}
w_{N}= & U_{1}+\sum_{i=2}^{N}\left(\delta^{\left(i\right)}-\delta^{\left(i-1\right)}\right)z_{i}-\delta^{\left(N\right)}z_{N}\\
> & U_{1}-\delta^{\left(1\right)}z_{1}-\delta^{\left(N\right)}\frac{\left(\bar{\delta}-\underline{\delta}\right)\left(-\delta_{1}+\frac{2\delta_{1}}{\delta^{\left(1\right)}}\sum_{n=1}^{N}\frac{1}{N}\delta^{\left(n\right)}\right)}{2\frac{\lambda}{R^{2}}+2\frac{\lambda}{R}\delta^{\left(2\right)}\delta^{\left(1\right)}}\\
> & \frac{R}{2\lambda}\delta_{1}-\delta^{\left(N\right)}\frac{\left(\bar{\delta}-\underline{\delta}\right)\left(-\delta_{1}+\frac{2\delta_{1}}{\delta^{\left(1\right)}}\sum_{n=1}^{N}\frac{1}{N}\delta^{\left(n\right)}\right)}{2\frac{\lambda}{R^{2}}+2\frac{\lambda}{R}\delta^{\left(2\right)}\delta^{\left(1\right)}}\\
= & \frac{1}{\lambda}\left(\frac{R}{2}\delta_{1}-\delta^{\left(N\right)}\frac{\left(\bar{\delta}-\underline{\delta}\right)\left(-\delta_{1}+\frac{2\delta_{1}}{\delta^{\left(1\right)}}\sum_{n=1}^{N}\frac{1}{N}\delta^{\left(n\right)}\right)}{\frac{2}{R^{2}}+\frac{2}{R}\delta^{\left(2\right)}\delta^{\left(1\right)}}\right).
\end{align*}
We will then use the implication that
\begin{align}
 & \delta_{1}-\frac{\lambda}{R}\phi^{\prime}\left(U_{1}+\sum_{i=2}^{N}\left(\delta^{\left(i\right)}-\delta^{\left(i-1\right)}\right)z_{i}-\delta^{\left(N\right)}z_{N}\right)\nonumber \\
< & \delta_{1}-\frac{2\lambda}{R}\frac{1}{\lambda}\left(\frac{R}{2}\delta_{1}-\delta^{\left(N\right)}\frac{\left(\bar{\delta}-\underline{\delta}\right)\left(-\delta_{1}+\frac{2\delta_{1}}{\delta^{\left(1\right)}}\sum_{n=1}^{N}\frac{1}{N}\delta^{\left(n\right)}\right)}{\frac{2}{R^{2}}+\frac{2}{R}\delta^{\left(2\right)}\delta^{\left(1\right)}}\right)\nonumber \\
= & \delta^{\left(N\right)}\frac{\left(\bar{\delta}-\underline{\delta}\right)\left(-\delta_{1}+\frac{2\delta_{1}}{\delta^{\left(1\right)}}\sum_{n=1}^{N}\frac{1}{N}\delta^{\left(n\right)}\right)}{\frac{2}{R^{2}}+\frac{2}{R}\delta^{\left(2\right)}\delta^{\left(1\right)}}.\label{eq:CONSEQUENCE_BOUND}
\end{align}

Now note that the necessary condition in Equation (\ref{eq:EULER}),
evaluated at $n=N$, is given by 
\[
0=\begin{array}{c}
\frac{\bar{\delta}-\underline{\delta}}{N-1}\frac{1}{N}\left(\delta_{1}-\frac{\lambda}{R}\phi^{\prime}\left(U_{1}+\sum_{i=2}^{N}\left(\delta^{\left(i\right)}-\delta^{\left(i-1\right)}\right)z_{i}-\delta^{\left(N\right)}z_{N}\right)\right)\\
+\frac{\lambda\frac{1}{N}\delta^{\left(N\right)}}{R}\phi^{\prime}\left(U_{1}+\sum_{i=2}^{N}\left(\delta^{\left(i\right)}-\delta^{\left(i-1\right)}\right)z_{i}-\delta^{\left(N\right)}z_{N}\right)-\frac{\lambda\frac{1}{N}}{R^{2}}\phi'\left(z_{N}\right).
\end{array}
\]
By Equation (\ref{eq:CONSEQUENCE_BOUND}), we then have
\begin{align*}
0< & \frac{\bar{\delta}-\underline{\delta}}{N-1}\frac{1}{N}\delta^{\left(N\right)}\frac{\left(\bar{\delta}-\underline{\delta}\right)\left(-\delta_{1}+\frac{2\delta_{1}}{\delta^{\left(1\right)}}\sum_{n=1}^{N}\frac{1}{N}\delta^{\left(n\right)}\right)}{\frac{2}{R^{2}}+\frac{2}{R}\delta^{\left(2\right)}\delta^{\left(1\right)}}\\
 & +\frac{\lambda\frac{1}{N}\delta^{\left(N\right)}}{R}\phi^{\prime}\left(U_{1}+\sum_{i=2}^{N}\left(\delta^{\left(i\right)}-\delta^{\left(i-1\right)}\right)z_{i}-\delta^{\left(N\right)}z_{N}\right)-\frac{\lambda\frac{1}{N}}{R^{2}}\phi'\left(z_{N}\right).
\end{align*}
This implies, using that $z_{n}$ is increasing,
\begin{align}
 & \frac{1}{R\delta^{\left(N\right)}}\phi'\left(z_{1}\right)-\frac{\left(\bar{\delta}-\underline{\delta}\right)^{2}\left(-\delta_{1}+\frac{2\delta_{1}}{\delta^{\left(1\right)}}\sum_{n=1}^{N}\frac{1}{N}\delta^{\left(n\right)}\right)}{\lambda\left(N-1\right)\left(\frac{2}{R}+2\delta^{\left(2\right)}\delta^{\left(1\right)}\right)}\nonumber \\
< & \phi^{\prime}\left(U_{1}+\sum_{i=2}^{N}\left(\delta^{\left(i\right)}-\delta^{\left(i-1\right)}\right)z_{i}-\delta^{\left(N\right)}z_{N}\right)\nonumber \\
=2 & w_{N.}\label{eq:KEY_INEQUALITY}
\end{align}
By Equation (\ref{eq:LOWER_BOUND_Z1}) left-hand side of the inequality
in Equation (\ref{eq:KEY_INEQUALITY}) is greater than
\[
\frac{2}{R\delta^{\left(N\right)}}\frac{\frac{R^{2}\left(\delta^{\left(1\right)}\right)^{2}+R}{2\lambda}\delta_{1}}{\delta^{\left(1\right)}+\frac{1}{R\delta^{\left(1\right)}}}-\frac{\left(\bar{\delta}-\underline{\delta}\right)^{2}\left(-\delta_{1}+\frac{2\delta_{1}}{\delta^{\left(1\right)}}\sum_{n=1}^{N}\frac{1}{N}\delta^{\left(n\right)}\right)}{\lambda\left(N-1\right)\left(\frac{2}{R}+2\delta^{\left(2\right)}\delta^{\left(1\right)}\right)}
\]
which becomes strictly positive as $N\rightarrow\infty$, since the
second term vanishes while the first does not. This shows that $w_{N}$
is indeed strictly positive for large $N$ as required.

\end{proof}

\bigskip{}

\begin{proof}[Proof of Proposition~\ref{prop:BACKLOADING}]

The arguments in the main text imply that, for $n\geq2$, we have
$\delta^{\left(n\right)}R\phi^{\prime}\left(w_{n}\right)<\phi^{\prime}\left(z_{n}\right)$.
Consider reducing utility at date 3 by $\varepsilon$, and increasing
it at date 2 by $\delta^{\left(n\right)}\varepsilon$, thus leaving
discounted utility unchanged. Then this decreases the NPV of firm
costs in case $\delta_{2}=\delta^{\left(n\right)}$ by
\[
-\delta^{\left(n\right)}\phi^{\prime}\left(w_{n}\right)\varepsilon+\frac{\phi^{\prime}\left(z_{n}\right)}{R}\varepsilon+o\left(\varepsilon\right)
\]
which is strictly positive for $\varepsilon$ small enough. The additional
profit can be used to further increase the date-2 utility, i.e. a
profit-preserving change that increases the agent's payoff.

\end{proof}

\bigskip{}

\begin{proof}[Proof of Proposition~\ref{prop:CONSTANT_PERTURBATION}]

Consider equilibrium utility, and consider any $t\in\left\{ 2,\dots,T-1\right\} $,
any history $\left(\delta_{1},H_{2}^{t-1}\right)$ (the remaining
case that considers shifting utility between date 1 and date 2 is
analogous and omitted). Consider reducing date $t$ utility at this
history by a constant $\eta\in\mathbb{R}$ independent of $\delta_{t}$
and raising utility at $t+1$ by $\nu\in\mathbb{R}$. The change in
expected profits conditional on $\left(\delta_{1},H_{2}^{t-1}\right)$
is, to the first order
\[
\sum_{n^{\prime}=1}^{N}p_{n}\frac{\phi^{\prime}\left(v_{t}^{E,T}\left(\delta_{1},H_{2}^{t-1},\delta^{\left(n^{\prime}\right)}\right)\right)}{R^{t-1}}\eta-\sum_{n^{\prime\prime}=1}^{N}\sum_{m=1}^{N}p_{n^{\prime\prime}}p_{m}\frac{\phi^{\prime}\left(v_{t+1}^{E,T}\left(\delta_{1},H_{2}^{t-1},\delta^{\left(n^{\prime\prime}\right)},\delta^{\left(m\right)}\right)\right)}{R^{t}}\nu.
\]
In particular, for the change to leave firm profits unchanged requires
\[
\nu=\frac{\sum_{n=1}^{N}p_{n}\frac{\phi^{\prime}\left(v_{t}^{E,T}\left(\delta_{1},H_{2}^{t-1},\delta^{\left(n\right)}\right)\right)}{R^{t-1}}}{\sum_{n=1}^{N}\sum_{m=1}^{N}p_{n}p_{m}\frac{\phi^{\prime}\left(v_{t+1}^{E,T}\left(\delta_{1},H_{2}^{t-1},\delta^{\left(n\right)},\delta^{\left(m\right)}\right)\right)}{R^{t}}}\eta+o\left(\eta\right)
\]
where $o\left(\eta\right)$ represents terms that vanish faster than
$\eta$ as $\eta\rightarrow0$.

Note that because utility changes by a constant, independent of $\delta_{t}$
and $\delta_{t+1}$, the change leaves incentive constraints unaffected.
The date-$t$ value to the agent of the change is
\[
-\eta+\sum_{n=1}^{N}q_{n}\delta^{\left(n\right)}\nu=\eta\left[-1+\sum_{i=1}^{N}q_{i}\delta^{\left(i\right)}\frac{\sum_{n=1}^{N}p_{n}\frac{\phi^{\prime}\left(v_{t}^{E,T}\left(\delta_{1},H_{2}^{t-1},\delta^{\left(n\right)}\right)\right)}{R^{t-1}}}{\sum_{n=1}^{N}\sum_{m=1}^{N}p_{n}p_{m}\frac{\phi^{\prime}\left(v_{t+1}^{E,T}\left(\delta_{1},H_{2}^{t-1},\delta^{\left(n\right)},\delta^{\left(m\right)}\right)\right)}{R^{t}}}\right]+o\left(\eta\right).
\]
Since the change leaves both firm profits and incentive constraints
unaffected, it must not raise the agent's expected payoff. Therefore,
a necessary condition for optimality in Problem IV is 
\begin{align*}
\sum_{n=1}^{N}p_{n}\sum_{m=1}^{N}p_{m}\phi^{\prime}\left(v_{t+1}^{E,T}\left(\delta_{1},H_{2}^{t-1},\delta^{\left(n\right)},\delta^{\left(m\right)}\right)\right) & =R\sum_{i=1}^{N}q_{i}\delta^{\left(i\right)}\sum_{n=1}^{N}p_{n}\phi^{\prime}\left(v_{t}^{E,T}\left(\delta_{1},H_{2}^{t-1},\delta^{\left(n\right)}\right)\right).
\end{align*}
This establishes the first equality in the proposition. The equalities
relating to the efficient policy follow from analogous reasoning.

\end{proof}

\bigskip{}

\begin{proof}[Proof of Corollary~\ref{Cor:LOG-CASE}]

The corollary follows by taking expectations of the equalities in
the proposition over histories $\tilde{H}_{2}^{t-1}$, and using that
agent beliefs first-order stochastically dominate the firm beliefs
so that $\mathbb{E}_{A}\left[\tilde{\delta}_{t}\right]\geq\mathbb{E}_{F}\left[\tilde{\delta}_{t}\right]$,
with the distributions differing and hence a strict inequality if
the agent is strictly more optimistic than the firm.

\end{proof}
\end{document}